\documentclass[paper=a4, fontsize=11pt]{scrartcl}
\usepackage[T1]{fontenc}
\usepackage{fourier}
\usepackage[english]{babel}															% English 
\usepackage{amssymb}
\usepackage[protrusion=true,expansion=true]{microtype}	
\usepackage{amsmath,amsfonts,amsthm} % Math packages
\usepackage[pdftex]{graphicx}	
\usepackage{url}
\usepackage[title,titletoc,page]{appendix} 
\usepackage{listings} 
\usepackage{color}
\usepackage{comment}
\usepackage{caption}
\usepackage{subcaption}
\usepackage[ruled,vlined]{algorithm2e}
\usepackage{sectsty}
\allsectionsfont{\centering \normalfont\scshape}

\usepackage{fancyhdr}
\numberwithin{equation}{section}		% Equationnumbering: section.eq#
\numberwithin{figure}{section}			% Figurenumbering: section.fig#
\numberwithin{table}{section}				% Tablenumbering: section.tab#

\newtheorem{thm}{Theorem}[section]
\newtheorem{defi}[thm]{Definition}
\newtheorem{lem}[thm]{Lemma}
\newtheorem{cor}[thm]{Corollary}
\newtheorem*{remark}{Remark}

\newcommand{\what}{\widehat}
\newcommand{\leqs}{\leqslant}
\newcommand{\geqs}{\geqslant}
\newcommand{\babs}[1]{\left|{#1}\right|}
\newcommand{\bnorm}[1]{\left|\left|{#1}\right|\right|}

\newcommand{\vect}[1]{\boldsymbol{\mathbf{#1}}}

\title{A mathematical theory of super-resolution and two-point resolution 
	\thanks{\footnotesize This work was supported in part by the Swiss National Science Foundation grant number
		200021--200307.}}
\author{Ping Liu\thanks{\footnotesize Department of Mathematics, ETH Z\"urich, R\"amistrasse 101, CH-8092 Z\"urich, Switzerland (ping.liu@sam.math.ethz.ch, habib.ammari@math.ethz.ch).}  \and Habib Ammari\footnotemark[2]}

\date{}
\begin{document}
\maketitle
\begin{abstract}
This paper focuses on the fundamental aspects of super-resolution, particularly addressing the stability of super-resolution and the estimation of two-point resolution. Our first major contribution is the introduction of two location-amplitude identities that characterize the relationships between locations and amplitudes of true and recovered sources in the one-dimensional super-resolution problem. These identities facilitate direct derivations of the super-resolution capabilities for recovering the number, location, and amplitude of sources, significantly advancing existing estimations to levels of practical relevance. As a natural extension, we establish the stability of a specific $l_0$ minimization algorithm in the super-resolution problem.

The second crucial contribution of this paper is the theoretical proof of a two-point resolution limit in multi-dimensional spaces.  The resolution limit is expressed as:
    \[
     \mathcal R = \frac{4\arcsin \left(\left(\frac{\sigma}{m_{\min}}\right)^{\frac{1}{2}} \right)}{\Omega}
    \]
    for $\frac{\sigma}{m_{\min}}\leqs\frac{1}{2}$, where $\frac{\sigma}{m_{\min}}$ represents the inverse of the signal-to-noise ratio ($\mathrm{SNR}$) and $\Omega$ is the cutoff frequency. It also demonstrates that for resolving two point sources, the resolution can exceed the Rayleigh limit $\frac{\pi}{\Omega}$ when the signal-to-noise ratio (SNR) exceeds $2$. Moreover, we find a tractable algorithm that achieves the resolution $\mathcal{R}$ when distinguishing two sources.
\end{abstract}

\vspace{0.5cm}

\vspace{0.2cm}
\vspace{0.5cm}

\section{Introduction}
Since the first report of the use of microscopes for observation in the 17th century, optical microscopes have played a central role in helping to untangle complex biological mysteries. Numerous scientific advancements and manufacturing innovations over the past three centuries have led to advanced optical microscope designs with significantly improved image quality. However, due to the physical nature of wave propagation and diffraction, there is a fundamental diffraction barrier in optical imaging systems which is called the \emph{resolution limit}. This resolution limit is one of the most important characteristics of an imaging system. In 19th century, Rayleigh \cite{rayleigh1879xxxi} gave a well-known criterion for determining the resolution limit (\emph{Rayleigh's diffraction limit}) for distinguishing two point sources, which is extensively used in optical microscopes for analyzing the resolution. The problem to resolve point sources separated below the Rayleigh diffraction limit is then called super-resolution and is commonly known to be very challenging for single snapshot. However, Rayleigh's criterion is based on intuitive notions and is more applicable to observations with the human eye. It also neglects the effect of the noise in the measurements and the aberrations in the modeling. On the other hand, due to the rapid advancement of technologies, modern imaging data is generally captured using intricate imaging techniques and sensitive cameras, and may also be subject to analysis by complex processing algorithms. Thus Rayleigh's deterministic resolution criterion is not well adapted to current quantitative imaging techniques, necessitating new and more rigorous definitions of resolution limit with respect to the noise, model and imaging methods \cite{ram2006beyond}. 

Our previous works \cite{liu2021mathematicaloned, liu2021theorylse, liu2021mathematicalhighd, liu2023improved} have achieved certain success in this respect and enable us to understand the performance of some super-resolution algorithms. Nevertheless, the derived estimates are still lacking enough guiding significance in practice on the possibility of super-resolution.  

In this paper, we introduce new and direct insights into the stability of super-resolution problems and significantly enhance many estimates to have practical significance. Our findings reveal several new facts; for example, we theoretically demonstrate that super-resolution from a single snapshot is indeed feasible in practice.

%In this paper, we present new and direct understandings for the stability of super-resolution problems and substantially improve many estimates to have practical significance.  Many new facts are disclosed by our results, for instance, it is theoretically demonstrated here that the super-resolution from single snapshot is actually possible. 

\subsection{Resolution limit, super-resolution and diffraction limit}

\subsubsection{Classical criteria for the resolution limit}
The resolution of an optical microscope is usually defined as the shortest distance between two points on a specimen that can still be distinguished by the observer or camera system as separate entities, but it is in some sense ambiguous. In fact, it does not explicitly require that just the source number should be detected or the source locations should be stably reconstructed. From the mathematical perspective, recovering the source number and stably reconstructing the source locations are actually two different tasks \cite{liu2021theorylse} demanding different separation distances. 

Historically, the resolution of optical microscopy focuses mainly on correctly detecting the number of sources, rather than on stably recovering the locations. This can be seen from the below discussions for the classical and semi-classical results on the resolution.

In the 18th and 19th centuries, many classical criteria were proposed to determine the resolution limit. For example, Rayleigh thought that two point sources observed are regarded as just resolved when the principal diffraction maximum of one Airy disk coincides with the first minimum of the other, as is shown by Figure \ref{fig:rayleilimit}. Since the separation of sources is still relatively large, not only the source number can be detected, but also the source locations can be stably recovered. 

\begin{figure}[!h]
	\centering
\includegraphics[width=2in,height=1.2in]{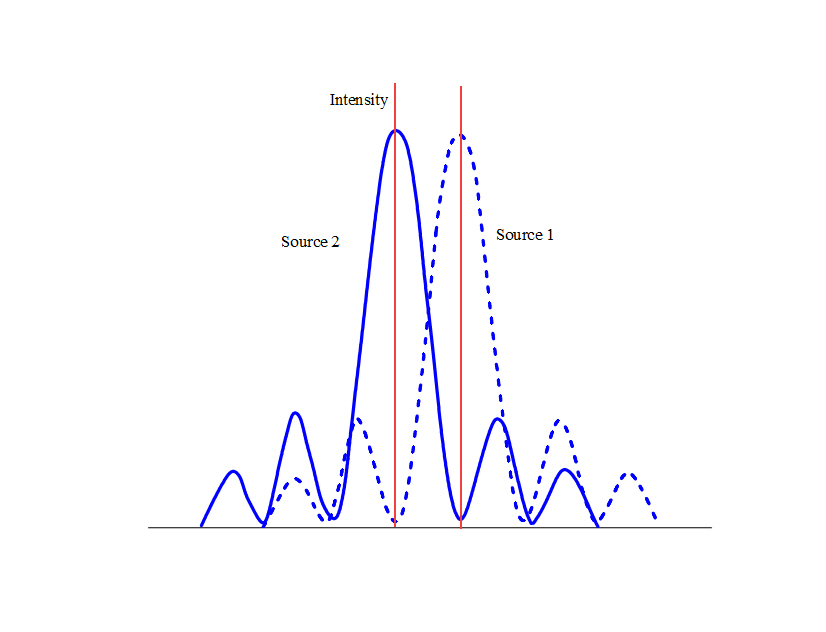}
\includegraphics[width=2.5in]{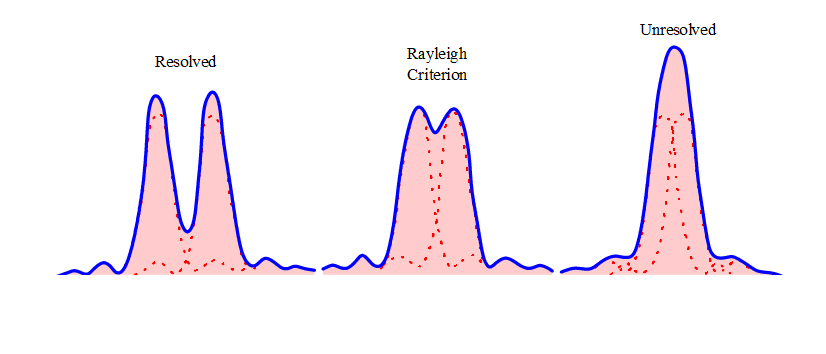}
\caption{Rayleigh's criterion and Rayleigh's diffraction limit.}
\label{fig:rayleilimit}
\end{figure}

However, Rayleigh's choice of resolution limit is based on the presumed resolving ability of human visual system, which at first glance seems arbitrary. In fact, Rayleigh said about his criterion that 
\begin{quotation}{``This rule is convenient on account of its simplicity and it is sufficiently accurate in view of the necessary uncertainty as to what exactly is meant by resolution."}
\end{quotation}
As is shown in Figure \ref{fig:rayleilimit}, the Rayleigh diffraction limit results in an $\sim 20\%$ dip in intensity between the two peaks of Airy disks \cite{demmerle2015assessing}. Schuster pointed out in 1904 \cite{schuster1904introduction} that the dip in intensity necessary to indicate resolution is a physiological phenomenon and there are other forms of spectroscopic investigation besides that of eye observation. Many alternative criteria were proposed by other physicists as illustrated in Figure \ref{fig:resolutionlimits}. 
\begin{figure}[!h]
	\centering	\includegraphics[width=4.8in, height = 1.8in]{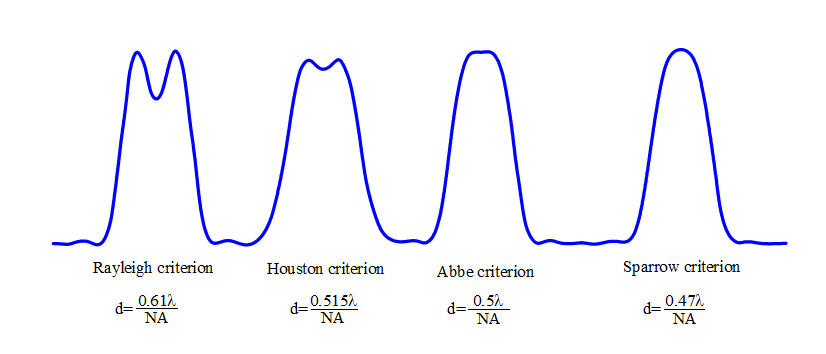}
	\caption{Different resolution limits.}
	\label{fig:resolutionlimits}
\end{figure}

A more mathematically rigorous criterion was proposed by Sparrow \cite{sparrow1916spectroscopic}, who advocates that the resolution limit should be the distance between two point sources where the images no longer have a dip between the central peaks of each Airy disk (as illustrated by Figure \ref{fig:resolutionlimits}). Based on Sparrow's criterion, the two point sources are so close that the source locations may not be stably resolved although the source number is detected. Indeed, the Sparrow resolution limit is less relevant with practical use \cite{chen2020algorithmic, demmerle2015assessing} because it is very signal-to-noise dependent and has no easy comparison to a readily measured value in real images. Therefore, these classical resolution criteria focus on the smallest distance between two sources above which we can be sure that there are two sources, regardless of whether their locations can be resolved stably or not.

The classical resolution criteria mentioned above deal with calculated images that are described by a known and exact mathematical model of the intensity distribution. However, if one has perfect access to the intensity profile of the diffraction image of two point sources, one could locate the exact source despite the diffraction. There would be no resolution limit for the reconstruction. This simple fact has been noticed by many researchers in the field \cite{di1955resolving, den1997resolution, chen2020algorithmic}. On the other hand, an imaging system constructed without any aberration or irregularity is not practical, because the shape of the point-spread function is never known exactly and the measurement noise is inevitable
\cite{ronchi1961resolving, den1997resolution}. Therefore, a rigorous and practically meaningful resolution limit could only be set when taking into account the aberrations and measurement noises \cite{ronchi1961resolving, goodman2015statistical}. In this setting, the images (detected by detectors in practice) were categorized as detected images by Ronchi  \cite{ronchi1961resolving} and their resolution was advocated to be more important to investigate than the resolution defined by those classical criteria.  Inspired by this, many researchers have analyzed the two-point resolution from the perspective of statistical inference \cite{helstrom1964detection, helstrom1969detection, lucy1992statistical, lucy1992resolution, goodman2015statistical, shahram2004imaging, shahram2004statistical, shahram2005resolvability}. 
For instance, in \cite{helstrom1969detection} Helstrom has shown that the resolution of two identical objects depends on deciding whether they  are both present in the field of view or only one of them is there, and their resolvability is measured by the probability of making this decision correctly.  
In all the papers mentioned above, the authors have derived quasi-explicit formulas or estimations for the minimum SNR that is required to discriminate two point sources or for the possibility of a correct decision. Although the resolutions (or the requirement) in this respect were thoroughly explored in these works which spanned the course of several decades, these results are rarely (even never) utilized in practical applications due to their complexity.

\subsubsection{Concept of super-resolution}
We next introduce the concept of super-resolution. Super-resolution microscopy is a series of techniques in optical microscopy that allow such images to have resolutions higher than those imposed by the diffraction limit (Rayleigh resolution limit). Due to the breakthrough in super-resolution techniques in recent years, this topic becomes increasingly popular in various fields and the concept of super-resolution becomes very general.  Some literature claims super-resolution although theoretically the sources should be separated by a distance above the Rayleigh limit. Bounds on the recovery of the amplitudes (or intensities) of the sources have been derived. Nevertheless, the original concept of super-resolution actually focuses mainly on both detecting the source number and recovering the source locations.

To the best of our knowledge, there is no unique and mathematically rigorous definition of super-resolution. On the other hand, as we have said, the number detection and location recovery are two inherently different \cite{liu2021theorylse} but important tasks in the super-resolution, thus we consider two different super-resolutions in the current paper. One is achieving resolution better than the Rayleigh diffraction limit in detecting the correct source number and is named "super-resolution in number detection". The other is achieving resolution better than the Rayleigh diffraction limit in stably recovering the source locations and is named "super-resolution in location recovery".

\subsection{Previous mathematical contributions}

Before introducing the mathematical contributions of this work, let us first introduce the mathematical model for the imaging problem in $k$-dimensional space. We consider the collection of point sources as a discrete measure $\mu=\sum_{j=1}^{n}a_{j}\delta_{\vect y_j}$, where $\vect y_j \in \mathbb R^k,j=1,\cdots,n$, represent the location of the point sources and the $a_j$'s their amplitudes. The imaging problem is to recover the sources $\mu$ from its noisy Fourier data, 
\begin{equation}\label{equ:highdmodelsetting1}
	\mathbf Y(\vect{\omega}) = \mathcal F[\mu] (\vect{\omega}) + \mathbf W(\vect{\omega})= \sum_{j=1}^{n}a_j e^{i \vect{y}_j\cdot \vect{\omega}} + \mathbf W(\vect{\omega}), \ \vect \omega \in \mathbb R^k, \ ||\vect{\omega}||_2\leqs \Omega, 
\end{equation}
where $\mathcal F[\cdot]$ denotes the Fourier transform in the $k$-dimensional space, $\Omega$ is the cut-off frequency, and $\mathbf W$ represents the total effect of noise and aberrations. We assume that 
\begin{equation}\label{equ:noiselevel1}
 \left|\mathbf W(\vect{\omega})\right|< \sigma, \ \vect \omega \in \mathbb R^k, \ ||\vect{\omega}||_2\leqs \Omega,
\end{equation}
with $\sigma$ being the noise level. We denote respectively the magnitude of the signal and the minimum separation distance between sources by
\begin{equation}\label{equ:intendisset}
m_{\min}=\min_{j=1,\cdots, n}|a_j|,
\quad 	d_{\min}=\min_{j\neq p}\bnorm{\vect y_j-\vect y_p}_2.
\end{equation}
As most of our analyses are on a one-dimensional space, throughout the paper, we use $y_j, \omega$ to denote the one-dimensional source locations and frequencies and reserve $\vect y_j, \vect \omega$ for the problem in spaces of general dimensionality. Especially, we denote the one-dimensional sources as $\mu=\sum_{j=1}^{n}a_{j}\delta_{y_j}$ and the noisy measurement as 
\[
\mathbf Y(\omega) = \sum_{j=1}^{n}a_j e^{i y_j\omega} + \mathbf W(\omega),\quad  \omega \in [-\Omega, \Omega]. 
\]

Model (\ref{equ:highdmodelsetting1}) is commonly used in the field of applied mathematics for theoretically analyzing the imaging problem \cite{donoho1992superresolution, candes2014towards, batenkov2019super}. It is directly the model in the frequency domain for the imaging modalities with $\mathrm{sinc}(||\vect x ||_2)$ being the point spread function \cite{den1997resolution}. Its discrete form is also a standard model in array signal processing. On the other hand, even for imaging with general point spread functions or optical transfer functions, some of the imaging enhancements such as inverse filtering method \cite{frieden1975image} will modify the measurements in the frequency domain to (\ref{equ:highdmodelsetting1}). These ensure the generality of the model (\ref{equ:highdmodelsetting1}) in the fields of imaging and array signal processing.

Based on Rayleigh's criterion, the corresponding resolution limit for imaging with the point spread function $\mathrm{sinc}(||\vect x||_2)^2$ is $\frac{\pi}{\Omega}$.  On the other hand, it was shown by many mathematical studies that $\frac{\pi}{\Omega}$ is also the critical limit for the imaging model \eqref{equ:highdmodelsetting1}. To be more specific,  in \cite{donoho1992superresolution} Donoho demonstrated that for sources on grid points spacing by $\Delta\geqs \frac{\pi}{\Omega}$, the stable recovery is possible from (\ref{equ:highdmodelsetting1}) in dimension one, but the stability becomes much worse in the case when $\Delta< \frac{\pi}{\Omega}$. Therefore, in the same way as \cite{donoho1992superresolution}, we regard $\frac{\pi}{\Omega}$ as the Rayleigh limit in this paper, and super-resolution refers to obtaining a better resolution than $\frac{\pi}{\Omega}$.

For the mathematical theory of super-resolving $n$ point sources or infinity point sources, to the best of our knowledge, the first result was derived by Donoho in 1992 \cite{donoho1992superresolution}. He developed a theory from the optimal recovery point of view to explain the possibility and difficulties of super-resolution via sparsity constraint. He considered discrete measures $\mu$ supported on a lattice $\{k\Delta\}_{k=-\infty}^{\infty}$ and regularized by a so-called “Rayleigh index” $R$. The available measurement is the noisy Fourier data of $\mu$ like model (\ref{equ:highdmodelsetting1}) in a one-dimensional space.  He showed that the minimax error $E^*$ for the amplitude recovery with noise level $\sigma$ was bounded as $$\beta_1(R,\Omega)\left(\frac{1}{\Delta}\right)^{2R-1}\sigma\leqs E^* \leqs  \beta_2(R,\Omega) \left(\frac{1}{\Delta}\right)^{4R+1}\sigma$$ for certain small $\Delta$. His results emphasize the importance of sparsity in the super-resolution. In recent years, due to the impressive development of super-resolution modalities in biological imaging \cite{hell1994breaking, westphal2008video, hess2006ultra, betzig2006imaging, rust2006sub} and super-resolution algorithms in applied mathematics \cite{candes2014towards, duval2015exact, poon2019, tang2013compressed, tang2014near, morgenshtern2016super, morgenshtern2020super, denoyelle2017support}, the inherent super-resolving capacity of the imaging problem becomes increasingly popular and the one-dimensional case was well-studied. In \cite{demanet2015recoverability}, the authors considered resolving the amplitudes of $n$-sparse point sources supported on a grid and improved the results of Donoho. Concretely, they showed that the scaling of the noise level for the minimax error should be $\mathrm{SRF}^{2n-1}$, where $\mathrm{SRF}:= \frac{1}{\Delta\Omega}$ is the super-resolution factor. Similar results for multi-clumps cases were also derived in \cite{batenkov2020conditioning, li2018stable}. Recently in  \cite{batenkov2019super}, the authors derived sharp minimax errors for the location and the amplitude recovery of off-the-grid sources. They showed that for $\sigma \lesssim (\mathrm{SRF})^{-2p+1}$, where $p$ is the number of nodes that form a cluster of certain type, the minimax error rate for reconstruction of the clustered nodes is of the order $(\mathrm{SRF})^{2p-2} \frac{\sigma}{\Omega}$, while for recovering the corresponding amplitudes the rate is of the order $(\mathrm{SRF})^{2p-1}\sigma$. Moreover, the corresponding minimax rates for the recovery of the non-clustered nodes and amplitudes are $\frac{\sigma}{\Omega}$ and $\sigma$ respectively. This was generalized to the case of resolving positive sources in \cite{liu2023super} recently. We also refer the readers to \cite{moitra2015super,chen2020algorithmic} for understanding the resolution limit from the perceptive of sample complexity and to  \cite{tang2015resolution, da2020stable} for the resolving limit of some algorithms.

On the other hand, in order to characterize the exact resolution rather than the minimax error in recovering multiple point sources, in the earlier works \cite{liu2021mathematicaloned, liu2021theorylse, liu2021mathematicalhighd, liu2023improved, liu2022rslpositive} we have defined the so-called "computational resolution limits", which characterize the minimum required distance between point sources so that their number and locations can be stably resolved under certain noise level. By developing a nonlinear approximation theory in so-called Vandermonde spaces, we have derived sharp bounds for computational resolution limits in the one-dimensional super-resolution problem (\ref{equ:highdmodelsetting1}). In particular, we have shown in \cite{liu2021theorylse} that the computational resolution limits for the number and location recoveries should be bounded above by respectively $\frac{4.4e\pi}{\Omega}\left(\frac{\sigma}{m_{\min }}\right)^{\frac{1}{2 n-2}}$ and $\frac{5.88 e\pi}{\Omega}\left(\frac{\sigma}{m_{\min }}\right)^{\frac{1}{2 n-1}}$, where the noise level $\sigma$ and signal magnitude $m_{\min}$ are defined as in (\ref{equ:noiselevel1}) and (\ref{equ:intendisset}), respectively. In the present paper, we substantially improve these estimates to have practical significance.

\subsection{Our main contributions in this paper}
 Our first contribution is two location-amplitude identities characterizing the relations between locations and amplitudes of true and recovered sources in the one-dimensional super-resolution problem. These identities allow us to directly derive the super-resolution capability for number, location, and amplitude recovery in the super-resolution problem and improve state-of-the-art estimations to an unprecedented level to have practical significance. Although these nonlinear inverse problems are known to be very challenging, we now have a clear and simple picture of all of them, which allows us to solve them in a unified way in just a few pages. To be more specific, we prove that, under model (\ref{equ:highdmodelsetting1}) in dimension one, it is definitely possible to detect the correct source number when the sources are separated by
\[
d_{\min}\geqs \frac{2e\pi }{\Omega }\Big(\frac{\sigma}{m_{\min}}\Big)^{\frac{1}{2n-2}},
\]
where $\frac{\sigma}{m_{\min}}$ represents the inverse of the signal-to-noise ratio ($\mathrm{SNR}$). This substantially improves the estimate in \cite{liu2021theorylse} and indicates that super-resolution in detecting correct source number (i.e., surpassing the diffraction limit $\frac{\pi}{\Omega}$) is definitely possible when $\frac{m_{\min}}{\sigma}\geqs (2e)^{2n-2}$.  Moreover, for the case when resolving two sources, the requirement for the separation is improved to 
\[
d_{\min}\geqs \frac{2\arcsin\left(2\left(\frac{\sigma}{m_{\min}}\right)^{\frac{1}{2}}\right)}{\Omega},
\]
indicating that surpassing the Rayleigh limit in distinguishing two sources is possible when $\mathrm{SNR}>4$. This is the first time where it is demonstrated theoretically that super-resolution ('in number detection') is actually possible in practice. For the stable location recovery, the estimate is improved to 
\[
d_{\min}\geqs \frac{2.36e\pi}{\Omega }\Big(\frac{\sigma}{m_{\min}}\Big)^{\frac{1}{2n-1}}
\]
as compared to the previous result in \cite{liu2021theorylse}, indicating that the location recovery is stable when $\frac{m_{\min}}{\sigma}\geqs (2.36 e)^{2n-1}$. Moreover, under the same separation condition, we also obtain stability results for amplitude recovery and a certain $l_0$ minimization algorithm in the super-resolution problem. These results provide us with a quantitative understanding of the super-resolution of multiple sources. Since our method is rather direct, it is very hard to substantially improve the estimates now and we even roughly know to what extent the constant factor in the estimates can be improved.

Our second crucial result is the theoretical proof of a two-point resolution limit in multi-dimensional spaces under only an assumption on the noise level.  It is given by
\begin{equation}\label{equ:diffractionlimit1}
	\frac{4\arcsin \left(\left(\frac{\sigma}{m_{\min}}\right)^{\frac{1}{2}} \right)}{\Omega}
\end{equation}
for $\frac{\sigma}{m_{\min}}\leqs\frac{1}{2}$. In the case when $\frac{\sigma}{m_{\min}}>\frac{1}{2}$, there is no super-resolution under certain circumstances. Our results show that, for resolving any two point sources, the resolution can exceed the Rayleigh limit when $\mathrm{SNR}>2$. When $\mathrm{SNR}>4$, one can achieve $1.5$ times improvement of the Rayleigh limit. This finding indicates that obtaining a resolution far better than the Rayleigh limit in practical imaging and direction-of-arrival problems is possible with refined sensors. As a comparison, former works for the two-point resolution \cite{helstrom1964detection, helstrom1969detection, lucy1992statistical, lucy1992resolution, goodman2015statistical, shahram2004imaging, shahram2004statistical, shahram2005resolvability} consider the model in the physical domain with the imaging process or the noise being random. In addition, the derived estimate is relatively complicated as the model considered is more complex than (\ref{equ:highdmodelsetting1}). For example, the $\mathrm{SNR}$ governing object detectability in \cite{helstrom1969detection} is given by 
\[
\mathrm{SNR} = (E/N)(TW)^{-\frac{1}{2}} \left|\int_A \phi_s(\mathbf{r}, \mathbf{r}) d^2 \mathbf{r}\right|^{-1}\times\left[\int_A \int_A\left|\phi_s\left(\mathbf{r}_1, \mathbf{r}_2\right)\right|^2 d^2 \mathbf{r}_1 d^2 \mathbf{r}_2\right]^{\frac{1}{2}},
\]
where $\phi_s(\cdot)$ denotes the autocovariance function of the field in the aperture plane and $E, N, T, W, A$ represent other factors.

The estimate of two-point resolution can be directly extended to the following more general setting: 
\begin{equation}\label{equ:highdgeneralmodel0}
	\mathbf Y(\vect{\omega}) = \chi(\vect \omega)\mathcal F[\mu] (\vect{\omega}) + \mathbf W(\vect{\omega})= \sum_{j=1}^{n}a_j \chi(\vect \omega) e^{i \vect{y}_j\cdot \vect{\omega}} + \mathbf W(\vect{\omega}), \ \vect \omega \in \mathbb R^k, \ ||\vect{\omega}||_2\leqs \Omega,
\end{equation}
where $\chi(\vect \omega)=0$ or $1$, $\chi(\vect 0)=1$ and $\chi(\vect \omega)=1, ||\vect \omega||_2 =\Omega$. This enables the application of our results to line spectral estimations and directional-of-arrival in signal processing. Moreover, our findings can be applied to imaging systems with general optical transfer functions. A new fact revealed in this paper is that the two-point resolution is actually determined by the boundary points of the transfer function and is not that dependent on the interior frequency information. Also, as revealed in Section \ref{section:optimalalgorithm1}, the measurements at $\vect \omega =\vect 0$ and $\bnorm{\vect \omega}=\Omega$ are already enough for the algorithm which provably achieves the resolution limit.

In the last part of the paper, we find an algorithm that achieves the optimal resolution when distinguishing two sources and conducts many numerical experiments to manifest its optimal performance and phase transition.  Although the noise and the aberration are inevitable and the point source is not an exact delta point, our results still indicate that super-resolving two sources in practice is possible for general imaging modalities, due to the excellent noise tolerance.  We plan to examine the practical feasibility of our method in the near future. 

To summarize, by this paper we have shed light on understanding quantitatively when super-resolution is definitely possible and when it is not.  It has been disclosed by our results that super-resolution when distinguishing two sources is far more possible than what was commonly recognized. %By this work, we hope to inspire a start of a new period where examining the resolution based on the signal-to-noise ratio becomes a feasible method in the field of imaging, which is also the hope of many physicists and opticists \cite{ronchi1961resolving}. 

\subsection{Organization of the paper}
The paper is organized in the following way. In Section 2, we present the theory of location-amplitude identities. In Section 3, we derive stability results for recovering the number, locations, and amplitudes of sources in the one-dimensional super-resolution problem. In Section 4, we derive the exact formula of the two-point resolution limit and, in Section 5, we devise algorithms achieving exactly the optimal resolution in distinguishing images from one and two sources. The Appendix consists of some useful inequalities.

\section{Location-amplitude identities}
In this section, we intend to derive two location-amplitude identities that characterize the relations between source locations and amplitudes in the one-dimensional super-resolution problem. We start from the following elementary model in dimension one: 
\begin{equation} \label{equ:modelsetting0} 
\mathcal F[\what \mu](\omega) = \mathcal F[\mu](\omega)+\vect w(\omega), \quad \omega \in [0, \Omega],
\end{equation}
where $\what \mu, \mu$ are discrete measures, $\mathcal F[f]=\int_{\mathbb R} e^{iy\omega} f(y)dy$ denotes the Fourier transform, and $\vect w(\omega) = \mathcal F[\what \mu](\omega) - \mathcal F[\mu](\omega)$. To be more specific, we set $\mu=\sum_{j=1}^na_j\delta_{y_j}$ and $\what \mu= \sum_{j=1}^d\what a_j\delta_{\what y_j}$ with $a_j, \what a_j$ being the source amplitudes and $y_j, \what y_j$ the source locations. 

\subsection{Statement of the identities}
Based on the above model, we have the following location-amplitude identities. 
\begin{thm} \label{thm:locaintenrelation1} [Location-amplitude identities] Consider the model
\[
\mathcal F[\what \mu](\omega) = \mathcal F[\mu](\omega)+\vect w(\omega), \quad \omega \in [0, \Omega],
\]
where $\what \mu = \sum_{j=1}^d\what a_j\delta_{\what y_j}$ and $\mu=\sum_{j=1}^na_j\delta_{y_j}$. For any fixed $y_t$ and $\what y_{t'}$, define the set $S_t$ containing all $y_j$'s and $\what y_j$'s except $y_t, \what y_{t'}$ that
\[
S_t:=\left\{y_1,\cdots, y_{t-1}, y_{t+1}, \cdots ,y_n, \what y_1,\cdots, \what y_{t'-1}, \what y_{t'+1},\cdots,  \what y_d \right\}.
\]
Let $\# S_t$ be the number of elements in $S_t$, i.e., $n+d-2$. Then, for any $0< \omega^*\leqs \frac{\Omega}{\#S_t}$, we have the following relations:
\begin{align}\label{equ:locaintenrelation1}
&\what {a}_{t'}\prod_{q\in S_t}\left(e^{i\what y_{t'}\omega^*}-e^{i q\omega^*}\right)- a_t \prod_{q\in S_t} \left(e^{i y_t\omega^*}- e^{i q\omega^*} \right)= \vect w_1^{\top}\vect v.
\end{align}
Moreover, for any $0< \omega^*\leqs \frac{\Omega}{\#S_t+1}$, we have
\begin{align}\label{equ:locaintenrelation2}
a_t \prod_{q\in S_t\cup \{\what y_{t'}\}}\left(e^{iy_t\omega^*} - e^{iq\omega^*}\right) =  \left(e^{i \what y_{t'}\omega^*}\vect w_1- \vect w_2\right)^{\top}\vect v.
\end{align}
Here, $\vect w_1 = (\vect w(0), \vect w(\omega^*), \cdots, \vect w((\#S_t)\omega^*))^{\top}$, $\vect w_2 = (\vect w(1), \vect w(\omega^*), \cdots, \vect w((\#S_t+1)\omega^*))^{\top}$ and the vector $\vect v$ is given by
\[
 \left((-1)^{\# S_t}\sum_{\{q_1,\cdots, q_{\#S_t}\}\in S_{j, \# S_t}} e^{iq_1\omega^*}\cdots e^{iq_{\#S_t}\omega^*}, \ \cdots,\ (-1)^2\sum_{\{q_1, q_2\}\in S_{t,2}} e^{iq_1\omega^*}e^{iq_2\omega^*}, \ (-1)\sum_{\{q_1\}\in S_{t,1}} e^{iq_1\omega^*}, 1 \right)^{\top},  
\]
where $S_{t,p}:=\left\{\{q_1, \cdots, q_p\}\bigg| q_j\in S_t, 1\leqs j\leqs p, \text{ $q_{j'}$ and $q_{j''}$ are different elements in $S_t$ when $j'\neq j''$} \right\}$, $\ p=1,\cdots, \#S_t$.
\end{thm}

\bigskip
For the convenience of the applications of our location-amplitude identities, we derive the following corollary, as a direct consequence of Theorem \ref{thm:locaintenrelation1}. 
\begin{cor}\label{cor:locaintenrelation1}
Consider the model
\[
\mathcal F[\what \mu](\omega) = \mathcal F[\mu](\omega)+\vect w(\omega), \quad \omega \in [0, \Omega],
\]
where $\what \mu = \sum_{j=1}^d\what a_j\delta_{\what y_j}$ and $\mu=\sum_{j=1}^na_j\delta_{y_j}$ and assume that $|\vect w(\omega)|< \sigma, \omega \in [0, \Omega]$. For any fixed $y_t$ and $\what y_{t'}$, define the set $S_t$ as 
\begin{equation}\label{equ:defiofst1}
S_t:=\left\{y_1,\cdots, y_{t-1}, y_{t+1}, \cdots ,y_n, \what y_1,\cdots, \what y_{t'-1}, \what y_{t'+1},\cdots,  \what y_d \right\}.
\end{equation}
Let $\# S_t$ be the number of elements in $S_t$, i.e., $n+d-2$. Then, for any $0< \omega^*\leqs \frac{\Omega}{\#S_t}$, we have 
\begin{align}\label{equ:locaintenrelation3}
\left|\what {a}_{t'}\prod_{q\in S_t}\left(e^{i\what y_{t'}\omega^*}-e^{i q\omega^*}\right)- a_t \prod_{q\in S_t} \left(e^{i y_t\omega^*}- e^{i q\omega^*} \right)\right|
<  2^{\#S_t}\sigma.
\end{align}
Moreover,  for any $0< \omega^*\leqs \frac{\Omega}{\#S_t+1}$, we have  
\begin{align}\label{equ:locaintenrelation4}
\left|a_t \prod_{q\in S_t\cup \{\what y_{t'}\}}\left(e^{iy_t\omega^*} - e^{iq\omega^*}\right)  \right|< 2^{\#S_t+1}\sigma. 
\end{align}
\end{cor}
\begin{proof}
This is a direct consequence of Theorem \ref{thm:locaintenrelation1} in view of 
\[
|\vect w_1^{\top}\vect v|< 2^{\#S_t}\sigma, \qquad |(e^{i \what y_{t'} \omega^*}\vect w_1 -\vect w_2)^{\top}\vect v|< 2^{\#S_t+1}\sigma. 
\]
\end{proof}

\subsection{Proof of Theorem \ref{thm:locaintenrelation1}}
Before starting the proof, we first introduce some notation and lemmas. Denote  by
\begin{equation}\label{equ:phiformula}
\phi_{p,q}(t) = \left(t^{p}, t^{p+1}, \cdots, t^{q}\right)^{\top}.
\end{equation}

The following lemma on the inverse of the Vandermonde matrix is standard. 
\begin{lem}{\label{lem:invervandermonde0}}
Let $V_k$ be the Vandermonde matrix $\left(\phi_{0, k-1}(t_1), \cdots, \phi_{0, k-1}(t_k)\right)$. Then its inverse $V_k^{-1}=B$ can be specified as follows:
\[
B_{j q}=\left\{\begin{array}{cc}(-1)^{k-q}\left(\frac{\sum_{\substack{1 \leqs m_1<\ldots<m_{k-q} \leqs k \\ m_1, \ldots, m_{k-q} \neq j}} t_{m_1} \cdots t_{m_{k-q}}}{\prod_{\substack{1 \leqs m \leqs k \\ m \neq j}}\left(t_j-t_m\right)}\right), & \quad  1 \leqs q<k, \\
\frac{1}{\prod_{\substack{1 \leqs m \leqs k \\ m \neq j}}\left(t_j-t_m\right)}, & \quad  q=k.
\end{array}\right.
\]
\end{lem}
The following lemma can be deduced from the inverse of the Vandermonde matrix and the readers can check Lemma 5 in \cite{liu2021theorylse} for a simple proof, although the numbers there are real numbers. 
\begin{lem}{\label{lem:invervandermonde1}}
	Let $t_1, \cdots, t_k$ be $k$ different complex numbers. For $t\in \mathbb C$, we have
	\[
	\left(V_k^{-1}\phi_{0,k-1}(t)\right)_{j}=\prod_{1\leqs q\leqs k,q\neq j}\frac{t- t_q}{t_j- t_q},
	\]
	where $V_k:=  \big(\phi_{0,k-1}(t_1),\cdots,\phi_{0,k-1}(t_k)\big)$ with $\phi_{0, k-1}(\cdot)$ being defined by (\ref{equ:phiformula}). 
\end{lem}

\medskip
Now we start the main proof. 
\begin{proof}
We only prove the theorem for $\omega^* \leqs  \frac{\Omega}{\#S_t+1}$. The case when $\omega^* \leqs  \frac{\Omega}{\#S_t}$ for (\ref{equ:locaintenrelation1}) is obvious afterwards. We fix  $t\in \{1,\cdots, n\}$ in the following derivations. From (\ref{equ:modelsetting0}),  we can write
\begin{align} \label{equ:matrixrelation0}
    \what A \what a = A a +W, 
\end{align}
where $\what a = (\what a_1, \cdots, \what a_d)^{\top}$, $a = (a_1, \cdots, a_n)^{\top}, W =\left(\vect w(0), \vect w(\omega^*), \cdots, \vect w((\#S_t+1)\omega^*)\right)^{\top}$ and 
\begin{align*}
&\what A = \left( \phi_{0,\#S_t+1}(e^{i \what y_1 \omega^*}), \ \phi_{0, \#S_t+1}(e^{i \what y_2 \omega^*}), \ \cdots, \ \phi_{0,\#S_t+1}(e^{i \what y_d   \omega^*})\right), \\
& A = \left( \phi_{0,\#S_t+1}(e^{i y_1 \omega^*}), \ \phi_{0,\#S_t+1}(e^{i  y_2 \omega^*}), \ \cdots, \ \phi_{0,\#S_t+1}(e^{i y_n \omega^*})\right),
\end{align*} 
with $0< \omega^*\leqs \frac{\Omega}{\#S_t+1}$. We further decompose (\ref{equ:matrixrelation0}) into the following two equations:
\begin{equation}\label{equ:matrixrelation1}
\what {a}_{t'} \phi_{0, \#S_t}(e^{i \what y_{t'}\omega^*}) = B_1 b + \vect w_1, \quad   \what {a}_{t'} \phi_{1, \#S_t+1}(e^{i \what y_{t'}\omega^*}) = B_2 b + \vect w_2, 
\end{equation}
where $\vect w_1 = (\vect w(0), \vect w(\omega^*), \cdots, \vect w(\#S_t \omega^*))^{\top}$, $\vect w_2 = (\vect w(\omega^*), \cdots, \vect w((\#S_t+1)\omega^*))^{\top}$ and 
\begin{align*}
B_1 &= \left( \phi_{0, \#S_t}(e^{i y_{1}\omega^*}),\ \cdots,\ \phi_{0, \#S_t}(e^{i y_{n}\omega^*}),\ \phi_{0, \#S_t}(e^{i\what  y_{1}\omega^*}),\ \cdots, \phi_{0, \#S_t}(e^{i\what  y_{t'-1}\omega^*}),\right. \\
& \qquad \qquad \ \left. \phi_{0, \#S_t}(e^{i\what  y_{t'+1}\omega^*}),\ \cdots, \ \phi_{0, \#S_t}(e^{i \what y_{d}\omega^*})  \right),\\
B_2 &= \left( \phi_{1, \#S_t+1}(e^{i y_{1}\omega^*}),\ \cdots,\ \phi_{1, \#S_t+1}(e^{i y_{n}\omega^*}),\ \phi_{1, \#S_t+1}(e^{i\what  y_{1}\omega^*}),\ \cdots, \phi_{1, \#S_t+1}(e^{i\what  y_{t'-1}\omega^*}),\right. \\
& \qquad \qquad \ \left. \phi_{1, \#S_t+1}(e^{i\what  y_{t'+1}\omega^*}),\ \cdots, \ \phi_{1, \#S_t+1}(e^{i \what y_{d}\omega^*})  \right).
\end{align*}
We first consider the case when all the $e^{y_j \omega^*}$'s and $e^{\what y_j\omega^*}$'s are distinct. Thus, $b=(b_1, \cdots, b_{\#S_t+1})$ in (\ref{equ:matrixrelation1}) is such that 
\[
b_{l} = \left\{\begin{array}{ll} 
a_l, & 1\leqs l\leqs n, \\
- \what a_{l-n}, & n<l\leq n+t'-1,\\
- \what a_{l-n+1}, & n+t'-1<l.
\end{array}\right.
\]
Observe that 
\begin{align*}
&B_2 = B_1 \mathrm{diag}\left(e^{i y_1\omega^*} ,\cdots, e^{i y_n \omega^*}, e^{i \what y_{1} \omega^*}, \cdots, e^{i \what y_{t'-1} \omega^*}, e^{i \what y_{t'+1} \omega^*},\cdots, e^{i \what y_{d} \omega^*}\right),\\
& \phi_{1, \#S_t+1}(e^{i y_{l}\omega^*}) = e^{i y_{l}\omega^*}\phi_{0, \#S_t}(e^{i y_{l}\omega^*}), 
\end{align*}
we rewrite (\ref{equ:matrixrelation1}) as 
\begin{equation}\label{equ:matrixrelation2}
\begin{aligned}
&\what {a}_{t'} \phi_{0, \#S_t}(e^{i \what y_{t'}\omega^*}) = B_1 b + \vect w_1,\\  
&e^{i \what y_{t'}\omega^*}\what {a}_{t'} \phi_{0, \#S_t}(e^{i \what y_{t'}\omega^*}) = B_1\mathrm{diag}\left(e^{i y_1\omega^*} ,\cdots, e^{i y_n \omega^*}, e^{i \what y_{1} \omega^*}, \cdots, e^{i \what y_{t'-1} \omega^*}, e^{i \what y_{t'+1} \omega^*},\cdots, e^{i \what y_{d} \omega^*}\right)b + \vect w_2. 
\end{aligned}
\end{equation}
Since all the $e^{iy_j\omega^*}$'s and $e^{i\what y_j \omega^*}$'s are pairwise distinct, $B_1$ is a regular matrix. We multiply  both sides of the above equations by the inverse of $B_1$ to get from Lemma \ref{lem:invervandermonde1} that 
\begin{align}
\what {a}_{t'}\prod_{q\in S_t}\frac{e^{i\what y_{t'}\omega^*}-e^{i q\omega^*}}{e^{i y_t\omega^*}- e^{i q\omega^*}} &= a_t + (B_1^{-1})_{t} \vect w_1, \label{equ:intensityrelation1}\\
e^{i \what y_{t'}\omega^*} \what {a}_{t'}\prod_{q\in S_t}\frac{e^{i\what y_{t'}\omega^*}-e^{i q\omega^*}}{e^{i y_t\omega^*}- e^{i q\omega^*}} &= e^{i y_{t}\omega^*} a_t + (B_1^{-1})_{t} \vect w_2, \label{equ:intensityrelation2}
\end{align}
where $(B_1^{-1})_{t}$ is the $t$-th row of $B_1^{-1}$. By Lemma \ref{lem:invervandermonde0},  it follows that
\begin{equation}\label{equ:matrixrelation3}
(B_1^{-1})_{t}\vect w_1 =  \frac{\sum_{p=0}^{\#S_t-1}\left(\vect w(p \omega^*) (-1)^{\#S_t-p}\sum_{\{q_1,\cdots, q_{\#S_{t}-p}\}\in S_{t, \#S_t-p}} e^{iq_1\omega^*}\cdots e^{iq_{\#S_t-p}\omega^*}\right)+ \vect w(\#S_t \omega^*)}{\prod_{q\in S_t}(e^{iy_t\omega^*} - e^{i q \omega^*})}.  
\end{equation}
Thus, multiplying $\prod_{q\in S_t}\left(e^{i y_t\omega^*}- e^{i q\omega^*}\right)$ in both sides of (\ref{equ:intensityrelation1}) proves (\ref{equ:locaintenrelation1}) in the case when all the $e^{iy_j\omega^*}$'s and $e^{i\what y_j \omega^*}$'s are pairwise distinct. Furthermore, equation (\ref{equ:intensityrelation1}) times $e^{i \what y_{t'} \omega^*}$ minus (\ref{equ:intensityrelation2}) yields 
\[
(e^{i y_{t}\omega^*} - e^{i \what y_{t'}\omega^*})a_t = (B_1^{-1})_{t}\left(e^{i \what y_{t'} \omega^*}\vect w_1-\vect w_2\right).
\]
Similarly, further expanding $(B_1^{-1})_{t}\left(e^{i \what y_{t'} \omega^*}\vect w_1-\vect w_2\right)$ explicitly by Lemma \ref{lem:invervandermonde0} and multiplying $\prod_{q\in S_t}\left(e^{i y_t\omega^*}- e^{i q\omega^*}\right)$ in both sides above yields (\ref{equ:locaintenrelation2}). 

Finally, we consider the case when the $e^{iy_j\omega^*}$'s and $e^{i\what y_j \omega^*}$'s are not pairwise distinct. Since it is a limiting case of the above cases, (\ref{equ:locaintenrelation1}) and (\ref{equ:locaintenrelation2}) still hold in this case. This completes the proof. 
\end{proof}

\section{Stability of super-resolution in dimension one}\label{section:stabilitySR}

In this section, based on our location-amplitude identities, we analyze the super-resolution capability of the reconstruction of the numbers, locations, and amplitudes of off-the-grid sources in the one-dimensional super-resolution problem. Note that these problems have been analyzed in \cite{liu2021theorylse, batenkov2019super} from different perspectives but the proofs are over several tens of pages. Now, by our method, we have a direct and clear picture of all these problems, which allows us to prove them in a unified way and in less than ten pages.

We consider the imaging model (\ref{equ:highdmodelsetting1}) and focus on the one-dimensional case in this section. Since the source locations $\vect y_j$'s are the supports of the Dirac masses in $\mu$, we use the support recovery for a substitution of the location reconstruction throughout the paper.

Since we focus on the resolution limit case, we consider the case when the point sources are tightly spaced and form a cluster. To be more specific, we denote the ball in the $k$-dimensional space by
\begin{equation}\label{equ:highdball}
B_\delta^k(\mathbf{x}):=\left\{\mathbf{y} \ \bigg| \mathbf{y} \in \mathbb{R}^k,\|\mathbf{y}-\mathbf{x}\|_2<\delta\right\},
\end{equation}
and assume $\mathbf{y}_j \in B_{\frac{(n-1) \pi}{2\Omega}}^k(\mathbf{0}), j=1, \ldots, n$, or equivalently $\left\|\mathbf{y}_j\right\|_2<\frac{(n-1) \pi}{2 \Omega}$. This assumption is a common assumption for super-resolving the off-the-grid sources \cite{liu2021theorylse, batenkov2019super} and is necessary for the analysis. Since we are interested in resolving closely-spaced sources, it is also reasonable. We remark that our results for sources in $B_{\frac{(n-1) \pi}{2\Omega}}^k(\mathbf{0})$ can be directly generalized to sources in $B_{\frac{(n-1) \pi}{2\Omega}}^k(\mathbf{x}),\ \vect x\in \mathbb R^k$. 

The reconstruction process is usually targeting at some specific solutions in a so-called admissible set, which comprises discrete measures whose Fourier data are sufficiently close to $\vect Y$. In general, every admissible measure is possibly the ground truth and it is impossible to distinguish which one is closer to the ground truth without any additional prior information. In our problem, we introduce the following concept of $\sigma$-admissible discrete measures. For simplicity, we also call them $\sigma$-admissible measures.

\begin{defi}\label{defi:highdsigmaadmissmeasure}
	Given the measurement $\mathbf Y$ in (\ref{equ:highdmodelsetting1}), $\what \mu=\sum_{j=1}^{d} \what a_j \delta_{\mathbf {\what y}_j}$ is said to be a $\sigma$-admissible discrete measure of \, $\mathbf Y$ if
	\[
	\babs{\mathcal F[\what \mu](\vect \omega)-\mathbf Y(\vect \omega)}< \sigma, \quad \bnorm{\vect \omega}_2 \leqs \Omega.
	\]
	If further $\what a_j>0, j=1, \cdots, d$, then $\what \mu$ is said to be a positive $\sigma$-admissible discrete measure of \, $\mathbf Y$. 
\end{defi}

\subsection{Stability of number detection}\label{section:numberdetection}
In this section, we estimate the super-resolving capability of number detection in the super-resolution problem. We introduce the concept of computational resolution limit for number detection \cite{liu2021mathematicaloned, liu2021theorylse, liu2021mathematicalhighd} and present a sharp bound for it. 

Note the set of $\sigma$-admissible measures of $\mathbf Y$ characterizes all possible solutions to our super-resolution problem with the given measurement $\mathbf Y$. Detecting the source number $n$ is possible only if all of the admissible measures have at least $n$ supports, otherwise, it is impossible to detect the correct source number without additional a priori information. Thus, following definitions similar to those in \cite{liu2021theorylse, liu2021mathematicaloned, liu2021mathematicalhighd}, we define the computational resolution limit for the number detection problem as follows. 

\begin{defi}\label{defi:highdcomputresolimit}
The computational resolution limit to the number detection problem in the super-resolution of sources in $\mathbb R^k$ is defined as the smallest nonnegative number $\mathcal D_{num}(k, n)$ such that for all $n$-sparse measures $\sum_{j=1}^{n}a_{j}\delta_{\vect y_j}, a_j\in \mathbb C,\vect y_j \in B_{\frac{(n-1) \pi}{2\Omega}}^k(\mathbf{0})$ and the associated measurement $\vect Y$ in (\ref{equ:highdmodelsetting1}), if 
	\[
	\min_{p\neq j} \bnorm{\vect y_j-\vect y_p}_2 \geqs \mathcal D_{num}(k, n),
	\]
then there does not exist any $\sigma$-admissible measure of \, $\mathbf Y$ with less than $n$ supports. In particular, when considering positive sources and positive $\sigma$-admissible measures, the corresponding computational resolution limit is denoted by $\mathcal D_{num}^+(k, n)$.
\end{defi}

The definition of ``computational resolution limit'' emphasizes the impossibility of correctly detecting the number of very close sources by any means. It depends crucially on the signal-to-noise ratio and the sparsity of the sources, which is fundamentally different from all the classical resolution limits \cite{abbe1873beitrage, volkmann1966ernst, rayleigh1879xxxi, schuster1904introduction, sparrow1916spectroscopic} that depend only on the cutoff frequency. 

Our first result is a sharp estimate for the upper bound of the computational resolution limit in the one-dimensional super-resolution problem. As we have said, we will use $y_j$'s to denote the one-dimensional source locations.

\begin{thm}\label{thm:upperboundnumberlimithm0}
    Let the measurement $\mathbf Y$ in (\ref{equ:highdmodelsetting1}) be generated by any one-dimensional source $\mu =\sum_{j=1}^{n}a_j\delta_{y_j}$ with $y_j \in B_{\frac{(n-1) \pi}{2\Omega}}^1(0), j=1,\cdots, n$. Let $n\geqs 2$ and assume that the following separation condition is satisfied 
	\begin{equation}\label{upperboundnumberlimithm0equ0}
	\min_{p\neq j}\Big|y_p-y_j\Big|\geqs \frac{2e\pi }{\Omega }\Big(\frac{\sigma}{m_{\min}}\Big)^{\frac{1}{2n-2}},
	\end{equation}
 where $\sigma, m_{\min}$ are defined as in (\ref{equ:noiselevel1}), (\ref{equ:intendisset}), respectively. Then there does not exist any $\sigma$-admissible measures of \,$\mathbf Y$ with less than $n$ supports. Moreover, for the cases when $n=2$ and $n=3$, if 
 \[
\min_{p\neq j}\Bigg|y_p-y_j\Bigg|\geqs \frac{2\arcsin\left(2\left(\frac{\sigma}{m_{\min}}\right)^{\frac{1}{2}}\right)}{\Omega}, \quad \min_{p\neq j}\Bigg|y_p-y_j\Bigg|\geqs \frac{2\pi}{\Omega }\Big(\frac{8\sigma}{m_{\min}}\Big)^{\frac{1}{4}}, \quad \text{respectively},
 \]
 then there does not exist any $\sigma$-admissible measures of \,$\mathbf Y$ with less than $n$ supports.
\end{thm}

Theorem \ref{thm:upperboundnumberlimithm0} gives a sharper upper bound for the computational resolution limit $\mathcal D_{num}(1, n)$ compared to the one in \cite{liu2021theorylse}. By the new estimate (\ref{thm:upperboundnumberlimithm0}), it is already possible to surpass the Rayleigh limit $\frac{\pi}{\Omega}$ in detecting source number when $\frac{m_{\min}}{\sigma}\geqs (2e)^{2n-2}$. Moreover, this upper bound is shown to be sharp by a lower bound provided in \cite{liu2022rslpositive}. Thus, we can conclude that 
\[
\frac{2e^{-1}}{\Omega }\Big(\frac{\sigma}{m_{\min}}\Big)^{\frac{1}{2n-2}}< \mathcal D_{num}(1, n) \leqs \frac{2e\pi}{\Omega }\Big(\frac{\sigma}{m_{\min}}\Big)^{\frac{1}{2n-2}}.
\]
It is also easy to generalize the estimates in Theorem \ref{thm:upperboundnumberlimithm0} to high-dimensional spaces by methods in \cite{liu2021mathematicalhighd, liu2023improved}, whereby we can obtain that 
\[
\frac{2e^{-1}}{\Omega }\Big(\frac{\sigma}{m_{\min}}\Big)^{\frac{1}{2n-2}}< \mathcal D_{num}(k, n) \leqs \frac{C_1(k,n)}{\Omega }\Big(\frac{\sigma}{m_{\min}}\Big)^{\frac{1}{2n-2}}
\]
for a constant $C_1(k,n)$.

In particular, for the case when $n=2$, our estimate demonstrates that when the signal-to-noise ratio $\mathrm{SNR}>4$, then the resolution is better than the Rayleigh limit and the "super-resolution in number detection" can be achieved thusly. This result is already practically important. As we will see later, our estimate is very sharp and close to the true two-point resolution limit.

\begin{remark}
We remark that our new techniques also provide a way to analyze the stability of number detection for sources with  multi-cluster patterns. Our former method (also the only one we know of) for analyzing the stability of number detection cannot handle such cases. The technique here is the first known method that can tackle the issue. But since the current paper focuses on understanding the resolution limits in the super-resolution, the multi-cluster case is out of scope and we leave it as a future work. 
\end{remark}

%\begin{equation}\label{upperboundnumberlimithm0equ0}
% 	\babs{y_p-y_j}_2\geqs \frac{2\arcsin\left(2\left(\frac{\sigma}{m_{\min}}\right)^{\frac{1}{2}}\right)}{\Omega}.
% 	\end{equation}
% 	Then there do not exist any positive $\sigma$-admissible measures of \,$\mathbf Y$ with only one support.
% \end{thm}

We now present the proof of Theorem \ref{thm:upperboundnumberlimithm0}. The problem is essentially a nonlinear approximation problem where we have to optimize the approximation over the coupled factors: source number $d$, source locations $\what y_j$'s, and amplitudes $\what a_j$'s. Here, by leveraging the location-amplitude identities, we prove it in a rather simple and direct way.

We first denote for an integer $k\geqs 1$, 
\begin{equation}\label{zetaxiformula1}
\begin{aligned}
&\zeta(k)= \left\{
\begin{array}{cc}
(\frac{k-1}{2}!)^2,& \text{if $k$ is odd,}\\
(\frac{k}{2})!(\frac{k-2}{2})!,& \text{if $k$ is even,}
\end{array} 
\right. \quad \xi(k)=\left\{
\begin{array}{cc}
\frac{1}{2},  & \text{if } k=1,\\
\frac{(\frac{k-1}{2})!(\frac{k-3}{2})!}{4},& \text{if $k$ is odd,\,\,$ k\geqs 3$,}\\
\frac{(\frac{k-2}{2}!)^2}{4},& \text{if $k$ is even}.
\end{array} 
\right.	
\end{aligned}
\end{equation}
We also define for positive integers $p, q$, and $z_1, \cdots, z_p, \what z_1, \cdots, \what z_q \in \mathbb C$, the following vector in $\mathbb{R}^p$
\begin{equation}\label{notation:eta}
\eta_{p,q}(z_1,\cdots,z_{p}, \what z_1,\cdots,\what z_q)=\left(\begin{array}{c}
|(z_1-\what z_1)|\cdots|(z_1-\what z_q)|\\
|(z_2-\what z_1)|\cdots|(z_2-\what z_q)|\\
\vdots\\
|(z_{p}-\what z_1)|\cdots|(z_{p}-\what  z_q)|
\end{array}\right).
\end{equation}
We recall the following useful lemmas. 
\begin{lem}\label{lem:multiproductlowerbound0}
	Let $-\frac{\pi}{2}\leqs \theta_1<\theta_2<\cdots<\theta_{k}  \leqs \frac{\pi}{2}$ with $\min_{p\neq j}|\theta_p-\theta_j|=\theta_{\min}$. We have the estimate 	
	\[
	\prod_{1\leqs p\leqs k,p\neq j}{\left|e^{i\theta_j}-e^{i\theta_p}\right|}\geqs \zeta(k)\left(\frac{2\theta_{\min}}{\pi}\right)^{k-1},\  j=1,\cdots, k,
	\]
 where $\zeta(k)$ is defined in (\ref{zetaxiformula1}).
\end{lem}
\begin{proof}
Note that 
\begin{equation} \label{eq-theta}
\left|e^{i\theta_j}-e^{i\theta_p} \right| \geqs \frac{2}{\pi} \left|\theta_j-\theta_p\right|, \quad \mbox{for all}\,\, \theta_j, \theta_p \in \Big[-\frac{\pi}{2}, \frac{\pi}{2}\Big].
\end{equation}
Then we have
\begin{align*}
\prod_{1\leqs p\leqs k,p\neq j}{\left|e^{i\theta_j}-e^{i\theta_p}\right|}\
\geqs  \left(\frac{2}{\pi}\right)^{k-1} \prod_{1\leqs p\leqs k,p\neq j}{\left|\theta_j-\theta_p\right|} \geqs \zeta(k)\left(\frac{2\theta_{\min}}{\pi}\right)^{k-1}.
\end{align*}
\end{proof}

% \begin{lem}\label{lem:multiproductlowerbound0}
% 	Let $-\frac{\pi}{2}\leqs \theta_1<\theta_2<\cdots<\theta_{k+1}  \leqs \frac{\pi}{2}$ and let $\theta_{\min}=\min_{p\neq j}|\theta_p-\theta_j|$. Then 
% 	\[\min_{\what \theta_1\in \mathbb R, \cdots,\what \theta_k\in \mathbb R}||\eta_{k+1,k}(\theta_1,\cdots,\theta_{k+1},\what \theta_1,\cdots,\what \theta_k)||_{\infty}\geqs \xi(k)(\theta_{\min})^k,\]
% 	where  $\xi(k)$ is defined in (\ref{zetaxiformula1}) and $\eta_{k+1,k}$ defined in (\ref{notation:eta}).
% \end{lem}
% \begin{proof}
% See Lemma 6 in \cite{liu2021theorylse}.
% \end{proof}

\begin{lem}\label{lem:multiproductlowerbound1}
	Let $-\frac{\pi}{2}\leqs \theta_1<\theta_2<\cdots<\theta_{k+1}  \leqs \frac{\pi}{2}$. Assume that $\min_{p\neq j}|\theta_p-\theta_j|=\theta_{\min}$. Then for any $\what \theta_1,\cdots, \what \theta_k\in \mathbb R$, we have the following estimate: 	
	\[\bnorm{\eta_{k+1,k}(e^{i\theta_1},\cdots,e^{i\theta_{k+1}},e^{i \what \theta_1},\cdots,e^{i\what \theta_k})}_{\infty}\geqs \xi(k)\left(\frac{2\theta_{\min}}{\pi}\right)^k, 
	\] 
 where $\eta_{k+1,k}$ is defined as in (\ref{notation:eta}). 
\end{lem}
\begin{proof}
See Corollary 7 in \cite{liu2021theorylse}.
\end{proof}

\begin{proof}
We are now ready to prove Theorem \ref{thm:upperboundnumberlimithm0}. 
Suppose that $\what \mu=\sum_{j=1}^{d}\what a_j \delta_{\what y_j}, d\leqs n-1$ is a $\sigma$-admissible measure of $\vect Y$. By the Definition \ref{defi:highdsigmaadmissmeasure} and the model (\ref{equ:highdmodelsetting1}), we have 
\[
\mathcal F[\what \mu](\omega)= \mathcal F[\mu](\omega)+\vect W_1(\omega),\quad  \omega \in [-\Omega,\Omega] 
\]
for some $\vect W_1$ with $|\vect W_1(\omega)|<2\sigma$. Note that by adding some point sources in $\what \mu$, from above we can actually construct
$\what \mu=\sum_{j=1}^{n-1}\what a_j \delta_{\what y_j}$ such that
\[
\mathcal F[\what \mu](\omega)= \mathcal F[\mu](\omega)+\vect W_2(\omega),\quad  \omega \in [-\Omega,\Omega], 
\]
for some $\vect W_2$ with $|\vect W_2(\omega)|<2\sigma$. For ease of exposition, we consider in the following that the measure $\what \mu$ is with $n-1$ point sources. On the other hand, the above equation implies that $\what \rho = \sum_{j=1}^{n-1}e^{-\what y_j \Omega}\what a_{j} \delta_{\what y_j}$ and $\rho = \sum_{j=1}^ne^{-y_j \Omega}a_{j} \delta_{y_j}$ satisfy  
\begin{equation}\label{equ:proofnumber0}
\mathcal F[\what \rho](\omega)= \mathcal F[\rho](\omega)+\vect W_3(\omega),\quad  \omega \in [0,2\Omega],
\end{equation}
for some $\vect W_3$ with $|\vect W_3(\omega)|<2\sigma, \omega \in [0,2\Omega]$. For any $y_t$ and $\what y_{t'}$, define $S_t$ as 
 \[
S_t:= \left\{y_1,\cdots, y_{t-1}, y_{t+1}, \cdots ,y_n, \what y_1,\cdots, \what y_{t'-1}, \what y_{t'+1},\cdots,  \what y_{n-1} \right\}.
\]
Then $\#S_t = 2n-3$. Let $\omega^*=\frac{2\Omega}{2n-2}$. Applying (\ref{equ:locaintenrelation4}) to (\ref{equ:proofnumber0}) we obtain that
\[
\prod_{q=1,\cdots, n, q\neq t} \left|e^{i y_{t}\omega^*} - e^{i y_{q}\omega^*} \right|\prod_{q=1,\cdots, n-1}\left|e^{iy_{t}\omega^*} - e^{i\what y_q\omega^*}\right|\left|a_t \right|< 2^{\#S_t+2}\sigma, \ t=1,\cdots, n. 
\]
This gives 
\[
\prod_{q=1,\cdots, n, q\neq t} \left|e^{i y_{t}\omega^*} - e^{i y_{q}\omega^*} \right|\prod_{q=1,\cdots, n-1}\left|e^{iy_{t}\omega^*} - e^{i\what y_q\omega^*}\right|< 2^{2n-1}\frac{\sigma}{m_{\min}}, \ t=1,\cdots, n. 
\]
Therefore, it follows that 
\[
\min_{t=1,\cdots, n}\prod_{q=1,\cdots, n, q\neq t} \left|e^{i y_{t}\omega^*} - e^{i y_{q}\omega^*} \right|\max_{t=1,\cdots, n}\prod_{q=1,\cdots, n-1}\left|e^{iy_{t}\omega^*} - e^{i\what y_q\omega^*}\right|< 2^{2n-1}\frac{\sigma}{m_{\min}}.
\]
Denote $y_q \omega^*$ by $\theta_q$ and $\theta_{\min} = \min_{p\neq q}|\theta_p-\theta_q|$. Since the $y_j$'s are in $B_{\frac{(n-1) \pi}{2\Omega}}^1(0)$ and $\omega^* = \frac{2\Omega}{2n-2}$, we have that the $\theta_j$'s are in $\left[-\frac{\pi}{2}, \frac{\pi}{2}\right]$. Thus, by Lemma \ref{lem:multiproductlowerbound0}, we get
\begin{equation}\label{equ:proofnumber3}
\min_{t=1,\cdots, n}\prod_{q=1,\cdots, n, q\neq t} \left|e^{i y_{t}\omega^*} - e^{i y_{q}\omega^*} \right|\geqs \zeta(n)\left(\frac{2\theta_{\min}}{\pi}\right)^{n-1}.
\end{equation}
Moreover, using Lemma \ref{lem:multiproductlowerbound1} yields
\begin{equation}\label{equ:proofnumber4}
\max_{t=1, \cdots, n}\prod_{q=1,\cdots,n-1}\left|e^{iy_{t}\omega^*} - e^{i\what y_q\omega^*}\right|\geqs \xi(n-1) \left(\frac{2\theta_{\min}}{\pi}\right)^{n-1}. 
\end{equation}
Combining the above estimates, it follows that 
\begin{equation}\label{equ:proofnumber5}
\zeta(n)\xi(n-1) \left(\frac{2\theta_{\min}}{\pi}\right)^{2n-2}< \frac{2^{2n-1}\sigma}{m_{\min}}.
\end{equation}
Thus, 
\[
\theta_{\min} < \pi\left( \frac{2}{\zeta(n)\xi(n-1)}\right)^{\frac{1}{2n-2}}\left(\frac{\sigma}{m_{\min}}\right)^{\frac{1}{2n-2}},
\]
and consequently,
\begin{equation}\label{equ:proofnumber2}
d_{\min}= \frac{\theta_{\min}}{\omega^*}< \frac{(2n-2)\pi}{2\Omega}\left( \frac{2}{\zeta(n)\xi(n-1)}\right)^{\frac{1}{2n-2}}\left(\frac{\sigma}{m_{\min}}\right)^{\frac{1}{2n-2}}\leqs \frac{2\pi e}{\Omega}\left(\frac{\sigma}{m_{\min}}\right)^{\frac{1}{2n-2}}, 
\end{equation}
where $d_{\min}:=\min_{p\neq q}|y_p-y_q|$ and the last inequality is from Lemma \ref{lem:upperboundnumbercalculate1}. Therefore, if $d_{\min}\geqs \frac{2\pi e}{\Omega}\left(\frac{\sigma}{m_{\min}}\right)^{\frac{1}{2n-2}}$, then  there is no $\sigma$-admissible measure of $\vect Y$ with less than $n$ supports. 

The last part consists in proving the cases when $n=2,3$. When $n=3$, the result is enhanced by noting that $\frac{2}{\zeta(n)\xi(n-1)}=8$ in (\ref{equ:proofnumber2}). When $n=2$, the result is enhanced by improving the estimates (\ref{equ:proofnumber3}) and (\ref{equ:proofnumber4}). For (\ref{equ:proofnumber3}), we now have
\[
\babs{e^{iy_1\omega^*} - e^{iy_2\omega^*}}\geqs 2\sin \left(\frac{\theta_{\min}}{2}\right),
\]
where $\theta_{\min} = |y_1\omega^*-y_2\omega^*|$ and $\omega^*= \Omega$. For (\ref{equ:proofnumber4}), we have 
\[
\max_{j=1,2}\left|e^{iy_{j}\omega^*} - e^{i\what y_1\omega^*}\right|\geqs 2\sin\left(\frac{\theta_{\min}}{4}\right).
\]
Thus, similarly to (\ref{equ:proofnumber5}),  we obtain that
\[
2\sin \left(\frac{\theta_{\min}}{2}\right)2\sin\left(\frac{\theta_{\min}}{4}\right)<\frac{2^{3}\sigma}{m_{\min}},
\]
which gives 
\[
\sin^2 \left(\frac{\theta_{\min}}{2}\right)<\frac{4\sigma}{m_{\min}}.
\]
It then follows that 
\[
d_{\min} = \frac{\theta_{\min}}{\omega^*} < \frac{2\arcsin\left(2\left(\frac{\sigma}{m_{\min}}\right)^{\frac{1}{2}}\right)}{\Omega},
\]
which completes the proof. 
\end{proof}

\medskip
\noindent \textbf{Comparison with results in \cite{liu2021theorylse}:} Considering that the results in this section are closely related to results in our previous work \cite{liu2021theorylse}, we highlight the major difference between them. Note first that Theorem \ref{thm:upperboundnumberlimithm0} gives a sharper upper bound for $\mathcal D_{num}(1, n)$ than the previous estimate in \cite{liu2021theorylse}. It also significantly improves the resolution estimate for resolving two sources, enhancing their practical relevance.  These improvements stem from employing location-amplitude identities, a more essential and powerful method than the one used in \cite{liu2021theorylse, liu2021mathematicaloned}. In particular, \cite{liu2021theorylse, liu2021mathematicaloned} established the approximation theory in the Vandermonde space by some algebraic manipulations, while the derived location-amplitude identities here reveal in depth the essence of the super-resolution problem. For example, identities (\ref{equ:locaintenrelation1}) and (\ref{equ:locaintenrelation2}) (or inequalities (\ref{equ:locaintenrelation3}) and (\ref{equ:locaintenrelation4})) reveal directly the relation between the amplitudes and locations of the true sources and the recovered ones, demonstrating that the stabilities of the recoveries are determined by $\frac{\sigma}{\prod_{q\in S_t}(e^{iy_t\omega^*} - e^{iq\omega^*})}$. This analytical perspective transforms the super-resolution problem into analyzing the distribution of the locations of true and recovered sources, leading to optimal stability results for the recovery of source numbers, locations, and amplitudes as substantiated in Theorems \ref{thm:upperboundnumberlimithm0}, \ref{thm:upperboundsupportlimithm0}, and \ref{thm:upperboundamplitudelimithm0}. As a comparison, the method in \cite{liu2021theorylse, liu2021mathematicaloned} deals with only the number detection and location recovery problem, since the algebraic manipulation incurs non-necessary noise amplifications when analyzing the amplitude reconstruction. The method in \cite{batenkov2019super} is only applicable for recovering the location and amplitude of sources due to the "quantitative inverse function theorem" necessitating an equal number of true and recovered sources. Furthermore, our method is capable of analyzing the stability of super-resolving multi-cluster sources under very general settings, making it highly effective.

On the other hand, due to the simplicity and directness of our method, the bounds derived here are nearly optimal and hard to improve. The only parts in the proof deteriorating the resolution estimate are the noise amplification in Corollary \ref{cor:locaintenrelation1} and inequality (\ref{equ:proofnumber5}), which also indicate the path for future improvement.

\subsection{Stability of location reconstruction}\label{section:locationrecovery}
We now consider the location (support) recovery problem in the super-resolution. We first introduce the following concept of $\delta$-neighborhood of a discrete measure. 
\begin{defi} 
Let $\mu=\sum_{j=1}^n a_j \delta_{\mathbf{y}_j}, \mathbf{y}_j \in \mathbb{R}^k$ be a discrete measure and let $\delta>0$ be such that the $n$ balls $B_\delta^k\left(\mathbf{y}_j\right), 1 \leqs j \leqs n$ are pairwise disjoint. We say that $\what{\mu}=\sum_{j=1}^n \what{a}_j \delta_{\what{\mathbf{y}}_j}$ is within $\delta$ neighborhood of $\mu$ if each $\what{\mathbf{y}}_j$ is contained in one and only one of the $n$ balls $B_\delta^k\left(\mathbf{y}_j\right), 1 \leqs j \leqs n$.
\end{defi}
According to the above definition, a measure in a $\delta$-neighborhood preserves the inner structure of the true set of sources. For any stable support recovery algorithm, the output should be a measure in some $\delta$-neighborhood, otherwise it is impossible to distinguish which is the reconstructed location of some $\vect y_j$'s. We now introduce the computational resolution limit for stable support recoveries. For ease of exposition, we only consider measures supported in $B_{\frac{(n-1)\pi}{2\Omega}}^{k}(\vect 0)$. 

\begin{defi}\label{defi:highdcomputresolutionlimit2}
The computational resolution limit to the stable support recovery problem in the super-resolution of sources in $\mathbb R^{k}$ is defined as the smallest nonnegative number $\mathcal D_{supp}(k, n)$ such that for all $n$-sparse measure $\sum_{j=1}^{n}a_{j}\delta_{\vect y_j},\vect  y_j \in B_{\frac{(n-1) \pi}{2\Omega}}^k(\mathbf{0})$ and the associated measurement $\vect Y$ in (\ref{equ:highdmodelsetting1}), if 
	\[
	\min_{p\neq j} \bnorm{\vect y_j-\vect y_p}_2 \geqs \mathcal D_{supp}(k, n),
	\]
	then there exists $\delta>0$ such that any $\sigma$-admissible measure for $\mathbf Y$ with $n$ supports in $B_{\frac{(n-1) \pi}{2\Omega}}^k(\mathbf{0})$ is within a $\delta$-neighborhood of $\mu$. In particular, when considering positive sources and positive $\sigma$-admissible measures, the corresponding computational resolution limit is denoted by $\mathcal D_{supp}^+(k, n)$. 
\end{defi}

%We shall establish both upper and lower bounds of  defined above. 
% To state the results on the resolution limit to stable support recovery in the one-dimensional super, we introduce the super-resolution factor which is defined as the ratio between Rayleigh limit and the minimum separation distance of sources:
% \[
% \mathrm{SRF}:= \frac{\pi}{\Omega d_{\min}},
% \]
% where $d_{\min}=\min_{p\neq j}|y_p-y_j|$. 
Leveraging the location-amplitude identities, we derive the following theorem for stably recovering the source locations in the one-dimensional super-resolution problem. 
\begin{thm}\label{thm:upperboundsupportlimithm0}
	Let $n\geqs 2$, consider the measurement $\vect Y$ in (\ref{equ:highdmodelsetting1}) generated by any one-dimensional source $\mu=\sum_{j=1}^{n}a_j \delta_{y_j}$ supported on $B_{\frac{(n-1) \pi}{2\Omega}}^1(0)$ satisfying 
	\begin{equation}\label{supportlimithm0equ0}
	d_{\min}=\min_{p\neq j}\left|y_p-y_j\right|\geqs \frac{2.36e\pi}{\Omega }\Big(\frac{\sigma}{m_{\min}}\Big)^{\frac{1}{2n-1}}. 
 \end{equation}
	If $\what \mu=\sum_{j=1}^{n}\what a_{j}\delta_{\what y_j}$ supported on $B_{\frac{(n-1) \pi}{2\Omega}}^1(0)$ is a $\sigma$-admissible measure of the measurement $\vect Y$, then $\what \mu$ is within the $\frac{d_{\min}}{2}$-neighborhood of $\mu$. After reordering the $\what y_j$'s, we have 
	\begin{equation}\label{supportlimithm0equ2}
	\Big|\what y_j-y_j\Big|<\frac{C(n)}{\Omega}\mathrm{SRF}^{2n-2}\frac{\sigma}{m_{\min}}, \quad 1\leqs j\leqs n,
	\end{equation}
	where $C(n)=\sqrt{n-0.5}\ 2^{2n-\frac{3}{2}}e^{2n}4.5^{-1}$ and $\mathrm{SRF}:= \frac{\pi}{\Omega d_{\min}}$ is the super-resolution factor. Moreover, for the case when $n=2$, the minimum separation can be improved to 
 \[
d_{\min}\geqs \frac{3}{\Omega}\arcsin\left(2\left(\frac{\sigma}{m_{\min}}\right)^{\frac{1}{3}}\right). 
 \]
\end{thm}

Theorem \ref{thm:upperboundsupportlimithm0} gives an upper bound for the computational resolution limit $\mathcal D_{supp}(1, n)$ that is better than the one in \cite{liu2021theorylse}. It shows that surpassing the Rayleigh limit in the location recovery is possible when $\frac{m_{\min}}{\sigma}\geqs (2.36e)^{2n-1}$. This upper bound is shown to be sharp by a lower bound provided in \cite{liu2022rslpositive}, by which we can conclude now that 
\[
\frac{2e^{-1}}{\Omega }\Big(\frac{\sigma}{m_{\min}}\Big)^{\frac{1}{2n-1}}< \mathcal D_{supp}(1, n) \leqs \frac{2.36e\pi}{\Omega }\Big(\frac{\sigma}{m_{\min}}\Big)^{\frac{1}{2n-1}}.
\]
It is also easy to generalize the estimates in Theorem \ref{thm:upperboundsupportlimithm0} to high-dimensional spaces by methods in \cite{liu2021mathematicalhighd, liu2023improved}, whereby we can obtain that 
\[
\frac{2e^{-1}}{\Omega }\Big(\frac{\sigma}{m_{\min}}\Big)^{\frac{1}{2n-1}}< \mathcal D_{supp}(k, n) \leqs \frac{C_2(k,n)}{\Omega }\Big(\frac{\sigma}{m_{\min}}\Big)^{\frac{1}{2n-1}},
\]
for a constant $C_2(k,n)$.  

Especially, for the case when $n=2$, our estimate demonstrates that when the signal-to-noise ratio $\mathrm{SNR}>12.5$, then the resolution is definitely better than the Rayleigh limit, and the "super-resolution in location recovery" can be achieved. 
%This result is already of practical importance.

% \begin{remark}
% We remark that similar to results in Section 2.1 of \cite{liu2022rslpositive}, by Theorems \ref{thm:upperboundnumberlimithm0} and \ref{thm:upperboundsupportlimithm0}, we can directly show that when 
% \[
% \min_{p\neq j}\left|y_p-y_j\right|\geqs \frac{2.36e\pi}{\Omega }\Big(\frac{\sigma}{m_{\min}}\Big)^{\frac{1}{2n-1}},
% \]
% targeting at a sparest solution in the $\sigma$-admissible measure set will give a solution comprising exactly $n$ sources and the recovered locations are stable. 
% \end{remark}

We now present the proof of Theorem \ref{thm:upperboundsupportlimithm0}. It follows in a straightforward manner after employing the location-amplitude identities.

\medskip
%\noindent\textbf{Proof of Theorem \ref{thm:upperboundsupportlimithm0}:}
We first recall the following auxiliary lemma. 
\begin{lem}\label{lem:stablemultiproduct1}
	For $-\frac{\pi}{2}\leqs \theta_1<\theta_2<\cdots<\theta_k \leqs \frac{\pi}{2}$ and $ \what \theta_1, \what \theta_2, \cdots, \what \theta_k \in \left[-\frac{\pi}{2}, \frac{\pi}{2}\right]$, if
	\begin{align*}
	\bnorm{\eta_{k,k}(e^{i \theta_1},\cdots,e^{i \theta_k}, e^{i \what \theta_1},\cdots, e^{i \what \theta_k})}_{\infty}< \left(\frac{2}{\pi}\right)^{k}\epsilon, 
	\text{ and } \theta_{\min} =\min_{q\neq j}\babs{\theta_q-\theta_j}\geqs  \Big(\frac{4\epsilon}{\lambda(k)}\Big)^{\frac{1}{k}},
	\end{align*}
	where $\eta_{k,k}$ is defined by (\ref{notation:eta}) and $\lambda(k)$ is given by \begin{equation}\label{equ:lambda1}
	\lambda(k)=\left\{
	\begin{array}{ll}
	1,  & k=2,\\
	\xi(k-2),& k\geqs 3,
	\end{array} 
	\right.
	\end{equation}
	with $\xi(\cdot)$ being defined in (\ref{zetaxiformula1}), 
	then after reordering $\what \theta_j$'s, we have
	\begin{equation}\label{equ:satblemultiproductcor1}
	\left|\what \theta_j -\theta_j\right|< \frac{\theta_{\min}}{2}  \text{ and }
	\left|\what \theta_j -\theta_j\right|\leqs \frac{2^{k-1}\epsilon}{(k-2)!(\theta_{\min})^{k-1}}, \quad j=1,\cdots, k.
	\end{equation}
\end{lem}
\begin{proof}
See Corollary 9 in \cite{liu2021theorylse}. 
\end{proof}

Now we start the main proof. 
\begin{proof} Since $\what \mu=\sum_{j=1}^{n}\what a_j \delta_{\what y_j}, \what y_j \in B_{\frac{(n-1) \pi}{2\Omega}}^1(0)$ is a $\sigma$-admissible measure of $\vect Y$, from the model (\ref{equ:highdmodelsetting1}) we have 
\[
\mathcal F[\what \mu](\omega)= \mathcal F[\mu](\omega)+\vect W_1(\omega),\quad  \omega \in [-\Omega,\Omega],  
\]
for some $\vect W_1$ with $|\vect W_1(\omega)|<2\sigma, \omega \in [-\Omega,\Omega] $. This implies that $\what \rho = \sum_{j=1}^{n}e^{-\what y_j \Omega}\what a_{j} \delta_{\what y_j}$ and $\rho = \sum_{j=1}^ne^{-y_j \Omega} a_{j} \delta_{y_j}$ satisfy 
\begin{equation}\label{equ:proofsupport0}
\mathcal F[\what \rho](\omega)= \mathcal F[\rho](\omega)+\vect W_2(\omega),\quad  \omega \in [0,2\Omega],
\end{equation}
for some $\vect W_2$ with $|\vect W_2(\omega)|<2\sigma, \omega \in [0,2\Omega]$. For any $y_t$, let $\what y_{t'}$ be the one in $\what y_j$'s that is the closest to $y_{t}$ and define $S_t$ as 
 \[
S_t:= \left\{y_1,\cdots, y_{t-1}, y_{t+1}, \cdots ,y_n, \what y_1,\cdots, \what y_{t'-1}, \what y_{t'+1},\cdots,  \what y_{n} \right\}.
\]
Then $\#S_t = 2n-2$. Let $\omega^*=\frac{2\Omega}{2n-1}$. Since (\ref{equ:proofsupport0}) holds, by (\ref{equ:locaintenrelation4}) we have 
\begin{equation}\label{equ:proofsupport2}
\prod_{q=1,\cdots, n, q\neq t} \left|e^{i y_{t}\omega^*} - e^{i y_{q}\omega^*} \right|\prod_{q=1,\cdots, n}\left|e^{iy_{t}\omega^*} - e^{i\what y_q\omega^*}\right|< \frac{2^{2n}\sigma}{m_{\min}}, \ t=1,\cdots, n. 
\end{equation}
Therefore, it follows that 
\[
\min_{t=1,\cdots, n}\prod_{q=1,\cdots, n, q\neq t} \left|e^{i y_{t}\omega^*} - e^{i y_{q}\omega^*} \right|\max_{t=1,\cdots, n}\prod_{q=1,\cdots, n}\left|e^{iy_{t}\omega^*} - e^{i\what y_q\omega^*}\right|< \frac{2^{2n}\sigma}{m_{\min}}.
\]
Denote $y_q \omega^*, \what y_q\omega^*$ by respectively $\theta_q, \what \theta_q$ and $\theta_{\min} = \min_{p\neq q}|\theta_p-\theta_q|$. Since $y_j$'s in $B_{\frac{(n-1) \pi}{2\Omega}}^1(0)$ and $\omega^* = \frac{2\Omega}{2n-1}$, we have $\theta_j, \what \theta_{j}$'s in $\left[-\frac{\pi}{2}, \frac{\pi}{2}\right]$. By Lemma \ref{lem:multiproductlowerbound0}, we further get 
\begin{equation}\label{equ:proofsupport1}
\max_{j=1, \cdots, n}\prod_{q=1,\cdots,n}\left|e^{iy_{j}\omega^*} - e^{i \what y_q\omega^*}\right|<\frac{1}{\zeta(n)}\left(\frac{\pi}{2\theta_{\min}}\right)^{n-1} \frac{2^{2n}\sigma}{m_{\min}}. 
\end{equation}
We then utilize Lemma \ref{lem:stablemultiproduct1} to estimate the recovery of the locations. For this purpose, let  $\epsilon = \frac{\pi^{2n-1}}{\zeta(n)(\theta_{\min})^{n-1}} \frac{2\sigma}{m_{\min}}$. Then (\ref{equ:proofsupport1}) is equivalent to 
$$\bnorm{\eta_{n ,n}(e^{i \theta_1},\cdots,e^{i \theta_n}, e^{i \what \theta_1},\cdots, e^{i \what \theta_n})}_{\infty}<\left(\frac{2}{\pi}\right)^n\epsilon.$$
 We thus only need to check the following condition:
\begin{equation}\label{upperboundsupportlimithm1equ4}
\theta_{\min}\geqs \Big(\frac{4\epsilon}{\lambda(n)}\Big)^{\frac{1}{n}}, \quad \mbox{or equivalently}\,\,\, (\theta_{\min})^n \geqs \frac{4\epsilon}{\lambda(n)}.
\end{equation}
Indeed, by the separation condition (\ref{supportlimithm0equ0}), 
\begin{equation}\label{upperboundsupportlimithm1equ-1}
\theta_{\min}= d_{\min}\omega^*\geqs  \frac{2.36\pi e}{n-\frac{1}{2}}\Big(\frac{\sigma}{m_{\min}}\Big)^{\frac{1}{2n-1}}\geqs  \pi \Big(\frac{4}{\lambda(n)\zeta(n)}\frac{2\sigma}{m_{\min}}\Big)^{\frac{1}{2n-1}}, \end{equation}
where we have used Lemma \ref{lem:upperboundsupportcalculate1} in the last inequality.  
Then
\[
(\theta_{\min})^{2n-1}\geqs \frac{4\pi^{2n-1}}{\lambda(n)\zeta(n)}\frac{2\sigma}{m_{\min}},
\]
whence we get (\ref{upperboundsupportlimithm1equ4}). Therefore, we can apply Lemma \ref{lem:stablemultiproduct1} to get that, after reordering $\what \theta_j$'s,
\begin{equation} \label{equ:upperboundsupportlimithm1equ7}
\begin{aligned}
&\left|\what \theta_{j}-\theta_j\right|< \frac{\theta_{\min}}{2},\\
&\left|\what \theta_{j}-\theta_j\right|< \frac{2^{n}\pi^{2n-1}}{\zeta(n)(n-2)!(\theta_{\min})^{2n-2}}\frac{\sigma}{m_{\min}} , 1\leqs j\leqs n.
\end{aligned}
\end{equation}

Finally, we estimate $\left|\what y_j - y_j\right|$.  Since $\left|\what \theta_{j}-\theta_j\right|< \frac{\theta_{\min}}{2}$, it is clear that 
$\left|\what y_j-y_j\right|< \frac{d_{\min}}{2}.$
Thus $\what \mu$ is within the $\frac{d_{\min}}{2}$-neighborhood of $\mu$.
On the other hand, 
\[
\left|\what y_j-y_j\right|= \frac{2n-1}{2\Omega}\left|\what \theta_j -\theta_j\right|, \quad j =1, \cdots, n. 
\]	
Using (\ref{equ:upperboundsupportlimithm1equ7}) and Lemma \ref{lem:upperboundsupportcalculate2}, a direct calculation shows that 
\begin{align}\label{equ:proofsupport3}
\left|\what y_j-y_j\right|< \frac{C(n)}{\Omega} \left(\frac{\pi}{\Omega d_{\min}}\right)^{2n-2} \frac{\sigma}{m_{\min}}, 
\end{align}
where $C(n)=\sqrt{n-0.5}\ 2^{2n-\frac{3}{2}}e^{2n}4.5^{-1}$. 

The last part is to prove the case when $n=2$. When $n=2$, by (\ref{equ:proofsupport2}) we have 
\begin{align}\label{equ:proofsupport4}
\left|e^{i y_{1}\omega^*} - e^{i  \what y_{1}\omega^*} \right| \left|e^{iy_{1}\omega^*} - e^{i\what y_2\omega^*}\right| \left|e^{iy_{1}\omega^*} - e^{iy_2\omega^*}\right|  < \frac{2^{4}\sigma}{m_{\min}}.
\end{align}
Denote $\omega^*|y_1-y_2|=\theta_{\min}$. Reordering $\hat y_j$'s so that $|\what y_1- y_1|\leqs |\what y_2- y_1|$. Thus, if $|\what y_{1}-y_1|\omega^*\geqs \frac{\theta_{\min}}{2}$, we have $|\what y_{2}-y_1|\omega^*\geqs \frac{\theta_{\min}}{2}$. Recall also that $\what y_j\in B_{\frac{(n-1) \pi}{2\Omega}}^1(0), j=1,2$. Then (\ref{equ:proofsupport4}) yields
\[
\left(2\sin\left(\frac{\theta_{\min}}{4}\right)\right)^2 2\sin\left(\frac{\theta_{\min}}{2}\right) < \frac{2^4\sigma}{m_{\min}},
\]
which gives
\[
\sin^3\left(\frac{\theta_{\min}}{2}\right)< \frac{2^3\sigma}{m_{\min}}.
\]
It then follows that 
\[
\theta_{\min}< 2\arcsin\left(\frac{8\sigma}{m_{\min}}\right)^{\frac{1}{3}}
\]
and 
\[
d_{\min}=\frac{\theta_{\min}}{\omega^*}< \frac{3}{\Omega}\arcsin\left(2\left(\frac{\sigma}{m_{\min}}\right)^{\frac{1}{3}}\right),
\]
where we have set $\omega^* = \frac{2\Omega}{3}$. Therefore, if 
\[
d_{\min}\geqs \frac{3}{\Omega}\arcsin\left(2\left(\frac{\sigma}{m_{\min}}\right)^{\frac{1}{3}}\right),
\]
then we must have $|\what y_{1}-y_1|\omega^*<\frac{\theta_{\min}}{2}$ and consequently,  $|\what y_{1}-y_1|< \frac{d_{\min}}{2}$. In the same manner, we also have $|\what y_{2}-y_2|< \frac{d_{\min}}{2}$. This completes the proof. 
\end{proof}

\subsection{Stability of amplitude reconstruction}\label{section:amplituderecovery}
We now consider the stability of the amplitude reconstruction. Note that for the off-the-grid case, it takes several tens pages in \cite{batenkov2019super} to prove the stability of the reconstruction of each amplitude $a_j$. Here we can take two pages to have a stronger understanding for the amplitude reconstruction. 

\begin{thm}\label{thm:upperboundamplitudelimithm0}
Let $n\geqs 2$ and let the measurement $\vect Y$ be generated from any one-dimensional source $\mu=\sum_{j=1}^{n}a_j \delta_{y_j}$ supported on $B_{\frac{(n-1) \pi}{2\Omega}}^1(0)$ and satisfying the separation condition 
\begin{equation}\label{equ:sepaconditioninamplitude}
	d_{\min}=\min_{p\neq j}\left|y_p-y_j\right|\geqs \frac{2.36e\pi}{\Omega }\Big(\frac{\sigma}{m_{\min}}\Big)^{\frac{1}{2n-1}}.
 \end{equation}
For any $\sigma$-admissible measure of $\vect Y$, $\what \mu = \sum_{j=1}^n\what a_j \delta_{\what y_j}, \what y_j \in B_{\frac{(n-1) \pi}{2\Omega}}^1(0)$, after reordering the $\what y_j$'s, we have 
 \begin{equation}\label{equ:amplitudeequ1}
\babs{\what y_j - y_j} < \frac{d_{\min}}{9}, \quad \babs{\what a_j - a_j} < C_1(n)\mathrm{SRF}^{2n-1}\sigma,
 \end{equation}
 for a certain constant $C_1(n)$. Moreover, if $\what y_j = y_j$, we have 
 \begin{equation}\label{equ:amplitudeequ2}
     \babs{\what a_j - a_j} < C_2(n)\mathrm{SRF}^{2n-2}\sigma,
 \end{equation}
for a certain constant $C_2(n)$.
\end{thm}
\begin{proof}
By Theorem \ref{thm:upperboundsupportlimithm0}, the separation condition (\ref{equ:sepaconditioninamplitude}) implies 
\begin{equation}
	\Big|\what y_j-y_j\Big|<\frac{\sqrt{n-0.5}\ 2^{2n-\frac{3}{2}}e^{2n}4.5^{-1}}{\Omega}\left(\frac{\pi}{\Omega d_{\min}}\right)^{2n-2}\frac{\sigma}{m_{\min}}, \quad 1\leqs j\leqs n.
\end{equation}
Together with (\ref{equ:sepaconditioninamplitude}), this gives $|\what y_j-y_j|<\frac{d_{\min}}{9}, j=1,\cdots, n$. Hence, among $\{\what y_q\}_{q=1}^n$, $\what y_j$ is the closest point to $y_j$. We write $y_{j} = \what y_j+\epsilon_j$ with $0\leqs 
\babs{\epsilon_{j}}< \frac{d_{\min}}{9}$. 
Since $\what \mu=\sum_{j=1}^{n}\what a_j \delta_{\what y_j}, \what y_j \in B_{\frac{(n-1) \pi}{2\Omega}}^1(0)$ is a $\sigma$-admissible measure of $\vect Y$, from the model (\ref{equ:highdmodelsetting1}) we have 
\[
\mathcal F[\what \mu](\omega)= \mathcal F[\mu](\omega)+\vect W_1(\omega),\quad  \omega \in [0,\Omega], 
\]
for some $\vect W_1$ with $|\vect W_1(\omega)|<2\sigma, \omega \in [0,\Omega] $. For each $j\in \{1,\cdots, n\}$, denote by $S_j$ the set containing all $y_p$'s and $\what y_p$'s except $y_j, \what y_{j}$. By $\babs{\what y_p - y_p}<\frac{d_{\min}}{9}, 1\leqs p\leqs n$, for all $q\in S_j$ we have $q\neq y_j$ and $q\neq \what y_j$. Let $\omega^* = \frac{\Omega}{2n-1}$. As $y_p$'s and $\what y_p$'s are in $B_{\frac{(n-1) \pi}{2\Omega}}^1(0)$, we also have $e^{iq\omega^*}\neq e^{iy_j\omega^*}, e^{iq\omega^*}\neq e^{i\what y_j\omega^*}, q\in S_j$. Thus by (\ref{equ:locaintenrelation3}) we can have 
\begin{align}\label{equ:proofamplitude0}
\left|\what {a}_{j}\prod_{q\in S_j}\frac{e^{i\what y_{j}\omega^*}-e^{i q\omega^*}}{e^{i y_j\omega^*}- e^{i q\omega^*}}-a_j\right|
<  \frac{2^{2n-1}\sigma}{\prod_{q\in S_j}\left|e^{iy_j\omega^*} - e^{iq\omega^*}\right|}.
\end{align}
Equivalently, we have
\begin{align}\label{equ:proofamplitude0.5}
\left|\what {a}_{j}-a_j\prod_{q\in S_j}\frac{e^{iy_{j}\omega^*}-e^{i q\omega^*}}{e^{i \what y_j\omega^*}- e^{i q\omega^*}}\right|
<  \frac{2^{2n-1}\sigma}{\prod_{q\in S_j}\left|e^{i\what y_j\omega^*} - e^{iq\omega^*}\right|}.
\end{align}
We rewrite its LHS as
\begin{align}\label{equ:proofamplitude1}
\left|\what {a}_{j}-a_j\prod_{q\in S_j}\left(\frac{e^{i y_{j}\omega^*}- e^{i \what y_{j}\omega^*}}{e^{i \what y_j\omega^*}- e^{i q\omega^*}}+1\right)\right|,
\end{align}
and expand that
\begin{equation}\label{equ:proofamplitude4}
\prod_{q\in S_j}\left(\frac{e^{i y_{j}\omega^*}- e^{i \what y_{j}\omega^*}}{e^{i \what y_j\omega^*}- e^{i q\omega^*}}+1\right) = 1+ (e^{i y_{j}\omega^*}- e^{i \what y_{j}\omega^*})g(\epsilon_j, S_j, \what y_j, y_j, \omega^*)
\end{equation}
 where
 \begin{equation}\label{equ:proofamplitude3}
 \babs{g(\epsilon_j, S_j, \what y_j, y_j, \omega^*)} \leqs \frac{C_3(n)}{d_{\min}}
 \end{equation}
 for a certain constant $C_{3}(n)$. Thus combining (\ref{equ:proofamplitude0.5}), (\ref{equ:proofamplitude1}), and (\ref{equ:proofamplitude4}) yields that 
 \begin{equation}\label{equ:proofamplitude5}
\left|\what {a}_{j}-a_j\right| < \left|a_j (e^{i y_{j}\omega^*}- e^{i \what y_{j}\omega^*}) g(\epsilon_j, S_j, \what y_j, y_j, \omega^*)\right|
+  \frac{2^{2n-1}\sigma}{\prod_{q\in S_j}\left|e^{i\what y_j\omega^*} - e^{iq\omega^*}\right|}.
\end{equation}
Now we estimate the two terms in the RHS of the above equation. First, by (\ref{equ:locaintenrelation4}), we have  
 \begin{equation}\label{equ:proofamplitude6}
\left|a_j (e^{i y_{j}\omega^*}- e^{i \what y_{j}\omega^*})\right| < \frac{2^{2n-1}\sigma}{\prod_{q\in S_j}\left|e^{iy_j\omega^*} - e^{iq\omega^*}\right|}.  
\end{equation}
Second, based on the estimate $|\what y_p -y_p|< \frac{d_{\min}}{9}, p=1,\cdots, n$,  it is easy to prove that
\begin{equation}\label{equ:proofamplitude2}
\frac{\sigma}{\prod_{q\in S_j}\left|e^{i\what y_j\omega^*} - e^{iq\omega^*}\right|}\leqs \frac{C_4(n)\sigma}{d_{\min}^{2n-2}}
\end{equation}
holds for some constant $C_4(n)$. Combining estimates (\ref{equ:proofamplitude3}), (\ref{equ:proofamplitude5}), (\ref{equ:proofamplitude6}), and (\ref{equ:proofamplitude2}) yields 
\[
\left|\what {a}_{j}-a_j\right|< \frac{C_5(n)}{d_{\min}^{2n-1}}\sigma 
\]
for some constant $C_5(n)$. Now we consider the case when $\what y_j = y_j$. This time, by (\ref{equ:proofamplitude0}), we have 
\[
\babs{\what a_j - a_j}<\frac{2^{\#S_j}\sigma}{\prod_{q\in S_j}\left|e^{iy_j\omega^*} - e^{iq\omega^*}\right|}.
\]
Together with (\ref{equ:proofamplitude2}), we demonstrate (\ref{equ:amplitudeequ2}) and complete the proof. 
\end{proof}

\subsection{Stability of a sparsity-promoting algorithm}
Nowadays, sparsity-promoting algorithms are popular methods in image processing, signal processing, and many other fields. As a direct consequence of results in the above sections, we derive a sharp stability result for the following $l_0$-minimization problem in the one-dimensional super-resolution:
\begin{equation}\label{equ:l0minimization}
\min_{\text{$\rho$ supported on $\mathcal O$, $\rho$ is a discrete measure}} \bnorm{\rho}_{0} \quad \text{subject to} \quad |\mathcal F\rho(\omega) -\vect Y(\omega)|< \sigma, \quad \omega \in[-\Omega,\Omega], 
\end{equation}	
where $||\rho||_{0}$ is the number of Dirac masses representing the discrete measure $\rho$. 

\begin{thm}\label{thm:sparspromstabilitythm0}
	Let $n\geqs 2$ and $\sigma \leqs m_{\min}$. Let the measurement $\vect Y$ in (\ref{equ:highdmodelsetting1}) be generated by a one-dimensional source $\mu=\sum_{j=1}^{n}a_j \delta_{y_j}, a_j \in \mathbb C, y_j \in B_{\frac{(n-1) \pi}{2\Omega}}^1(0)$. Assume that
	\begin{equation}\label{equ:sparspromstabilitythm0equ1}
	d_{\min}=\min_{p\neq j}\left| y_p-y_j\right|\geqs \frac{2.36 e\pi }{\Omega }\Big(\frac{\sigma}{m_{\min}}\Big)^{\frac{1}{2n-1}}. \end{equation}
	Let $\mathcal O$ in the minimization problem (\ref{equ:l0minimization}) be (or be included in) $B_{\frac{(n-1) \pi}{2\Omega}}^1(0)$ , then the solution to (\ref{equ:l0minimization}) contains exactly $n$ point sources. For any solution $\what \mu=\sum_{j=1}^{n}\what a_{j}\delta_{\what y_j}$, it is in a $\frac{d_{\min}}{9}$-neighborhood of $\mu$. Moreover, after reordering the ${\what y}_j$'s, we have 
	\begin{equation}
	\left|\what y_j-y_j\right|< \frac{C_1(n)}{\Omega}\mathrm{SRF}^{2n-2}\frac{\sigma}{m_{\min}}, \quad \babs{\what a_j - a_j} < C_2(n)\mathrm{SRF}^{2n-1}\sigma, \quad 1\leqs j\leqs n,
	\end{equation}
for certain constants $C_1(n)$ and $C_2(n)$. 
\end{thm}

Theorem \ref{thm:sparspromstabilitythm0} reveals that sparsity promoting over admissible solutions can resolve the source locations to the resolution limit level. Particularly, under the separation condition (\ref{equ:sparspromstabilitythm0equ1}), any tractable sparsity-promoting algorithms (such as total variation minimization algorithms \cite{candes2014towards}) rendering the sparsest solution
could stably reconstruct all the source locations and amplitudes.

\section{Two-point resolution in the multi-dimensional super-resolution}\label{section:twopointlimit}
Now we have understood the stability of super-resolving multiple point sources. Specifically, we've shown that with a $\mathrm SNR$ greater than 4, super-resolution is definitively achievable for resolving two point sources. However, we are not merely interested in estimations; we aim to determine the precise resolution limit for distinguishing two sources. In this section, we will develop the exact formula for this two-point resolution limit. We will particularly address the multi-dimensional imaging problem as described by (\ref{equ:highdmodelsetting1}), focusing on super-resolving two point sources, represented as $\mu=\sum_{j=1}^{2}a_{j}\delta_{\vect y_j}, \vect y_j \in \mathbb R^k$.

\subsection{Two-point resolution for resolving sources with identical amplitudes}

Inspired by the classic diffraction limit problem, we derived the following lemma for the resolution limit when resolving two point sources with identical amplitudes. 

\begin{lem} \label{thm:twopointresolution0}
Let $\frac{\sigma}{m_{\min}} \leqs \frac{1}{2}$. For all measures $\mu =\sum_{j=1}^{2}a_j\delta_{\vect y_j}, \vect y_j\in \mathbb R^{k}$ with $m_{\min}=a_1=a_2>0$, if 
\begin{equation}\label{equ:twopointresoformula1}
 \min_{p\neq j} \bnorm{\vect y_j-\vect y_p}_2 \geqs  \frac{4\arcsin\left(\left(\frac{\sigma}{m_{\min}}\right)^{\frac{1}{2}}\right)}{\Omega},
\end{equation}
then there does not exist any $\sigma$-admissible measure of \, $\mathbf Y$ with less than two supports. On the other hand, when (\ref{equ:twopointresoformula1}) fails to hold, there exists a $\sigma$-admissible measure of some $\vect Y$ with only one point source. When $\frac{\sigma}{m_{\min}} >\frac{1}{2}$, no matter what the separation distance is, there are always some $\sigma$-admissible measures of some $\vect Y$ with only one point source.
\end{lem}
\begin{proof}
\textbf{Step 1.} We first prove the one-dimensional case. Let $\mu = \sum_{j=1}^2a_j \delta_{y_j}$ and $\what \mu = a \delta_{\what y}$. A crucial relation is
\begin{equation}\label{equ:prooftwopointlimit1}
\babs{\mathcal F[\what \mu](\omega)- \mathcal F[\mu](\omega) }< 2\sigma, \ \omega \in [-\Omega, \Omega].
\end{equation}
Note that if (\ref{equ:prooftwopointlimit1}) holds, $\what \mu$ can be a $\sigma$-admissible measure of some $\vect Y$ generated by model (\ref{equ:highdmodelsetting1}). This time, resolving two point sources is impossible. Conversely, if (\ref{equ:prooftwopointlimit1}) does not hold, $\what \mu$ cannot be any $\sigma$-admissible measure of some $\vect Y$ generated by $\mu$ as in model (\ref{equ:highdmodelsetting1}). Let $\mathcal R$  be the constant such that (\ref{equ:prooftwopointlimit1}) holds when $|y_1-y_2|< \mathcal R$ and fails to hold in the opposite case. Based on the above discussions, proving Lemma \ref{thm:twopointresolution0} is to show $\mathcal R = \frac{4\arcsin\left(\left(\frac{\sigma}{m_{\min}}\right)^{\frac{1}{2}}\right)}{\Omega}$.

Instead of considering all the $\omega \in [-\Omega,\Omega]$ directly, we consider
\begin{equation}\label{equ:prooftwopointlimit2}
\babs{\mathcal F[\what \mu](\omega)-\mathcal F[\mu](\omega)}< 2\sigma, \ \omega \in [0, \Omega].
\end{equation}
In the sequel, we intend to find $\mathcal R$ so that (\ref{equ:prooftwopointlimit2}) holds when $|y_1-y_2|< \mathcal R$ and does not hold in the opposite case. Afterward, we will show that (\ref{equ:prooftwopointlimit1}) holds as well under some circumstances when $|y_1-y_2|< \mathcal R$.  

\textbf{Step 2.} Note that for the general source locations $y_1, y_2$, shifting them by $x$ and get that
\[
\babs{\mathcal F[\what \mu](\omega)e^{ix\omega}- \mathcal F[\mu](\omega)e^{ix\omega} }< 2\sigma, \ \omega \in [-\Omega, \Omega],
\]
we can transform the problem into the case when $y_1=-y_2$. Thus we consider that the true source is $\mu = m_{\min} \delta_{y_1}+ m_{\min} \delta_{y_2}$ with $y_1>0, y_1 = -y_2$. The measure $\what \mu$ is $a\delta_{\what y}$ with $a$ and $\what y$ to be determined.  

Then the existence of $\what \mu$ satisfying (\ref{equ:prooftwopointlimit2}) is equivalent to solving the condition on $y_1$ so that
\[
\min_{a\in \mathbb C, \hat y\in \mathbb R}\babs{ae^{i\what 
 y\omega}-2m_{\min}\cos\left(y_1\omega\right)}<2\sigma, \ \omega \in [0, \Omega],
\]
where we have used 
\begin{align*}
ae^{i \what y\omega}-m_{\min}(e^{iy_1\omega}+e^{i y_2\omega})= ae^{i\what 
 y\omega}-2m_{\min}\cos\left(y_1\omega\right). 
\end{align*}
We denote $d_{\min}:=|y_1-y_2|$ and first consider the case when $0<d_{\min}<\frac{\pi}{\Omega}$. Note that for two non-negative values $x,y$, we have 
\[
\left|xe^{i\theta}-y\right|^2 = (x\cos\theta-y)^2+x^2\sin^2\theta = x^2+y^2-2xy\cos\theta \geqs (x-y)^2 
\]
and the equality is attained when $\theta=0$. Since $0<d_{\min}\leqs \frac{\pi}{\Omega}$ $\left(0<y_1\leqs \frac{\pi}{2\Omega}\right)$, we have $\cos(y_1\omega)\geqs 0, \omega \in [0,\Omega]$. Thus for every $\omega$, 
\[
\min_{a\in \mathbb C, \hat y\in \mathbb R}\babs{ae^{i\what 
 y\omega}-2m_{\min}\cos\left(y_1\omega\right)}\geqs \left||a|-2m_{\min}\cos\left(y_1\omega\right)\right|
\]
and the minimum is attained when $\what y=0$ and $a$ is a positive number. We now try to find the condition on $y_1$ so that there exists $a\in \mathbb R^+$ satisfying
\[
\left|a-2m_{\min}\cos\left(y_1\omega\right)\right|<2\sigma, \quad \omega\in [0,\Omega]. 
\]
This is equivalent to
\begin{equation}\label{equ:prooftwopointlimit3}
\max_{ \omega, \omega' \in [0,\Omega]}\left|2m_{\min}\left(\cos\left(y_1\omega\right)- \cos\left(y_1\omega'\right)\right)\right|<4\sigma,
\end{equation}
since the existence of $\omega_1, \omega_2 \in [0,\Omega]$ so that   
\[
\left|2m_{\min}\left(\cos\left(y_1\omega_1\right)- \cos\left(y_1\omega_2\right)\right)\right|\geqs 4\sigma,
\]
results in 
\[
\max_{\omega \in [0,\Omega]} \left|a-2m_{\min}\cos\left(y_1\omega\right)\right|\geqs 2\sigma,\quad \forall a\in \mathbb R^+.
\]

If $d_{\min}=2y_1\leqs \frac{\pi}{\Omega}$, then $0\leqs y_1\omega\leqs \frac{\pi}{2}, \omega \in [0,\Omega]$. Then problem (\ref{equ:prooftwopointlimit3}) becomes 
\[
2m_{\min}\left|1-\cos\left(\frac{d_{\min}}{2}\Omega\right)\right|< 4\sigma. 
\]
Thus, $4\sin^2\left(\frac{d_{\min}\Omega}{4}\right)<\frac{4\sigma}{m_{\min}}$, and equivalently
\[
d_{\min} < \frac{4\arcsin\left(\left(\frac{\sigma}{m_{\min}}\right)^{\frac{1}{2}}\right)}{\Omega}.
\]
Note that when $d_{\min} < \frac{4\arcsin\left(\left(\frac{\sigma}{m_{\min}}\right)^{\frac{1}{2}}\right)}{\Omega}$, choosing $a= \frac{m_{\min}+m_{\min}\cos(y_1\Omega)}{2}$ and $\what y=0$ makes 
\[
\babs{\mathcal F[\what \mu](\omega)-\mathcal F[\mu](\omega)}< 2\sigma, \ \omega \in [0, \Omega].
\]
As $\mathcal F[\what \mu](-\omega)-\mathcal F[\mu](-\omega) = \mathcal F[\what \mu](\omega)-\mathcal F[\mu](\omega)$ this time, the solution makes 
\[
\babs{\mathcal F[\what \mu](\omega)-\mathcal F[\mu](\omega)}< 2\sigma, \ \omega \in [-\Omega, \Omega].
\]
Thus this is exactly the resolution limit $\mathcal R$ when $0<d_{\min}\leqs \frac{\pi}{\Omega}$ for the case when $\frac{\sigma}{m_{\min}}\leqs \frac{1}{2}$. On the other hand, for the case when $d_{\min}> \frac{\pi}{\Omega}$, by $\frac{4\arcsin\left(\left(\frac{\sigma}{m_{\min}}\right)^{\frac{1}{2}}\right)}{\Omega}\leqs \frac{\pi}{\Omega}$ when $\frac{\sigma}{m_{\min}}\leqs \frac{1}{2}$ and above discussions, we have 
\[
\min_{a\in \mathbb C, \hat y\in \mathbb R} \max_{\omega\in[0,\Omega]}\babs{ae^{i\what 
 y\omega}-2m_{\min}\cos\left(y_1\omega\right)}\geqs 2\sigma. 
\]
Thus there does not exist any $\sigma$-admissible measure of \, $\mathbf Y$ with less than two supports. This proves the statements for the case when $\frac{\sigma}{m_{\min}}\leqs\frac{1}{2}$ in the lemma.

Now, we consider the case when $\frac{\sigma}{m_{\min}}>\frac{1}{2}$. 
We choose a specific case where $a=m_{\min}$ and $\what y =y_1$. Then 
\[
\babs{m_{\min}e^{iy_1\omega}-m_{\min}(e^{iy_1\omega}+e^{i y_2\omega})} = \babs{m_{\min}e^{iy_2\omega}} < 2\sigma
\]
gives
\[
\babs{\mathcal F[\what \mu](\omega)-\mathcal F[\mu](\omega)}<2\sigma, \ \omega \in \mathbb R.
\]
Thus the case when $\frac{\sigma}{m_{\min}}>\frac{1}{2}$ is meaningless. There are always some $\sigma$-admissible measures for some images with only one point source. 

\textbf{Step 3.} Now we consider the case when the  sources $\vect y_j$'s are in $\mathbb R^k$. We still consider the crucial relation that
\begin{equation}\label{equ:prooftwopointlimit4}
\babs{\mathcal F[\what \mu](\vect \omega)- \mathcal F[\mu](\vect \omega) }< 2\sigma, \ \bnorm{\vect \omega}_{2}\leqs \Omega. 
\end{equation}
By a similar argument as the one in step 1, we need to compute constant $\mathcal R$ such that (\ref{equ:prooftwopointlimit4}) holds when $\bnorm{\vect y_1-\vect y_2}_2< \mathcal R$ and fails to hold in the opposite case. Note that by choosing suitable axes or transforming the problem, we can make $\vect y_1= (y_1, 0, \cdots, 0)^\top, \vect y_2= (y_2, 0, \cdots, 0)^\top$. Consider $\what \mu = a \delta_{\mathbf{\what y}}, \mathbf{\what y} \in \mathbb R^k$ with $a$ and $\mathbf{\what y}$ to be determined. We now have 
\[
\mathcal F[\what \mu](\vect \omega)- \mathcal F[\mu](\vect \omega) = ae^{i \mathbf{\what y} \cdot \vect \omega} - \sum_{j=1}^{2}a_j e^{i \vect {y}_j \cdot \vect{\omega}}= ae^{i \mathbf{\what y}_{2:k}\cdot \vect \omega_{2:k}} e^{\mathbf{\what y}_1 \vect \omega_1} - \sum_{j=1}^{2}a_j e^{i {y}_j \vect{\omega}_1},
\]
where $\vect x_1$ is the first vector element and $\vect x_{2:k}$ is the vector consisting of the $2$-nd to the $k$-th elements of $\vect x$. Thus analyzing when (\ref{equ:prooftwopointlimit4}) holds can be reduced to the one-dimensional case and it is not hard to see that the result for the one-dimensional space still holds for multi-dimensional spaces.

\end{proof}

\subsection{Resolution limit for detecting twp positive sources}
Now, we compute the computational resolution limit for detecting two positive sources. We have the following theorem showing that the resolution limit is the one in Lemma \ref{thm:twopointresolution0}. 
\begin{thm}\label{thm:positivetwopointresolution0}
For $\frac{\sigma}{m_{\min}} \leqs \frac{1}{2}$, the computational resolution limit $\mathcal D_{num}^+(k,2)$ for resolving two positive sources in $\mathbb R^k$ is given by
\[
 \mathcal D_{num}^+(k,2) =  \frac{4\arcsin\left(\left(\frac{\sigma}{m_{\min}}\right)^{\frac{1}{2}}\right)}{\Omega}.
\]
It can be attained if $a_1=a_2$. When $\frac{\sigma}{m_{\min}} >\frac{1}{2}$, no matter what the separation distance is, there are always some $\sigma$-admissible measures of some $\vect Y$ with only one point source. 
\end{thm}

 When $\frac{\sigma}{m_{\min}}<\frac{1}{2}$, the two-point resolution $\mathcal D_{k, num}^+(k,2)$ is already less than the Rayleigh limit $\frac{\pi}{\Omega}$, which far exceeds all expectations. This indicates that, in contrast to what was commonly supposed, super-resolution from a single snapshot is in fact very possible. 

\begin{remark}
Although $\mathcal D_{num}^{+}(k,2)$ is defined for measures on $B_{\frac{(n-1) \pi}{2\Omega}}^k(\mathbf{0})$, similar arguments as those in the proof of Lemma \ref{thm:twopointresolution0} can generalize the resolution estimate to measures on $\mathbb R^k$. 
\end{remark}

\medskip
Now we introduce the proof.
\begin{proof}
\textbf{Step 1.}
We only need to consider the case when $\frac{\sigma}{m_{\min}} \leqs \frac{1}{2}$, as the case when $\frac{\sigma}{m_{\min}} >\frac{1}{2}$ is trivial. Also, we only consider the one-dimensional case since the treatment for multi-dimensional spaces is similar to the one in the proof of Theorem \ref{thm:twopointresolution0}. 

Let $\mu = \sum_{j=1}^2a_j \delta_{y_j}$ and $\what \mu = a \delta_{\what y}$. Similarly to step 1 in the proof of Theorem \ref{thm:twopointresolution0}, the computational resolution limit $\mathcal D_{num}^+(k, 2)$ should be the constant such that the following 
\begin{equation}\label{equ:proofpositivetwopointlimit0}
\babs{\mathcal F[\what \mu](\omega)- \mathcal F[\mu](\omega)}< 2\sigma, \ \omega \in [-\Omega, \Omega],
\end{equation}
holds when $|y_1-y_2|< \mathcal D_{k, num}^+$ and fails to hold in the opposite case. Choosing a suitable axis, we assume that the true source is $\mu = m_{\min}\alpha \delta_{y_1}+ m_{\min} \delta_{y_2}$ with $y_1 = -y_2$, $\alpha \geqs 1$.  We consider $\what \mu = a\delta_{\what y}$ with $a>0$ and $\what y$ to be determined. We shall prove that if 
\[
\left|y_1-y_2\right|\geqs \frac{4\arcsin\left(\left(\frac{\sigma}{m_{\min}}\right)^{\frac{1}{2}}\right)}{\Omega},
\]
then (\ref{equ:proofpositivetwopointlimit0}) doesn't hold for any $\what \mu$ consisting of only one positive source. In the opposite case, Theorem \ref{thm:twopointresolution0} already ensures the existence of such $\what \mu, \mu$ making (\ref{equ:proofpositivetwopointlimit0}) holds. By the above two results, we prove the theorem.  

In the following proof, we will find a necessary condition for  
\[\min_{a>0, \alpha\geqs 1, \what y\in \mathbb R}\babs{\mathcal F[\what \mu](\omega)- \mathcal F[\mu](\omega)}<2\sigma, \omega \in [-\Omega,\Omega].
\]
By the definition of $\mathcal D_{num}^+(k, 2)$ in Definition \ref{defi:highdcomputresolimit}, we only consider the case when $d_{\min}\leqs \frac{\pi}{\Omega}$.  

\textbf{Step 2.} We analyze a necessary condition, that is, 
\begin{align}\label{equ:proofpositivetwopointlimit1.2}
\min_{a>0, \alpha\geqs 1, \what y\in \mathbb R}\babs{\mathcal F[\what \mu](\Omega)- \mathcal F[\mu](\Omega)}+\babs{\mathcal F[\what \mu](0)- \mathcal F[\mu](0)}<4\sigma.
\end{align}
Since from (\ref{equ:proofpositivetwopointlimit0}) and the assumption for $\mu$, 
\begin{equation}\label{equ:proofpositivetwopointlimit1.1}
\mathcal F[\what \mu](\omega)- \mathcal F[\mu](\omega) = ae^{i \what y\omega}-m_{\min}\left(\alpha e^{iy_1\omega}+ e^{-iy_1\omega}\right),
\end{equation}
we thus consider
\begin{align}\label{equ:proofpositivetwopointlimit1.3}
\min_{a>0, \alpha\geqs 1, \what y\in \mathbb R}\left|ae^{i\what y \Omega} -m_{\min}\left(\alpha e^{iy_1 \Omega}+ e^{-i y_1\Omega}\right)\right|+\left|a -m_{\min}\left(\alpha+1\right)\right|.
\end{align}
% Note that at the minimum value $a\leqs m_{\min}(\alpha+1)$. The above equation is 
% \[
% \min_{a, \alpha, \what y}\left|a -m_{\min}\left(\alpha e^{i(y_1-\what y) 2\Omega}+ e^{-i (y_1+\what y) 2\Omega}\right)\right|+m_{\min}\left(\alpha+1\right)-a.
% \]
% Rewrite it as 
% \[
% \min_{a, \alpha, \what y}\left|ae^{i\what y2\Omega} -m_{\min}\left(\alpha e^{iy_1 2\Omega}+ e^{-iy_1 2\Omega}\right)\right|+m_{\min}\left(\alpha+1\right)-a,
% \]
Let $\alpha = 1+h, h\geqs 0$, and rewrite the above formula as 
\[
\min_{a>0, h\geqs 0, \what y\in \mathbb R}\left|ae^{i\what y\Omega} -h m_{\min} e^{iy_1 \Omega}- 2m_{\min} \cos(y_1 \Omega)\right|+|a-(2m_{\min}+hm_{\min})|.
\]
A key observation is that if $2m_{\min}\cos(y_1\Omega)+hm_{\min}\leqs a \leqs 2m_{\min}+hm_{\min}$, then we have 
\begin{align}
&\min_{a>0, h\geqs 0, \what y\in \mathbb R}\left|ae^{i\what y\Omega} -h m_{\min} e^{iy_1 \Omega}- 2m_{\min} \cos(y_1 \Omega)\right|+|a-(2m_{\min}+hm_{\min})|\nonumber \\
\geqs &\min_{a>0, h\geqs 0, \what y\in \mathbb R, \what x\in \mathbb R}\left|ae^{i\what y\Omega} -h m_{\min} e^{i\what x \Omega}- 2m_{\min} \cos(y_1 \Omega)\right|+|a-(2m_{\min}+hm_{\min})|\nonumber \\
= &\min_{a> 0, h\geqs 0}\left|a -h m_{\min}- 2m_{\min} \cos(y_1 \Omega)\right|+|a-(2m_{\min}+hm_{\min})|\nonumber \\
= &\min_{2m_{\min}\cos(y_1\Omega)\leqs b \leqs 2 m_{\min}}\left|b- 2m_{\min} \cos(y_1 \Omega)\right|+|2m_{\min}-b|  \nonumber\\
= & 2m_{\min} - 2m_{\min} \cos(y_1 \Omega),
\label{equ:proofpositivetwopointlimit7}
\end{align}
where the second equality is because $a\geqs 2m_{\min}\cos(y_1\Omega)+hm_{\min}$, $hm_{\min}\geqs 0$ and $2m_{\min}\cos(y_1\Omega)= 2m_{\min}\cos(\frac{d_{\min}}{2}\Omega)\geqs 0$ for $d_{\min}\leqs \frac{\pi}{\Omega}$. 

% Now we analyze (\ref{equ:proofpositivetwopointlimit7}). If $2m_{\min} \cos(y_1 \Omega)\leqs b\leqs 2m_{\min}$,
% \begin{align*}
% &\min_{2m_{\min} \cos(y_1 \Omega)\leqs b\leqs 2m_{\min}}\left|b- 2m_{\min} \cos(y_1 \Omega)\right|+|2m_{\min}-b| = 2m_{\min}-2m_{\min} \cos(y_1 \Omega).
% \end{align*}
% If $b> 2m_{\min}$, 
% \begin{align*}
% & \min_{b> 2m_{\min}}\left|b- 2m_{\min} \cos(y_1 \Omega)\right|+|2m_{\min}-b|\\
% =&\min_{b> 2m_{\min}}2b- 2m_{\min} \cos(y_1 \Omega)-2m_{\min}\\
% > &2m_{\min}-2m_{\min} \cos(y_1 \Omega).
% \end{align*}
% Therefore,
% \begin{equation}\label{equ:proofpositivetwopointlimit8}
% \min_{b\in \mathbb R, \ b\geqs 2m_{\min}\cos(y_12\Omega)}\left|b- 2m_{\min} \cos(y_1 \Omega)\right|+|2m_{\min}-b| \geqs 2m_{\min}-2m_{\min} \cos(y_1 \Omega).
% \end{equation}
On the other hand, letting $h=0, \what y =0$, we obtain that
\begin{align*}
&\min_{a>0}\left|ae^{i\what y\Omega} -h m_{\min} e^{iy_1 \Omega}- 2m_{\min} \cos(y_1 \Omega)\right|+|a-(2m_{\min}+hm_{\min})|\\
=&\min_{a>0}\left|a - 2m_{\min} \cos(y_1 \Omega)\right|+|a-2m_{\min}|\\
=&\min_{a>0}(a - 2m_{\min} \cos(y_1 \Omega))+2m_{\min}-a \quad \big(\text{choose $2m_{\min} \cos(y_1 \Omega)\leqs a\leqs 2m_{\min}$}\big)\\ 
=&2m_{\min}-2m_{\min} \cos(y_1 \Omega).
\end{align*}
Together with (\ref{equ:proofpositivetwopointlimit7}), this yields 
\begin{align*}
&\min_{a>0, h\geqs 0, \what y\in \mathbb R}\left|a e^{i\what y\Omega} -h m_{\min} e^{iy_1 \Omega}- 2m_{\min} \cos(y_1 \Omega)\right|+|a-(2m_{\min}+hm_{\min})|\\
=&2m_{\min}-2m_{\min} \cos(y_1 \Omega)
\end{align*}
in the case when $2m_{\min}\cos(y_1\Omega)+hm_{\min}\leqs a \leqs 2m_{\min}+hm_{\min}$.

Now we consider the case when $a< 2m_{\min}\cos(y_1\Omega)+hm_{\min}$. In this case, we have 
\begin{align*}
&\min_{a>0, h\geqs 0, a< 2m_{\min}\cos(y_1\Omega)+hm_{\min}, \what y\in \mathbb R}\left|ae^{i\what y\Omega} -h m_{\min} e^{iy_1 \Omega}- 2m_{\min} \cos(y_1 \Omega)\right|+|a-(2m_{\min}+hm_{\min})|\\
\geqs & \min_{a>0, h\geqs 0, a< 2m_{\min}\cos(y_1\Omega)+hm_{\min}}|a-(2m_{\min}+hm_{\min})|\\
=& \min_{a>0, h\geqs 0, a< 2m_{\min}\cos(y_1\Omega)+hm_{\min}}2m_{\min}+hm_{\min}-a\\
>& 2m_{\min} -2m_{\min}\cos(y_1 \Omega).
\end{align*}

Last, we consider the case when $a>2m_{\min}+hm_{\min}$. This time, we have 
\begin{align*}
&\min_{a>0, h\geqs 0, a>2m_{\min}+hm_{\min}, \what y\in \mathbb R}\left|ae^{i\what y\Omega} -h m_{\min} e^{iy_1 \Omega}- 2m_{\min} \cos(y_1 \Omega)\right|+|a-(2m_{\min}+hm_{\min})|\\
\geqs & \min_{a>0, h\geqs 0, a>2m_{\min}+hm_{\min}, \what y \in \mathbb R}\left|ae^{i\what y\Omega} -h m_{\min} e^{iy_1 \Omega}- 2m_{\min} \cos(y_1 \Omega)\right|\\
\geqs & \min_{a>0, h\geqs 0, a>2m_{\min}+hm_{\min}}a -h m_{\min}- 2m_{\min} \cos(y_1 \Omega)\\
>& 2m_{\min} -2m_{\min}\cos(y_1 \Omega).
\end{align*}

Therefore, combining the above discussions yields
\begin{align*}
&\min_{a\geqs 0, h\geqs 0, \what y\in \mathbb R}\left|ae^{i\what y\Omega} -h m_{\min} e^{iy_1 \Omega}- 2m_{\min} \cos(y_1 \Omega)\right|+|a-(2m_{\min}+hm_{\min})|\\
=&2m_{\min}-2m_{\min} \cos(y_1 \Omega),
\end{align*}
and i.e.,
\begin{align*}
&\min_{a\geqs 0, \alpha\geqs 1, \what y\in \mathbb R}\left|a e^{\what y\Omega}-m_{\min}\left(\alpha e^{iy_1\Omega}+ e^{-i y_1 \Omega}\right)\right|+\left|a -m_{\min}\left(\alpha+1\right)\right|\\
=&2m_{\min}-2m_{\min} \cos(y_1 \Omega).
\end{align*}
Thus, (\ref{equ:proofpositivetwopointlimit1.2}) is equivalent to 
\[
2m_{\min}-2m_{\min} \cos(y_1 \Omega)<4\sigma.
\]
Similarly to the proof of Theorem \ref{thm:twopointresolution0}, this shows that
\[
d_{\min} < \frac{4\arcsin\left(\left(\frac{\sigma}{m_{\min}}\right)^{\frac{1}{2}}\right)}{\Omega},
\]
and completes the proof. 
\end{proof}

\subsection{Resolution limit for detecting two complex sources}
Now we consider super-resolving complex sources and have the following theorem. 
\begin{thm}\label{thm:computatwopointresolution0}
For $\frac{\sigma}{m_{\min}} \leqs \frac{1}{2}$, the computational resolution limit $\mathcal D_{num}(k, 2)$ for resolving two sources in $\mathbb R^k$ is given by 
\begin{equation}\label{equ:computatwopointresolu0}
 \mathcal D_{num}(k, 2) =  \frac{4\arcsin\left(\left(\frac{\sigma}{m_{\min}}\right)^{\frac{1}{2}}\right)}{\Omega}.
\end{equation}
It can be attained if $a_1=a_2$. When $\frac{\sigma}{m_{\min}} >\frac{1}{2}$, no matter what the separation distance is, there are always some $\sigma$-admissible measures of some $\vect Y$ with only one point source.
\end{thm}

Theorem \ref{thm:computatwopointresolution0} demonstrates that when $\frac{\sigma}{m_{\min}}<\frac{1}{2}$, the two-point resolution for distinguishing general sources is already better than the Rayleigh limit. 

\medskip
We now prove the theorem.  
\begin{proof}
\textbf{Step 1.} We only need to analyze the case when $\frac{\sigma}{m_{\min}} \leqs \frac{1}{2}$, as the case when $\frac{\sigma}{m_{\min}} >\frac{1}{2}$ is trivial. Also, we only consider the one-dimensional case since the treatment for multi-dimensional spaces is similar to the one in the proof of Theorem \ref{thm:twopointresolution0}. 

Let $\mu = \sum_{j=1}^2a_j \delta_{y_j}$ and $\what \mu = a \delta_{\what y}$. Similarly to step 1 in the proof of Theorem \ref{thm:twopointresolution0}, the resolution limit $\mathcal D_{num}(k, 2)$ should be the constant such that the following estimate:
\begin{equation}\label{equ:proofcomplextwopointlimit0}
\babs{\mathcal F[\what \mu](\omega)- \mathcal F[\mu](\omega) }< 2\sigma, \ \omega \in [-\Omega, \Omega],
\end{equation}
holds when $|y_1-y_2|< \mathcal D_{num}(k,2)$ and fails to hold in the opposite case. 

\textbf{Step 2.} 
Without loss of generality, we assume the true source is $$\mu = m_{\min}\alpha e^{-i\beta} \delta_{y_1}+ m_{\min}e^{i\beta} \delta_{y_2}$$ with $y_1 = -y_2$, $0< y_1\leqs \frac{\pi}{2\Omega}$, $\alpha \geqs 1$ and $0\leqs \beta\leqs \frac{\pi}{2}$. It is not hard to see that the other cases can all be transformed to the above setting. We consider $\what \mu = ae^{i\gamma}\delta_{\what y}$ with $a>0$, $\gamma$ and $\what y$ to be determined. We shall prove that if 
\[
\left|y_1-y_2\right|\geqs \frac{4\arcsin\left(\left(\frac{\sigma}{m_{\min}}\right)^{\frac{1}{2}}\right)}{\Omega},
\]
then (\ref{equ:proofcomplextwopointlimit0}) does not hold for any $\what \mu$ consisting of only one source. In the opposite case, Theorem \ref{thm:twopointresolution0} already ensures the existence of such $\what \mu, \mu$ satisfying (\ref{equ:proofcomplextwopointlimit0}). By the two results, we prove the theorem. 

From (\ref{equ:proofcomplextwopointlimit0}), we have 
\begin{align*}
\mathcal F[\what \mu](\omega)- \mathcal F[\mu](\omega) &= ae^{i\gamma}e^{i \what y \omega}-m_{\min}\left(\alpha e^{-i\beta}e^{iy_1\omega}+ e^{i\beta}e^{-iy_1\omega}\right). 
\end{align*}
We rewrite it as 
\begin{equation}
 ae^{i\gamma}e^{i \what y \omega}-m_{\min}\left(\alpha e^{i(y_1 \omega-\beta)}+ e^{i(\beta-y_1 \omega)}\right).    
\end{equation}
and analyze it by considering the two cases: (1) $y_1\Omega\geqs\beta$; (2) $y_1\Omega<\beta$. \\

\textbf{Part 1:} ($y_1\Omega\geqs\beta$)\\
In the first case, when $y_1\Omega \geqs \beta$, we define $\omega^*=\frac{\beta}{y_1}\in[0,\Omega]$. Considering
\begin{align}
F[\what \mu](\omega+\omega^*)- \mathcal F[\mu](\omega+\omega^*)=&ae^{i\gamma}e^{i \what y (\omega+\omega^*)}-m_{\min}\left(\alpha e^{i(y_1 \omega-\beta+y_1\omega^*)}+ e^{i(\beta-y_1\omega^*-y_1 \omega)}\right)\nonumber \\
= &ae^{i\gamma+\what y\omega^*}e^{i \what y \omega}-m_{\min}\left(\alpha e^{i(y_1 \omega)}+ e^{i(-y_1 \omega)}\right),\label{equ:proofcomplextwopointlimit6}
\end{align}
(\ref{equ:proofcomplextwopointlimit0}) is equivalent to
\[
\left|ae^{i\gamma+\what y\omega^*}e^{i \what y \omega}-m_{\min}\left(\alpha e^{i(y_1 \omega)}+ e^{i(-y_1 \omega)}\right)\right|<2\sigma, \quad  \omega\in 
[-\Omega-\omega^*, \Omega-\omega^*].
\]
Note that this reduces the problem to a case similar to the one for positive sources. Since the interval $[-\Omega-\omega^*, \Omega-\omega^*]$ includes the interval $[-\Omega, 0]$, in the same fashion as the proof for positive sources, we consider the necessary condition that 
\begin{equation}\label{equ:proofcomplextwopointlimit9}
\min_{a>0, \alpha\geqs 1, \gamma \in \mathbb R,  \what y\in \mathbb R, 0\leqs \beta\leqs y_1\Omega}\babs{F[\what \mu](-\Omega+\omega^*)- \mathcal F[\mu](-\Omega+\omega^*)}+\babs{F[\what \mu](\omega^*)- \mathcal F[\mu](\omega^*)}<4\sigma.    
\end{equation}

Note that minimizing over $0\leqs \beta \leqs y_1\Omega$ is now equivalent to minimizing over $0\leqs \omega^*\leqs \Omega$.
We thus consider
\begin{align*}
\min_{a>0, \alpha\geqs 1, \gamma \in \mathbb R, \what y\in \mathbb R, \omega^*\in [0,\Omega]}\left|ae^{i\gamma+\what y\omega^*}e^{-i\what y \Omega} -m_{\min}\left(\alpha e^{-iy_1 \Omega}+ e^{i y_1\Omega}\right)\right|+\left|ae^{i\gamma+\what y\omega^*} -m_{\min}\left(\alpha+1\right)\right|.
\end{align*}
Let $\alpha = 1+h, h\geqs 0$, and rewrite the above formula as
\begin{align*}
&\min_{a>0, h\geqs 0, \gamma \in \mathbb R, \what y\in \mathbb R, \omega^*\in [0,\Omega]}\left|ae^{i\gamma+\what y\omega^*}e^{-i\what y \Omega} -h m_{\min} e^{-iy_1 \Omega}- 2m_{\min} \cos(y_1 \Omega)\right|\\
&\qquad \qquad \qquad \qquad \qquad+\left|ae^{i\gamma+\what y\omega^*}-(2m_{\min}+hm_{\min})\right|.
\end{align*}
A key observation is that if $2m_{\min}\cos(y_1\Omega)+hm_{\min}\leqs a\leqs 2m_{\min}+hm_{\min}$, we have 
\begin{align}
&\min_{a>0, h\geqs 0, \gamma \in \mathbb R, \what y\in \mathbb R, \omega^*\in [0,\Omega]}\left|ae^{i\gamma+\what y\omega^*}e^{-i\what y\Omega} -h m_{\min} e^{-iy_1 \Omega}- 2m_{\min} \cos(y_1 \Omega)\right|\nonumber \\
&\qquad \qquad \qquad \qquad \qquad +\left|ae^{i\gamma+\what y\omega^*}-(2m_{\min}+hm_{\min})\right|\nonumber \\
\geqs &\min_{a>0, h\geqs 0, \gamma \in \mathbb R, \what y\in \mathbb R, \what x\in \mathbb R}\left|ae^{-i\what y\Omega} -h m_{\min} e^{-i\what x \Omega}- 2m_{\min} \cos(y_1 \Omega)\right|+|ae^{i\gamma}-(2m_{\min}+hm_{\min})|\nonumber \\
= &\min_{a> 0, h\geqs 0}\left|a -h m_{\min}- 2m_{\min} \cos(y_1 \Omega)\right|+\left|a-(2m_{\min}+hm_{\min})\right|\nonumber \\
= &\min_{2m_{\min}\cos(y_1\Omega)\leqs b\leqs 2m_{\min}}\left|b- 2m_{\min} \cos(y_1 \Omega)\right|+|2m_{\min}-b|\nonumber \\
=&2m_{\min}-2m_{\min} \cos(y_1 \Omega), \label{equ:proofcomplextwopointlimit7}
\end{align}
where the second equality is because $2m_{\min}\cos(y_1\Omega)+hm_{\min}\leqs a\leqs 2m_{\min}+hm_{\min}$ and $2m_{\min}\cos(y_1\Omega)= 2m_{\min}\cos(\frac{d_{\min}}{2}\Omega)\geqs 0$ for $d_{\min}\leqs \frac{\pi}{\Omega}$. 

On the other hand, letting $h=0, \what y =0, \gamma =0, \omega^*=0$, we have 
\begin{align*}
&\min_{a>0}\left|ae^{i\gamma+\what y\omega^*}e^{-i\what y \Omega} -h m_{\min} e^{-iy_1 \Omega}- 2m_{\min} \cos(y_1 \Omega)\right|+\left|ae^{i\gamma+\what y\omega^*}-(2m_{\min}+hm_{\min})\right|.\\
=&\min_{a>0}\left|a - 2m_{\min} \cos(y_1 \Omega)\right|+|a-(2m_{\min}+hm_{\min})|\\
=&\min_{a>0}(a - 2m_{\min} \cos(y_1 \Omega))+2m_{\min}-a \quad \big(\text{choose $2m_{\min} \cos(y_1 \Omega)\leqs a\leqs 2m_{\min}$}\big)\\ 
=&2m_{\min}-2m_{\min} \cos(y_1 \Omega).
\end{align*}
Together with (\ref{equ:proofcomplextwopointlimit7}), this yields 
\begin{align*}
&\min_{a>0, h\geqs 0, \gamma \in \mathbb R, \what y\in \mathbb R, \omega^*\in [0,\Omega]}\left|ae^{i\gamma+\what y\omega^*}e^{-i\what y \Omega} -h m_{\min} e^{-iy_1 \Omega}- 2m_{\min} \cos(y_1 \Omega)\right|\\
&\qquad \qquad \qquad \qquad \qquad+\left|ae^{i\gamma+\what y\omega^*}-(2m_{\min}+hm_{\min})\right|\\
=&2m_{\min}-2m_{\min} \cos(y_1 \Omega),
\end{align*}
in the case when $2m_{\min}\cos(y_1\Omega)+hm_{\min}\leqs a\leqs 2m_{\min}+hm_{\min}$.

Now, we consider the case when $a< 2m_{\min}\cos(y_1\Omega)+hm_{\min}$. In this case, we have 
\begin{align*}
&\min_{a>0, h\geqs 0, \gamma \in \mathbb R, \what y\in \mathbb R, \omega^*\in [0,\Omega]}\left|ae^{i\gamma+\what y\omega^*}e^{-i\what y \Omega} -h m_{\min} e^{-iy_1 \Omega}- 2m_{\min} \cos(y_1 \Omega)\right|\\
&\qquad \qquad \qquad \qquad \qquad+\left|ae^{i\gamma+\what y\omega^*}-(2m_{\min}+hm_{\min})\right|\\
\geqs & \min_{a>0, h\geqs 0, \gamma \in \mathbb R, \omega^*\in [0,\Omega], a< 2m_{\min}\cos(y_1\Omega)+hm_{\min}}|ae^{i\gamma+\what y\omega^*}-(2m_{\min}+hm_{\min})|\\
\geqs& \min_{a>0, h\geqs 0, a< 2m_{\min}\cos(y_1\Omega)+hm_{\min}}|a-(2m_{\min}+hm_{\min})|\\
=& \min_{a>0, h\geqs 0, a< 2m_{\min}\cos(y_1\Omega)+hm_{\min}}2m_{\min}+hm_{\min}-a\\
>& 2m_{\min} -2m_{\min}\cos(y_1 \Omega).
\end{align*}

Finally, we consider the case when $a> 2m_{\min}+hm_{\min}$. In this case, we have 
\begin{align*}
&\min_{a>0, h\geqs 0, \gamma \in \mathbb R, \what y\in \mathbb R, \omega^*\in [0,\Omega]}\left|ae^{i\gamma+\what y\omega^*}e^{-i\what y \Omega} -h m_{\min} e^{-iy_1 \Omega}- 2m_{\min} \cos(y_1 \Omega)\right|\\
&\qquad \qquad \qquad \qquad \qquad+\left|ae^{i\gamma+\what y\omega^*}-(2m_{\min}+hm_{\min})\right|\\
\geqs & \min_{a>0, h\geqs 0, \gamma \in \mathbb R, \what y\in \mathbb R,  \omega^*\in [0,\Omega], a> 2m_{\min}+hm_{\min}}\left|ae^{i\gamma+\what y\omega^*}e^{-i\what y \Omega} -h m_{\min} e^{-iy_1 \Omega}- 2m_{\min} \cos(y_1 \Omega)\right|\\
\geqs& \min_{a>0, h\geqs 0, a> 2m_{\min}+hm_{\min}}a-(2m_{\min}\cos(y_1 \Omega)+hm_{\min})\\
>& 2m_{\min} -2m_{\min}\cos(y_1 \Omega).
\end{align*}

Therefore, combining all the above discussions, we arrive at
\begin{align*}
&\min_{a>0, h\geqs 0, \gamma \in \mathbb R, \what y\in \mathbb R, \omega^*\in [0,\Omega]}\left|ae^{i\gamma+\what y\omega^*}e^{-i\what y \Omega} -h m_{\min} e^{-iy_1 \Omega}- 2m_{\min} \cos(y_1 \Omega)\right|\\
&\qquad \qquad \qquad \qquad \qquad+\left|ae^{i\gamma+\what y\omega^*}-(2m_{\min}+hm_{\min})\right|\\
=&2m_{\min}-2m_{\min} \cos(y_1 \Omega),
\end{align*}
or equivalently,
\begin{align*}
&\min_{a>0, \alpha\geqs 1, \gamma \in \mathbb R, \what y\in \mathbb R, \omega^*\in [0,\Omega]}\left|ae^{i\gamma+\what y\omega^*}e^{-i\what y \Omega} -m_{\min}\left(\alpha e^{-iy_1 \Omega}+ e^{i y_1\Omega}\right)\right|+\left|ae^{i\gamma+\what y\omega^*} -m_{\min}\left(\alpha+1\right)\right|\\
=&2m_{\min}-2m_{\min} \cos(y_1 \Omega).
\end{align*}
Thus (\ref{equ:proofcomplextwopointlimit9}) is equivalent to 
\[
2m_{\min}-2m_{\min} \cos(y_1 \Omega)<4\sigma.
\]
Similar to the proof of Theorem \ref{thm:twopointresolution0}, this yields
\[
d_{\min} < \frac{4\arcsin\left(\left(\frac{\sigma}{m_{\min}}\right)^{\frac{1}{2}}\right)}{\Omega}.
\]

\bigskip
\textbf{Part 2:} ($y_1\Omega<\beta$)\\
In part 2, because $y_1\Omega<\beta$, the trick used in the former proof doesn't work now. We utilize another finding for the proof. Suppose $d_{\min}\geqs \frac{4\arcsin\left(\left(\frac{\sigma}{m_{\min}}\right)^{\frac{1}{2}}\right)}{\Omega}$ and there exist some measure $\what \mu=a\delta_{\what y}$ so that 
\[
|\mathcal F[\what \mu](\omega)-\vect Y(\omega)|<\sigma, \qquad \omega \in [-\Omega, \Omega].
\]
Then, this is in contradiction with (\ref{equ:numberalgo1}) and (\ref{equ:numberalgo2}) in Theorem \ref{thm:onednumberdetectalgo1}. Thus we have proved that 
\[
d_{\min}< \frac{4\arcsin\left(\left(\frac{\sigma}{m_{\min}}\right)^{\frac{1}{2}}\right)}{\Omega}. 
\]
Note that this new finding can also be used to prove the first part, but we keep the first part for a stronger understanding of the optimization problem and its underlying difficulty. The new finding comes from an optimal algorithm described in the next section. 
Now we have completed the proof. 
\end{proof}

\begin{remark}
We remark that the objectives of Section \ref{section:stabilitySR} and Section \ref{section:twopointlimit} are different. Section \ref{section:stabilitySR} provides estimates for the computational resolution limits in super-resolving $n$-sparse sources. In contrast, Section \ref{section:twopointlimit} focuses on deriving the exact formula for the computational resolution limit when super-resolving two point sources. Regarding the methods of proof, directly addressing the optimization problem (\ref{equ:proofpositivetwopointlimit1.2}) and (\ref{equ:proofcomplextwopointlimit9}) yields optimal estimates for two-point resolution. However, this approach does not extend to the case of super-resolving $n$-sources. Conversely, location-amplitude identities offer a robust framework for analyzing the resolution of super-resolving $n$-sparse sources. Moreover, according to Theorem \ref{thm:computatwopointresolution0}, the resolution estimate in Theorem \ref{thm:upperboundnumberlimithm0} is already very sharp, demonstrating the effectiveness of location-amplitude identities in delivering precise insights into super-resolution problems. 
\end{remark}

\subsection{Two-point resolution for  very general imaging models}
The two-point resolution estimate in previous sections  can actually be generalize to very general imaging problems as we shall discuss next.  We assume that the available measurement is 
\begin{equation}\label{equ:twopointgeneralmodelsetting1}
	\mathbf Y(\vect{\omega}) = \chi(\vect \omega)\left(\mathcal F[\mu] (\vect{\omega}) + \mathbf W(\vect{\omega})\right)= \sum_{j=1}^{n}a_j \chi(\vect \omega) e^{i \vect{y}_j\cdot \vect{\omega}} + \chi(\vect \omega) \mathbf W(\vect{\omega}), \ \vect \omega \in \mathbb R^k, \ ||\vect{\omega}||_2\leqs \Omega,
\end{equation}
where $\chi(\vect \omega)=0$ or $1$, $\chi(\vect 0)=1$ and $\chi(\vect \omega)=1, ||\vect \omega||_2 =\Omega$. Moreover, the noise $\vect W$ is assumed to be bounded as:
\[
|\vect W(\vect \omega)|<\sigma, \quad  \bnorm{\vect \omega}\leqs \Omega. 
\]

For the imaging model (\ref{equ:twopointgeneralmodelsetting1}), consider similar definitions to the previous ones for  $\sigma$-admissible measures and the computational resolution limit.  It is not hard to see that the estimates in the previous sections still hold and we have the following theorem. 

\begin{thm}\label{thm:twopointresogeneralsetting}
Consider the imaging model (\ref{equ:twopointgeneralmodelsetting1}). For $\frac{\sigma}{m_{\min}} \leqs \frac{1}{2}$, the resolution limits $\mathcal D_{k, num}^+, \mathcal D_{k, num}$ for resolving two sources in $\mathbb R^k$ are 
\begin{equation}\label{equ:twopointresogeneralcase}
\frac{4\arcsin\left(\left(\frac{\sigma}{m_{\min}}\right)^{\frac{1}{2}}\right)}{\Omega}.
\end{equation}
These resolution limits can be attained if $a_1=a_2$. When $\frac{\sigma}{m_{\min}} >\frac{1}{2}$, no matter what the separation distance is, there are always some $\sigma$-admissible measures of some $\vect Y$ corresponding to one point source.
\end{thm}

Compared to (\ref{equ:highdmodelsetting1}), the model (\ref{equ:twopointgeneralmodelsetting1}) is more general, for instance, super-resolution from discrete measurements can be modeled by (\ref{equ:twopointgeneralmodelsetting1}). Thus Theorem \ref{thm:twopointresogeneralsetting} can be applied directly to super-resolution in practice, line spectral estimation, and direction-of-arrival. Moreover, by the inverse filtering methods, our results can be applied to imaging problems with very general optical transfer functions, such as the one shown in Figure \ref{fig:opticaltranfun}.  We believe that this will inspire new understandings for the resolution of a number of imaging modalities. We remark that it is more appropriate to apply Theorem \ref{thm:twopointresogeneralsetting} to imaging problems where the noise level at $0$ and $||\vect \omega||_2=\Omega$ are close or comparable after modifying the model to  (\ref{equ:twopointgeneralmodelsetting1}). 
%When the noise levels at these sample points are not comparable, we suggest using the same idea as the one introduced in the previous sections in order to derive more accurate estimates.  

In fact, Theorem \ref{thm:twopointresogeneralsetting} reveals the fact that the two-point resolution is actually not that related to the continuous band of frequencies but rather mostly determined by the boundary points. In particular, in the one-dimensional case, if we have only measurements in $[-\Omega+\epsilon, \Omega-\epsilon]$ for $\epsilon>0$, then the resolution in (\ref{equ:twopointresogeneralcase}) does not hold anymore. In the multi-dimensional cases, similar conclusions hold as well.  Thus the condition $\bnorm{\vect \omega}_2=\Omega$ is nearly a necessary condition for Theorem \ref{thm:twopointresogeneralsetting} to hold. 

\begin{figure}[!h]
		\centering
		\includegraphics[width=9.6cm, height=6cm]{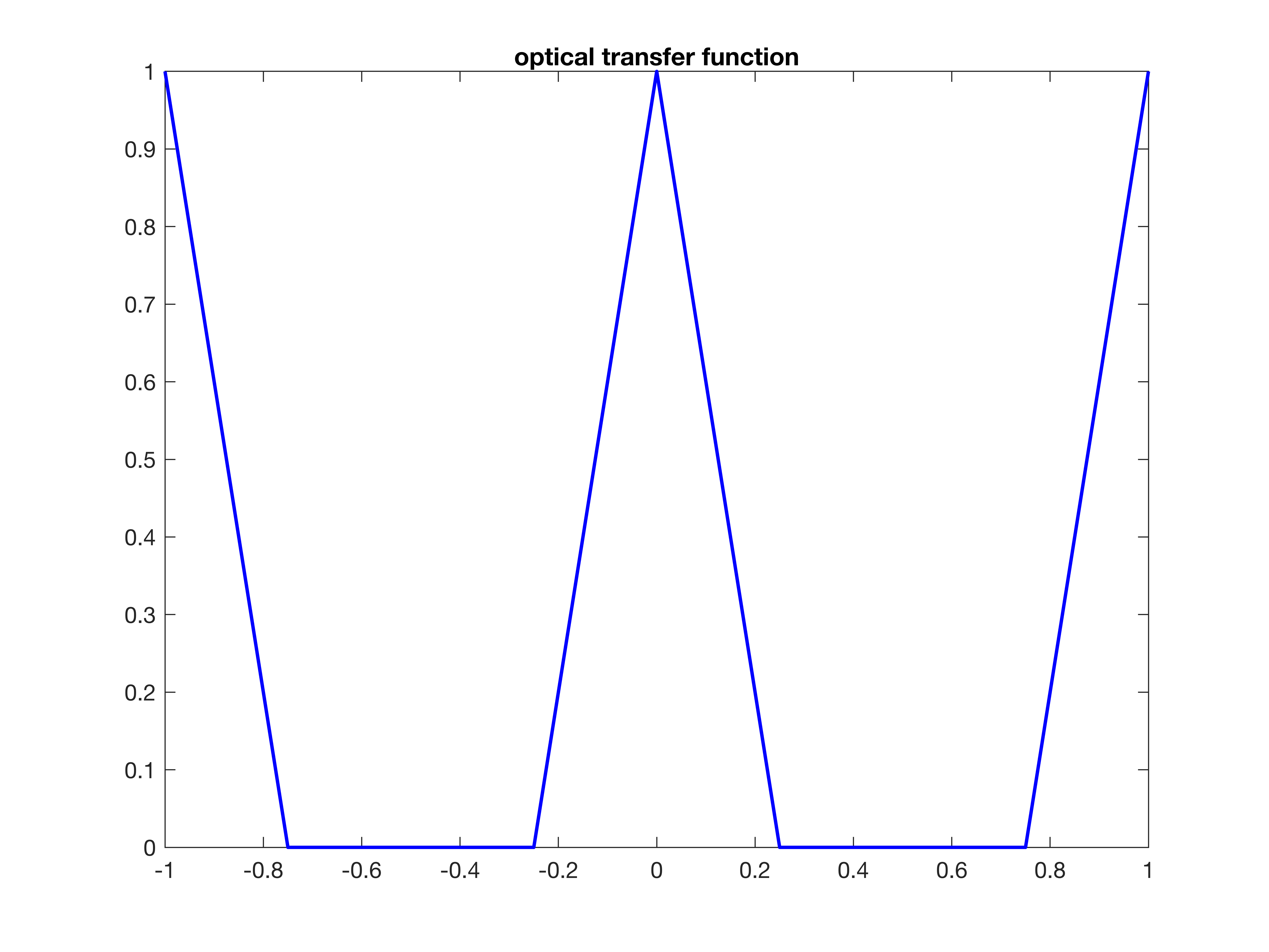}
		\caption{Optical transfer function.}
	\label{fig:opticaltranfun}
\end{figure}

\section{Optimal Algorithms}\label{section:optimalalgorithm1}
We now have the exact resolution limit for determining whether the image is generated by one or two sources. This is a new benchmark for super-resolution and model order detection algorithms. A natural question is whether we can find the optimal algorithm to distinguish between one and two sources in the image. Note that, according to our theoretical results, when the two sources are separated by more than 
\[
\babs{y_1-y_2}\geqs \frac{4\arcsin\left(\left(\frac{\sigma}{m_{\min}}\right)^{\frac{1}{2}}\right)}{\Omega},
\]
any algorithm targeting certain solutions in the set of admissible measures provides a solution with more than one source. But we still cannot confirm that there is more than one source inside. Only by considering the sparest solution in the set of admissible measures can we confirm this fact. However, since $l_0$ minimization is intractable, this direction is still unrealistic and we resort to other means. In \cite{liu2021theorylse}, a simple singular value thresholding-based algorithm was proposed to detect the source number. In this section, we consider a variant of it and theoretically demonstrate that the algorithm exactly attains the resolution limit. We remark that the optimality of our algorithm refers to achieving the two-point resolution limit in Theorems \ref{thm:positivetwopointresolution0} and \ref{thm:computatwopointresolution0} under the imaging model (\ref{equ:highdmodelsetting1}). Algorithms leveraging specific noise patterns may outperform this one.

\subsection{An optimal algorithm for detecting two sources in dimension one}
In \cite{liu2021theorylse}, the authors proposed a number detection algorithm called sweeping singular value thresholding number detection algorithm. It determines the number of sources by thresholding the singular value of a Hankel matrix formulated from the  measurement data. Here we consider a simple variant of it. 

To be more specific, we first assemble the following Hankel matrix from the measurements (\ref{equ:highdmodelsetting1}), that is, 
\begin{equation}\label{equ:hankelmatrix1}
\mathbf H=\left(\begin{array}{cc}
\mathbf Y(-\Omega)& \mathbf Y(0)\\
\mathbf Y(0)&\mathbf Y(\Omega)
\end{array}
\right).\end{equation}
We denote the singular value decomposition of $\mathbf H$ as  
\[\mathbf H=\what U\what \Sigma \what U^*,\]
where $\what\Sigma =\text{diag}(\what \sigma_1, \what \sigma_2)$ with the singular values $\what \sigma_1, \what \sigma_2$ ordered in a decreasing manner. We then determine the source number by a thresholding of the singular values. We derive the following Theorem \ref{thm:onednumberdetectalgo1} for the threshold and the resolution of the algorithm.

\begin{thm}\label{thm:onednumberdetectalgo1}
Consider $\mu=\sum_{j=1}^{2}a_j \delta_{y_j},y_j\in  B_{\frac{(n-1) \pi}{2\Omega}}^1(0)$ and the measurement $\vect Y$ in (\ref{equ:highdmodelsetting1}) that is generated from $\mu$. If the following separation condition is satisfied
	\begin{equation}\label{equ:numberalgo0}
	\babs{y_1-y_2}\geqs \frac{4\arcsin\left(\left(\frac{\sigma}{m_{\min}}\right)^{\frac{1}{2}}\right)}{\Omega},
	\end{equation}
 then we have
	\begin{equation}\label{equ:numberalgo1}
	\what\sigma_{2}> 2\sigma
 \end{equation}
for $\what \sigma_2$ being the minimum singular value of the matrix $\vect H$ in (\ref{equ:hankelmatrix1}). On the other hand, if there exists $\what \mu$ consisting of only one source being a $\sigma$-admissible measure of \ $\vect Y$, then
 \begin{equation}\label{equ:numberalgo2}
 \what \sigma_2<2\sigma.
 \end{equation}
\end{thm}
\begin{proof}
Observe that $\mathbf H$ has the decomposition
\begin{equation}\label{equ:hankeldecomp1}
\mathbf H = DAD^{\top}+\Delta,
\end{equation}
where $A=\text{diag}(e^{-iy_1\Omega}a_1, e^{-iy_2\Omega}a_2)$ and $D=\big(\phi_{1}(e^{i y_1 \Omega}), \phi_{1}(e^{i y_2\Omega})\big)$ with $\phi_{1}(\omega)$ being defined as $(1, \omega)^\top$ and 
\begin{equation*}
\Delta = \left(\begin{array}{cc}
\mathbf {W}(-\Omega)& \mathbf {W}(0)\\
\mathbf {W}(0)&\mathbf {W}(\Omega)
\end{array}
\right).
\end{equation*}
We denote the singular values of $DAD^{\top}$ by $\sigma_1, \sigma_2$.

We first estimate $\bnorm{\Delta}_2$. We have
\begin{align*}
&\max_{x_1^2+x_2^2=1}\bnorm{\Delta(x_1, x_2)^{\top}}_2\\
= &\max_{x_1^2+x_2^2=1}\sqrt{(x_1\vect W(-\Omega) +x_2  \vect W(0))^2+ (x_1\vect W(0) +x_2  \vect W(\Omega))^2}\\
= &\max_{x_1^2+x_2^2=1}\sqrt{\vect W(0)^2+ 2x_1x_2\vect W(0)(\vect W(-\Omega)+\vect W(\Omega))+x_1^2\vect W(-\Omega)^2+x_2^2\vect W(\Omega)^2}\\
<&\max_{x_1^2+x_2^2=1}\sqrt{\sigma^2+ 4 \sigma^2 x_1x_2+ (x_1^2+x_2^2)\sigma^2} \quad \Big(\text{by the condition on the noise}\Big)\\
=&2\sigma.
\end{align*}
Thus we have $||\Delta||_2<2\sigma$. By Weyl's theorem, we have 
\begin{equation}\label{equ:proofnumberalgoeq1}
\babs{\what \sigma_j-\sigma_j}\leqs ||\Delta||_2<2\sigma, j=1,2.
\end{equation} 

Now we estimate the minimum singular value of $DAD^{\top}$ in the presence of two sources. Denote $\sigma_{\min}(M)$ and $\lambda_{\min}(M)$ as respectively the minimum singular value and eigenvalue of matrix $M$. We have 
\begin{align*}
\sigma_{\min}(DAD^{\top})\geqs m_{\min}\sigma_{\min}(D)^2=m_{\min}\lambda_{\min}(DD^*)= 4m_{\min}\sin^2\left(\babs{\frac{y_1-y_2}{4}}\Omega\right). 
\end{align*}
Therefore, when $y_j\in B_{\frac{(n-1) \pi}{2\Omega}}^1(0),j=1,2,$ and (\ref{equ:numberalgo0}) holds, $\sigma_{\min}(DAD^{\top})\geqs 4\sigma$. This is $\sigma_2 \geqs 4\sigma$. Similarly, by Weyl's theorem, $|\what \sigma_2-\sigma_2|\leqs ||\Delta||_2$. Thus, $\what \sigma_2\geqs 4\sigma-||\Delta||_2>2\sigma$. Conclusion (\ref{equ:numberalgo1}) follows. 

On the other hand, note that if there exists $\what \mu= \what a_1\delta_{\what y_1}$ consisting of one source being a $\sigma$-admissible measure of $\vect Y$, we can substitute the $D$ in (\ref{equ:hankeldecomp1}) by $\big(\phi_{1}(e^{i \what y_1 \Omega})\big)$ with the $\vect W$ and $\Delta$ being modified. Now we have $\sigma_2=0$ and also $||\Delta||_2< 2\sigma$. Thus by (\ref{equ:proofnumberalgoeq1}) we get $|\what \sigma_2|\leqs ||\Delta||_2< 2\sigma$ and prove (\ref{equ:numberalgo2}). 
\end{proof}

We summarize the algorithm in the following \textbf{Algorithm \ref{algo:onedsinguvaluenumberalgo}}. Note that in practical applications one can estimate a noise level although not tight and utilize our algorithm to detect the source number. By Theorem \ref{thm:onednumberdetectalgo1}, for all estimated $\sigma$'s less than $\frac{m_{\min}}{2}$, our algorithm can achieve super-resolution. 

\begin{algorithm}[H]
	\caption{\textbf{Singular-value-thresholding number detection algorithm}}
	\textbf{Input:} Noise level $\sigma$;\\
	\textbf{Input:} Measurement: $\mathbf{Y}(\omega), \omega \in [-\Omega, \Omega]$;\\	
	1: Formulate the Hankel matrix \[
\vect H = \begin{pmatrix}
\mathbf Y(-\Omega)& \mathbf Y(0)\\
\mathbf Y(0)&\mathbf Y(\Omega)
\end{pmatrix}
    \]
    from measurement $\vect Y(\omega)$;\\
    2: Compute the singular value of $\mathbf H$ as $\what \sigma_{1}, \what \sigma_{2}$ distributed in a decreasing manner;\\
    3: If $\what \sigma_2\geqs 2\sigma$, determine source number $n=2$ and otherwise,  determine $n=1$;\\
    \textbf{Return:} $n$. 
    \label{algo:onedsinguvaluenumberalgo}
\end{algorithm}

\noindent \textbf{Numerical experiments:}\\
We conduct many numerical experiments to elucidate the performance of \textbf{Algorithm \ref{algo:onedsinguvaluenumberalgo}}. We consider $\Omega=1$ and measurements $\vect Y$ generated by two sources. The noise level is $\sigma$ and the minimum separation distance between sources is $d_{\min}$. We first perform $100000$ random experiments (the randomness is in the choice of $(d_{\min},\sigma, y_j, a_j)$) and the results were shown in Figure \ref{fig:onedtwopointsphasetransition} (a)-(c). The green points and red points represent respectively the cases of successful detection and failed detection. It is indicated that in many cases, our \textbf{Algorithm \ref{algo:onedsinguvaluenumberalgo}} can surpass the two-point resolution limit. We also conduct $100000$ experiments for the worst-case scenario; see results in Figure \ref{fig:onedtwopointsphasetransition} (d)-(f). As shown numerically, our algorithm successfully detects the source number when $d_{\min}$ is above the two-point resolution limit and fails in exactly the opposite cases. Last, we consider the worst cases when detecting the source number is impossible when $\frac{\sigma}{m_{\min}}>\frac{1}{2}$. The results were presented in Figure \ref{fig:onedtwopointsphasetransition} (g)-(i) and there is no successful case when $\frac{\sigma}{m_{\min}}>\frac{1}{2}$. Note that the failed cases when $\frac{\sigma}{m_{\min}}<\frac{1}{2}$ and $d_{\min}$ above the two-point resolution limit is due to the fact that $|e^{iy_1\Omega}-e^{iy_2\Omega}|$ becomes small when $|y_1-y_2|\Omega$ approaching $2\pi$. 

We also conduct several experiments to illustrate that our algorithm can detect the correct source number even if it seems very unlikely to distinguish the two sources by other methods. We consider $5$ cases where the source number is correctly detected by our algorithm; see Figure \ref{fig:numberdetectandmusic} (a). However, as shown by Figure \ref{fig:numberdetectandmusic} (b)-(f), their MUSIC images only have one peak.

\begin{figure}[!h]
	\centering
        \begin{subfigure}[b]{0.3\textwidth}
		\centering
		\includegraphics[width=\textwidth]{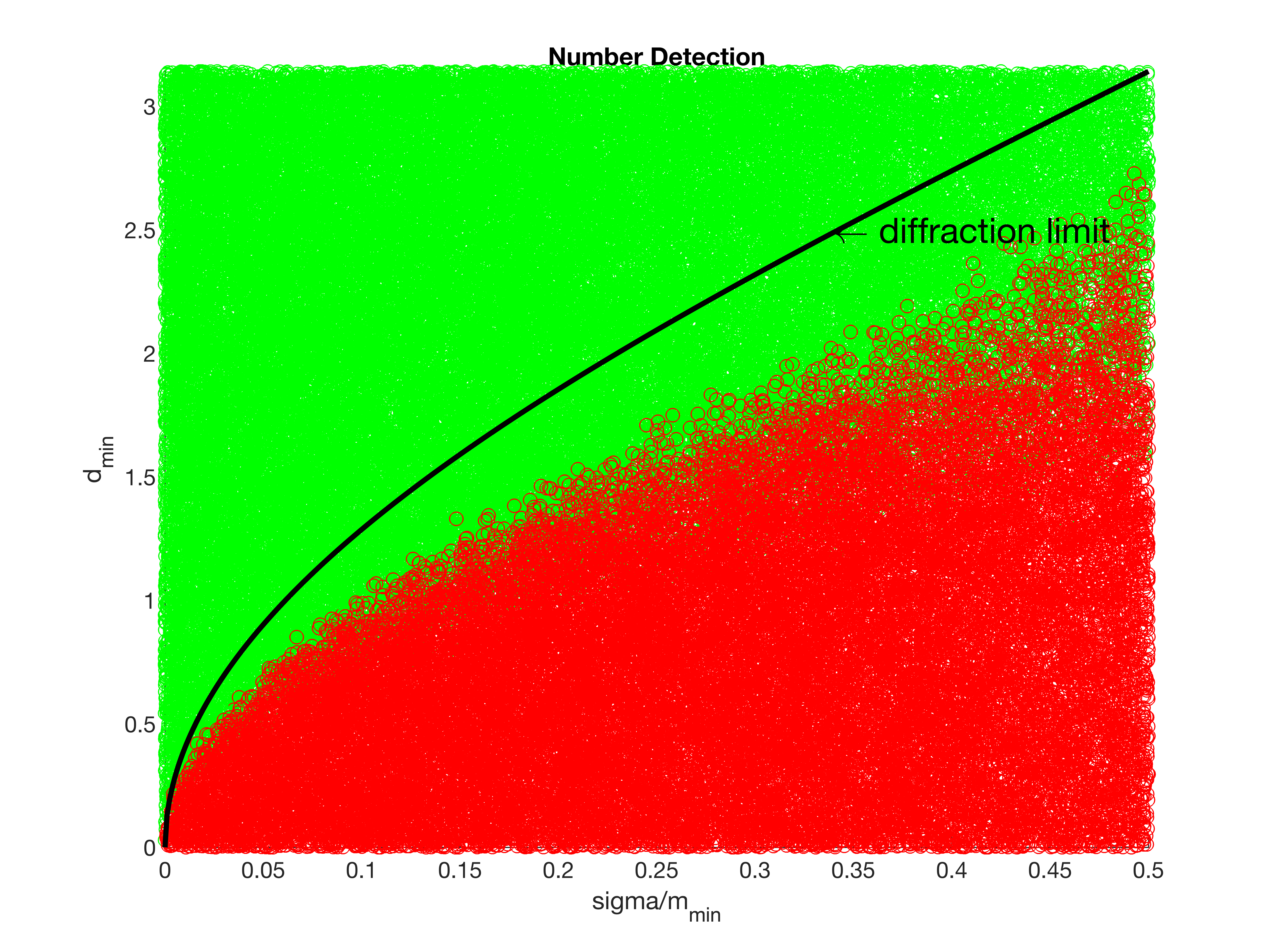}
		\caption{detection results}
	\end{subfigure}
	\begin{subfigure}[b]{0.3\textwidth}
		\centering
		\includegraphics[width=\textwidth]{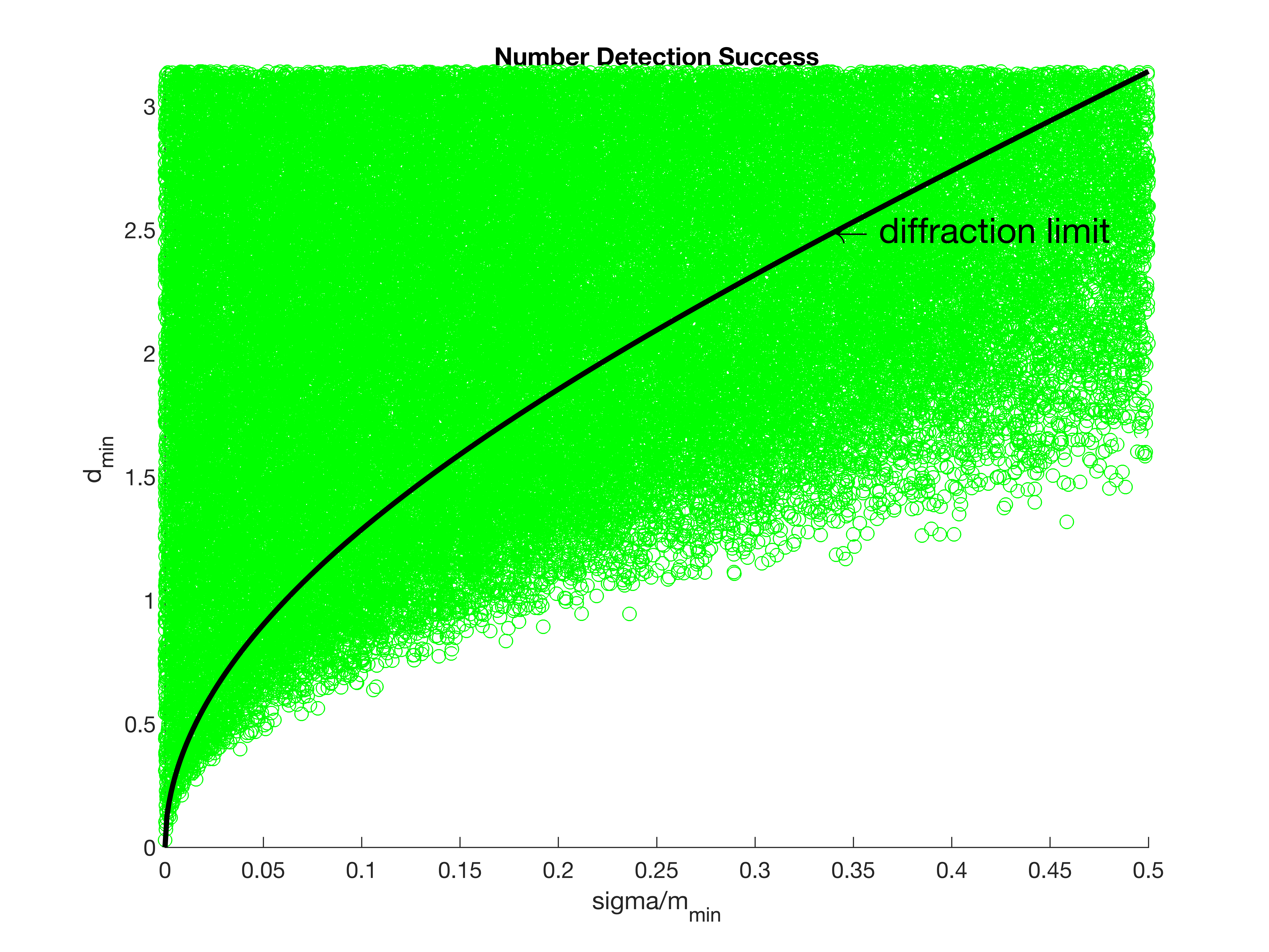}
		\caption{detection success}
	\end{subfigure}
        \begin{subfigure}[b]{0.3\textwidth}
		\centering
		\includegraphics[width=\textwidth]{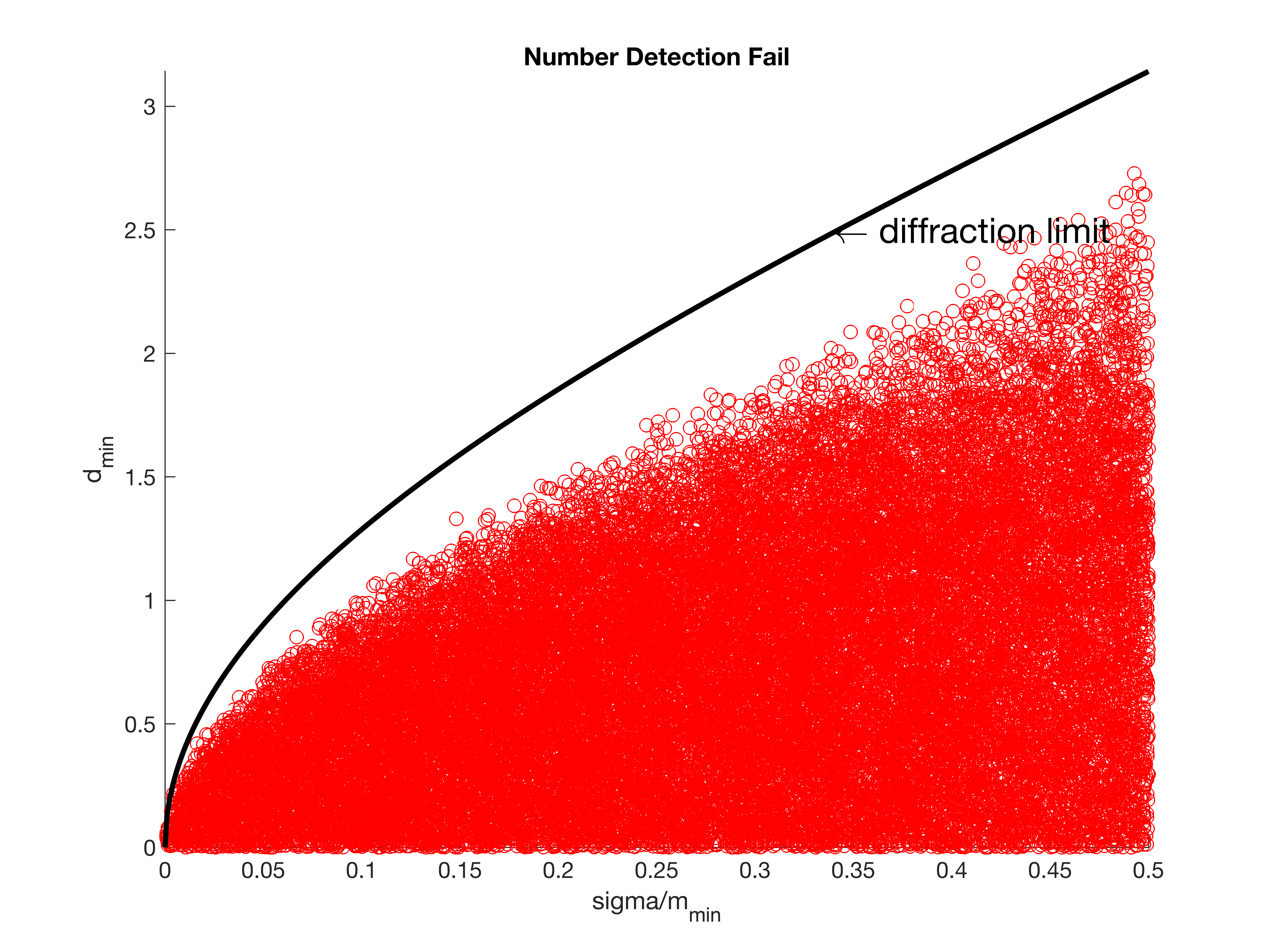}
		\caption{detection fail}
	\end{subfigure}
	\begin{subfigure}[b]{0.3\textwidth}
		\centering
		\includegraphics[width=\textwidth]{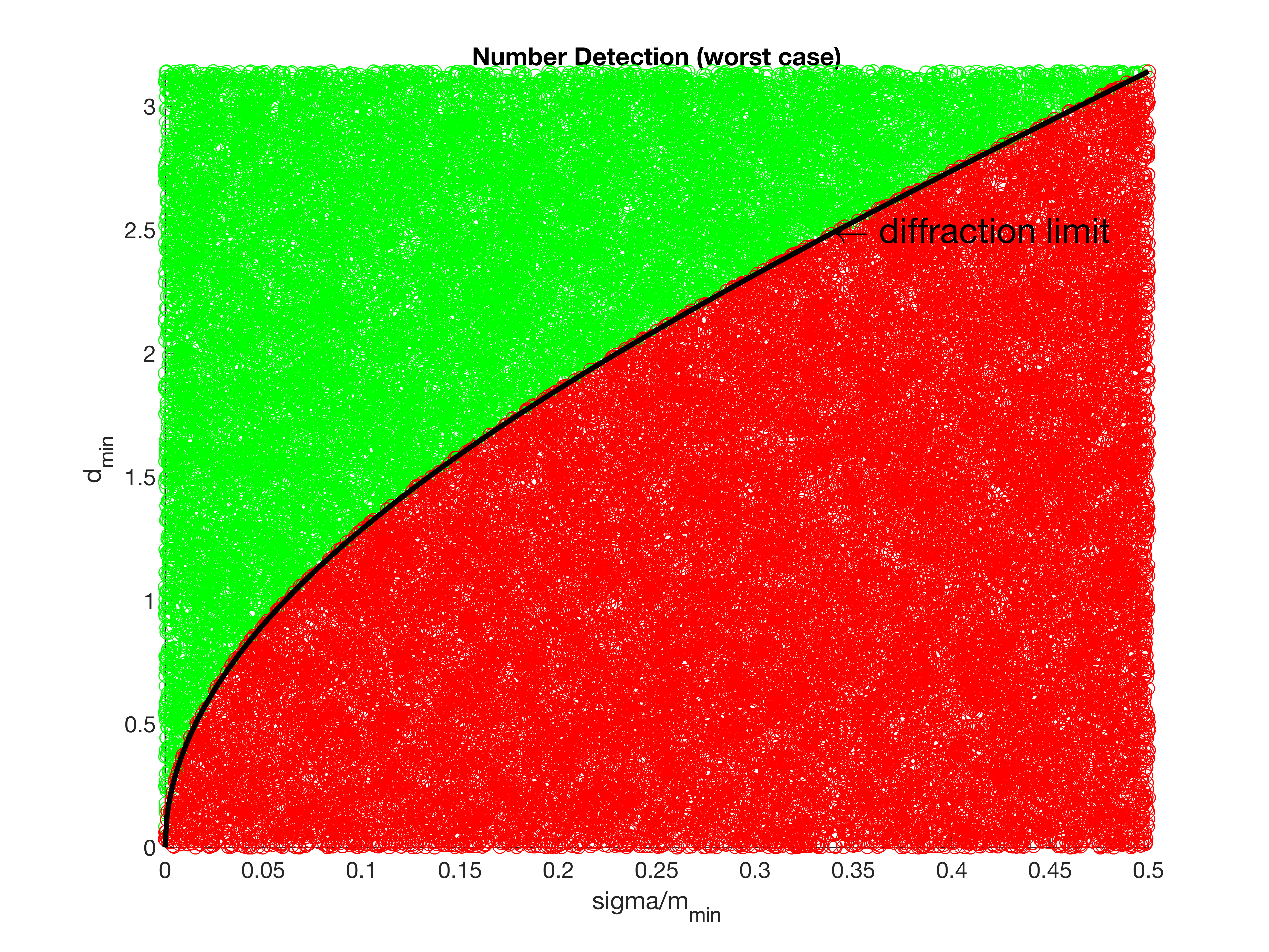}
		\caption{detection results}
	\end{subfigure}
	\begin{subfigure}[b]{0.3\textwidth}
		\centering
		\includegraphics[width=\textwidth]{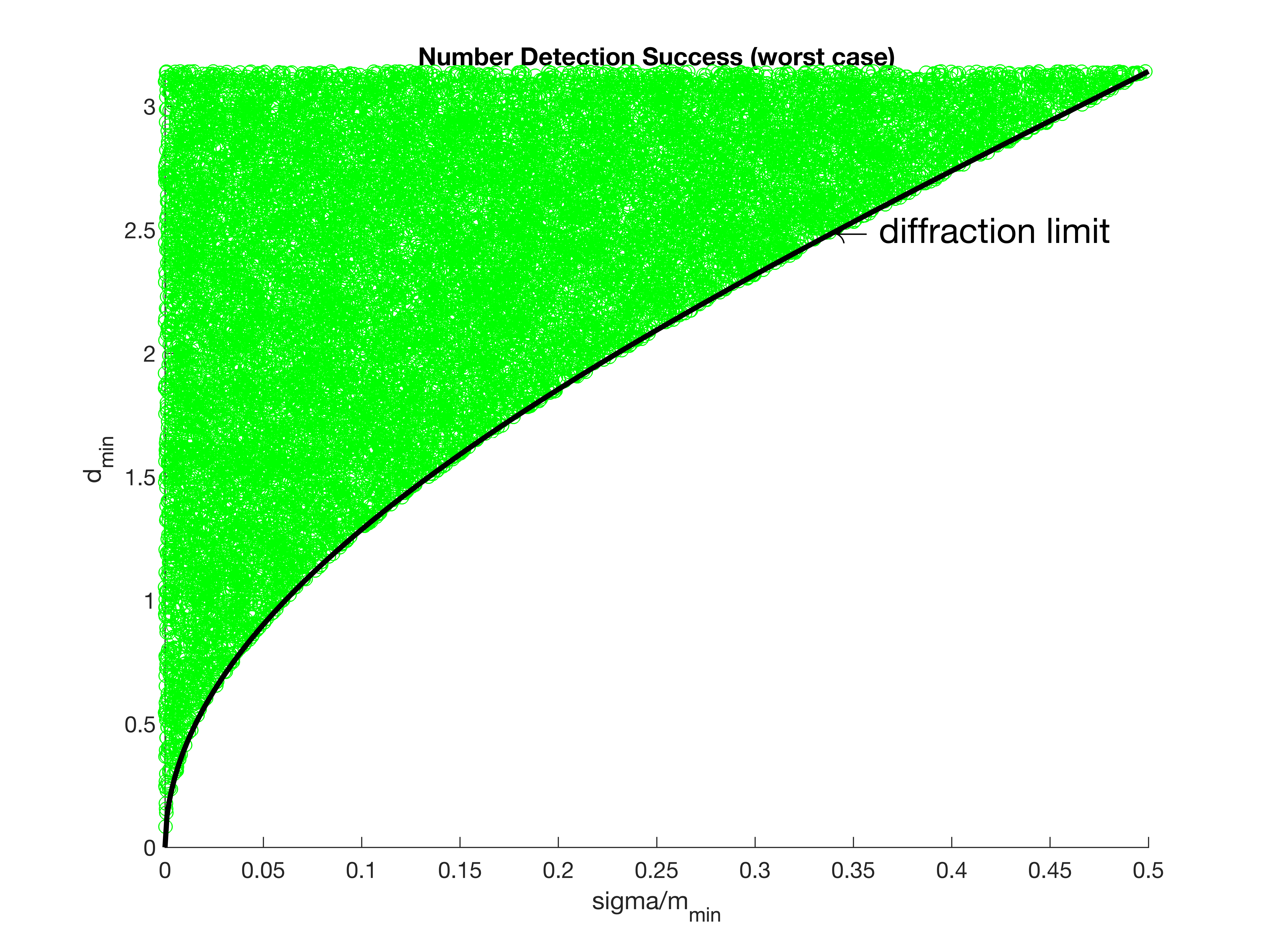}
		\caption{detection success}
	\end{subfigure}
        \begin{subfigure}[b]{0.3\textwidth}
		\centering
		\includegraphics[width=\textwidth]{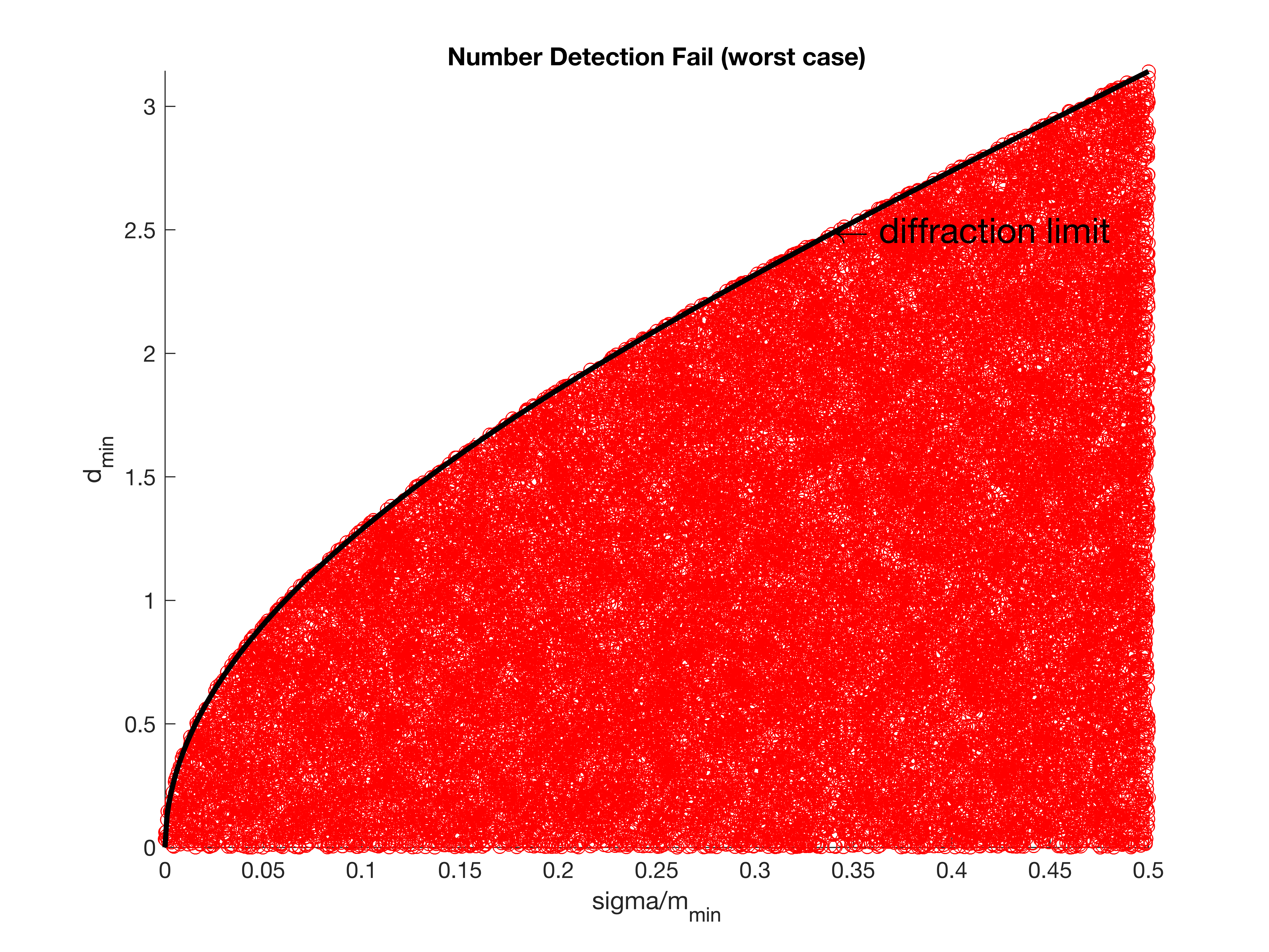}
		\caption{detection fail}
	\end{subfigure}
        \begin{subfigure}[b]{0.3\textwidth}
		\centering
		\includegraphics[width=\textwidth]{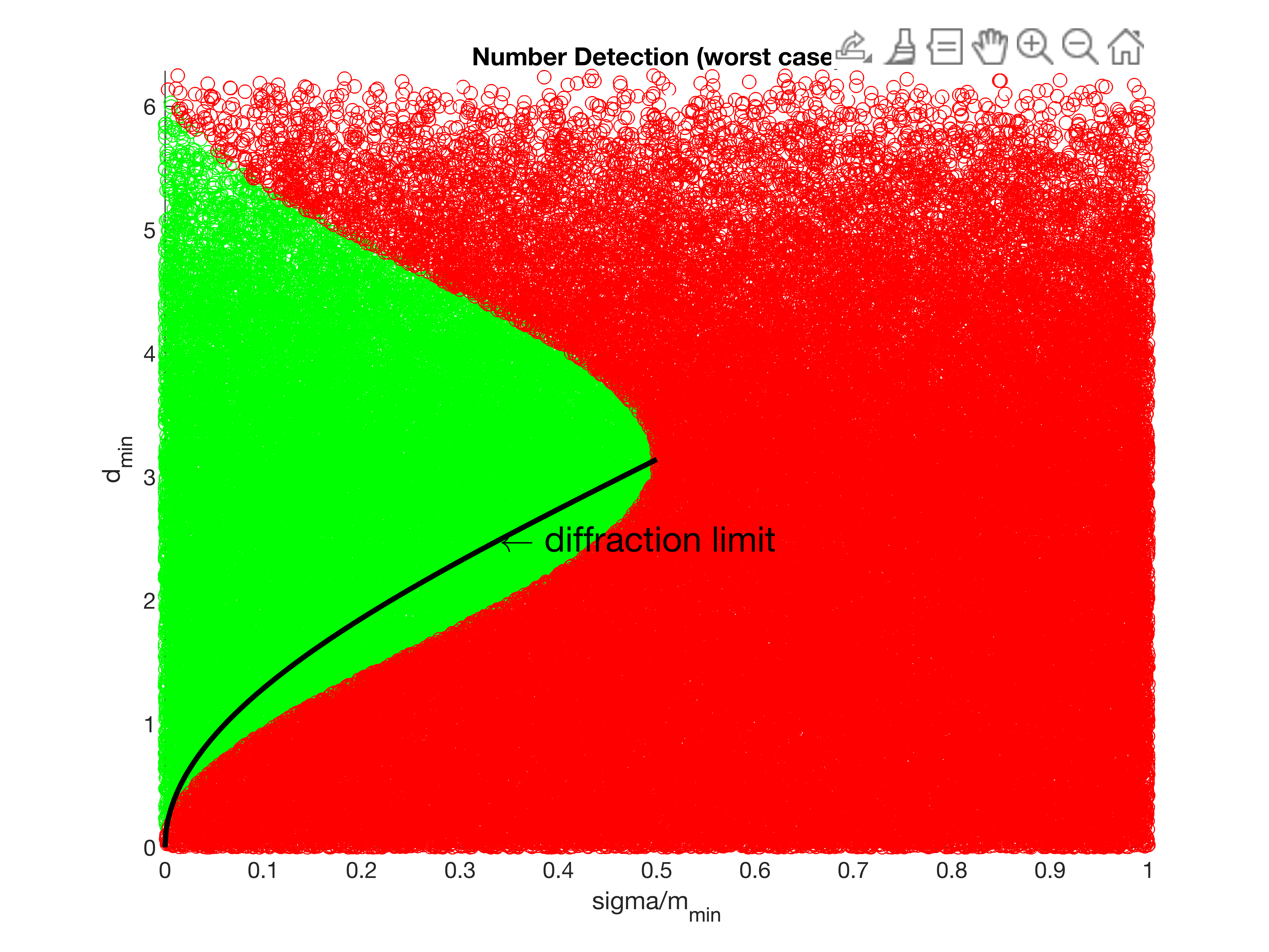}
		\caption{detection results}
	\end{subfigure}
	\begin{subfigure}[b]{0.3\textwidth}
		\centering
		\includegraphics[width=\textwidth]{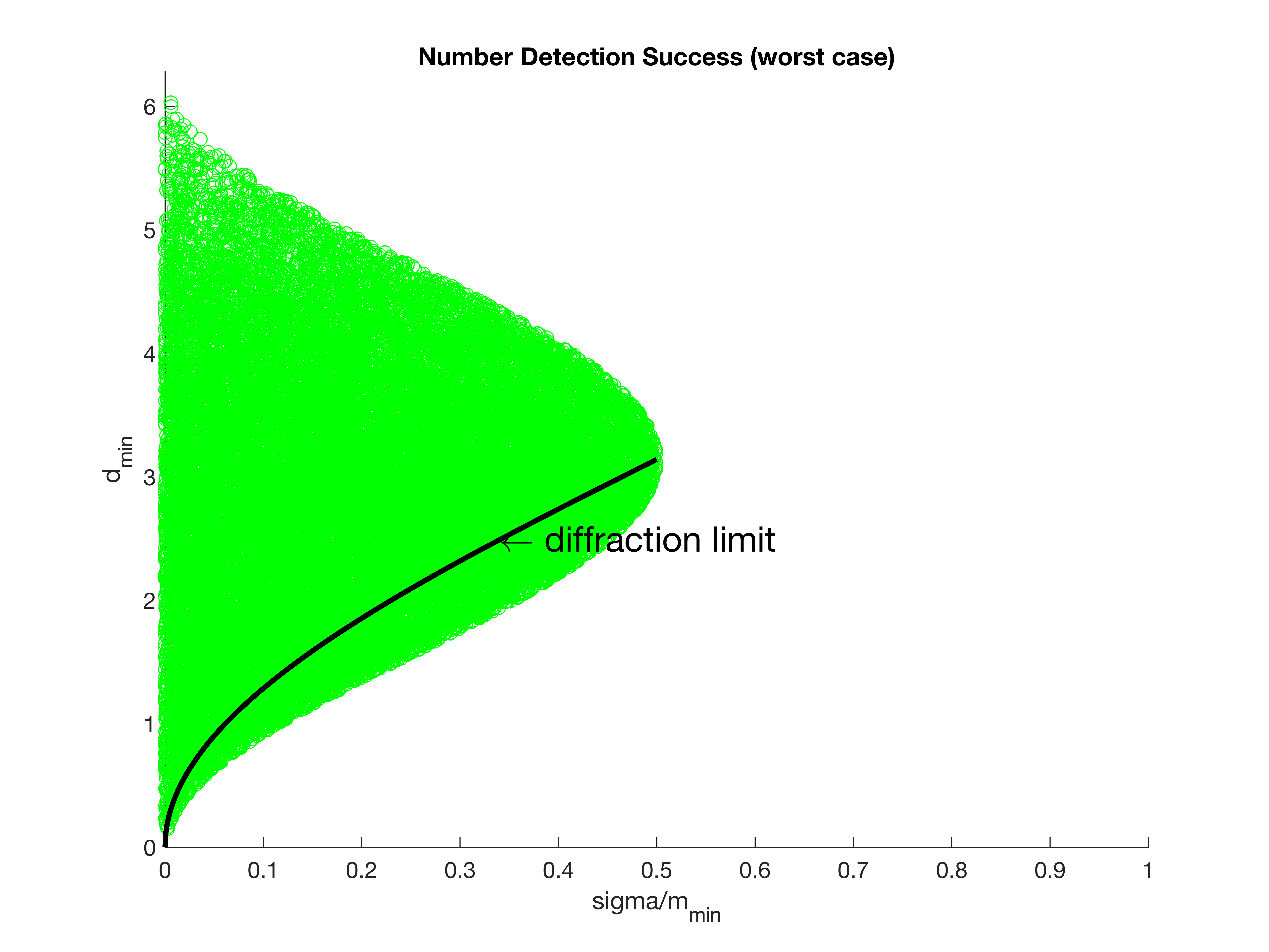}
		\caption{detection success}
	\end{subfigure}
        \begin{subfigure}[b]{0.3\textwidth}
		\centering
		\includegraphics[width=\textwidth]{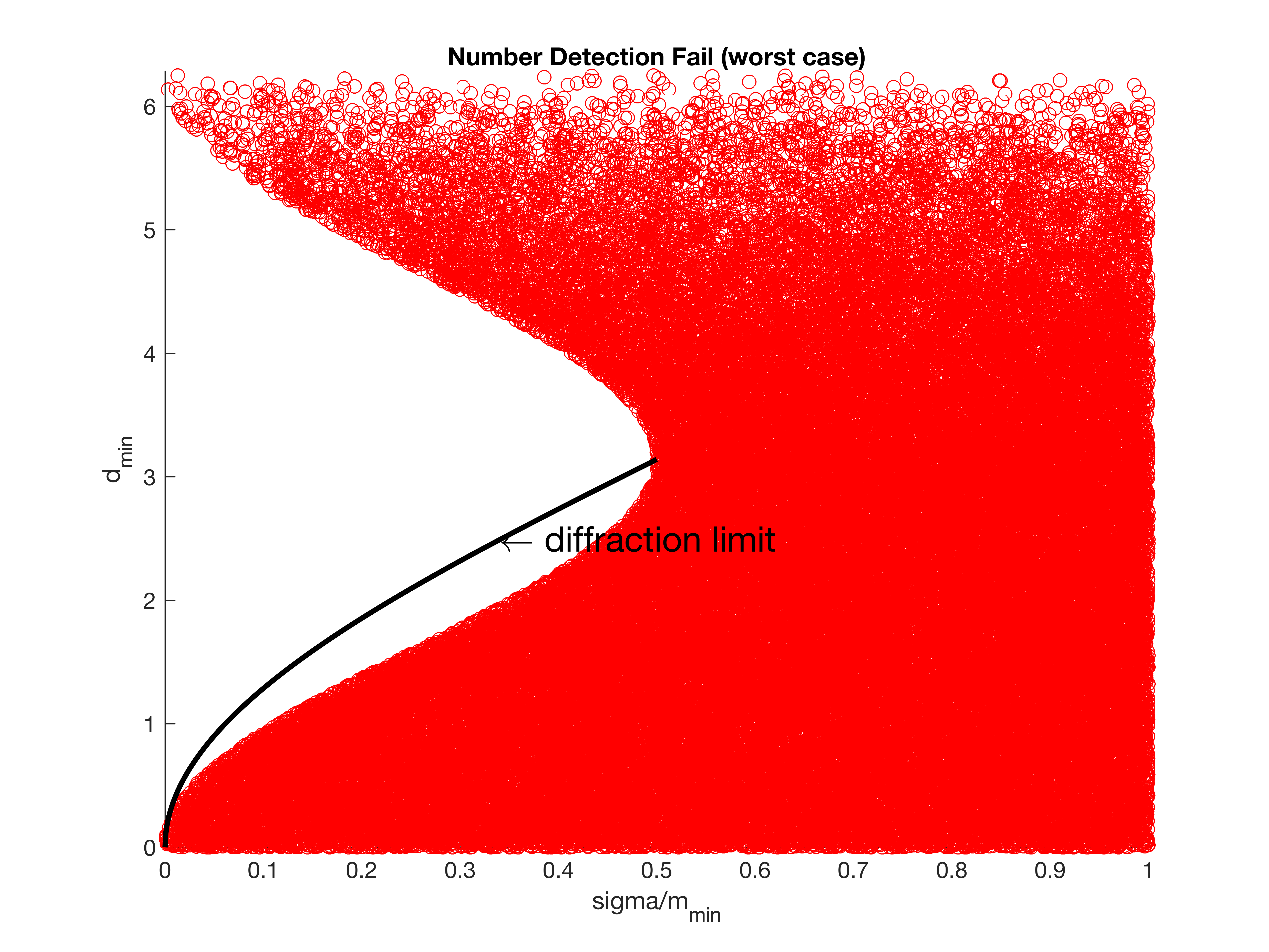}
		\caption{detection fail}
	\end{subfigure}
	\caption{Plots of the successful and the unsuccessful number detections by \textbf{Algorithm \ref{algo:onedsinguvaluenumberalgo}} depending on the relation between $\frac{\sigma}{m_{\min}}$ and $d_{\min}$. The green points and red points represent respectively the cases of successful detection and failed detection. The black line is the two-point resolution limit $\mathcal D_{num}(k, 2)$ derived in Theorem \ref{thm:computatwopointresolution0}.}
	\label{fig:onedtwopointsphasetransition}
\end{figure}

\begin{figure}[!h]
	\centering
	\begin{subfigure}[b]{0.3\textwidth}
		\centering
		\includegraphics[width=\textwidth]{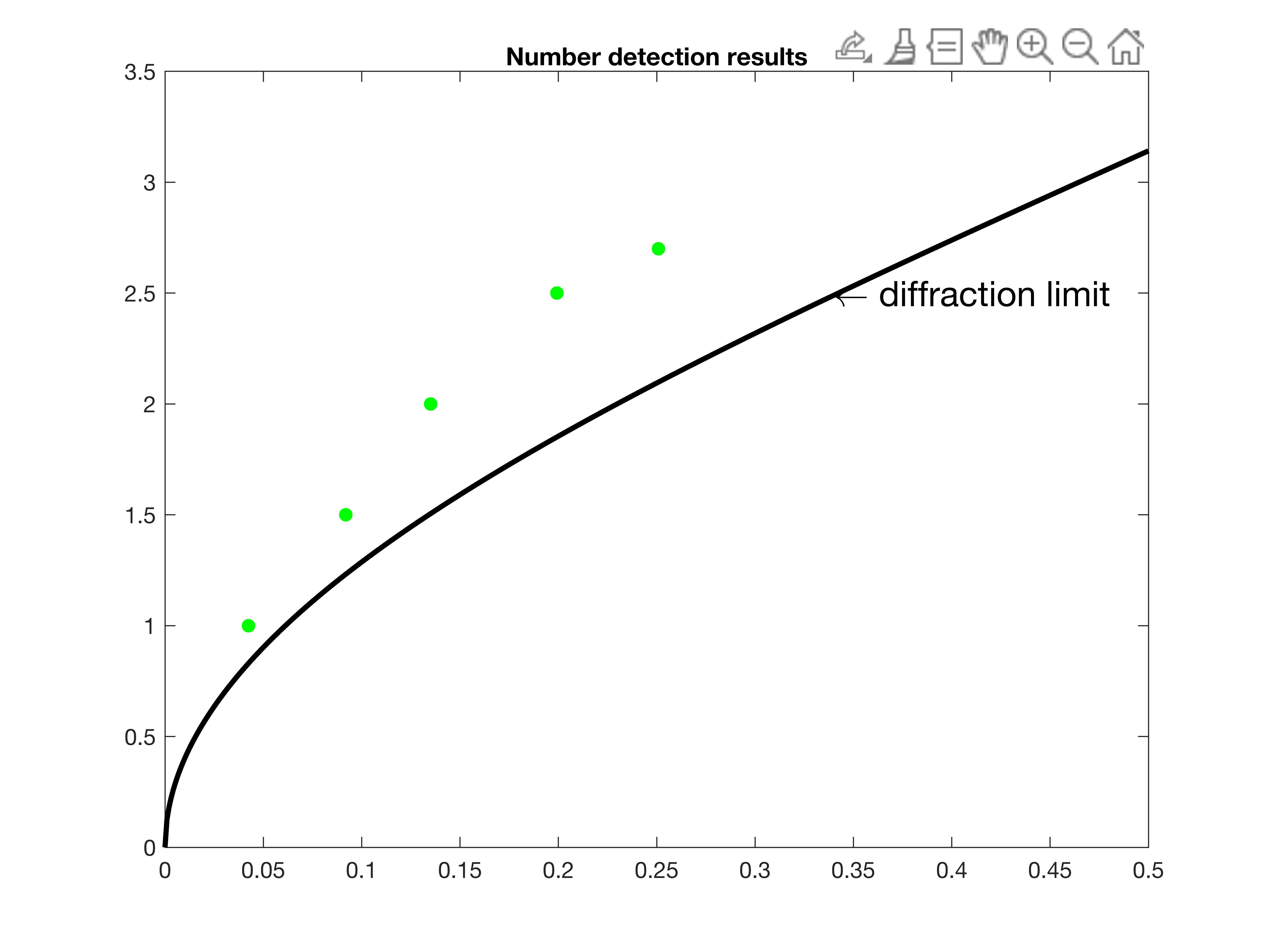}
		\caption{detection results}
	\end{subfigure}
	\begin{subfigure}[b]{0.3\textwidth}
		\centering
		\includegraphics[width=\textwidth]{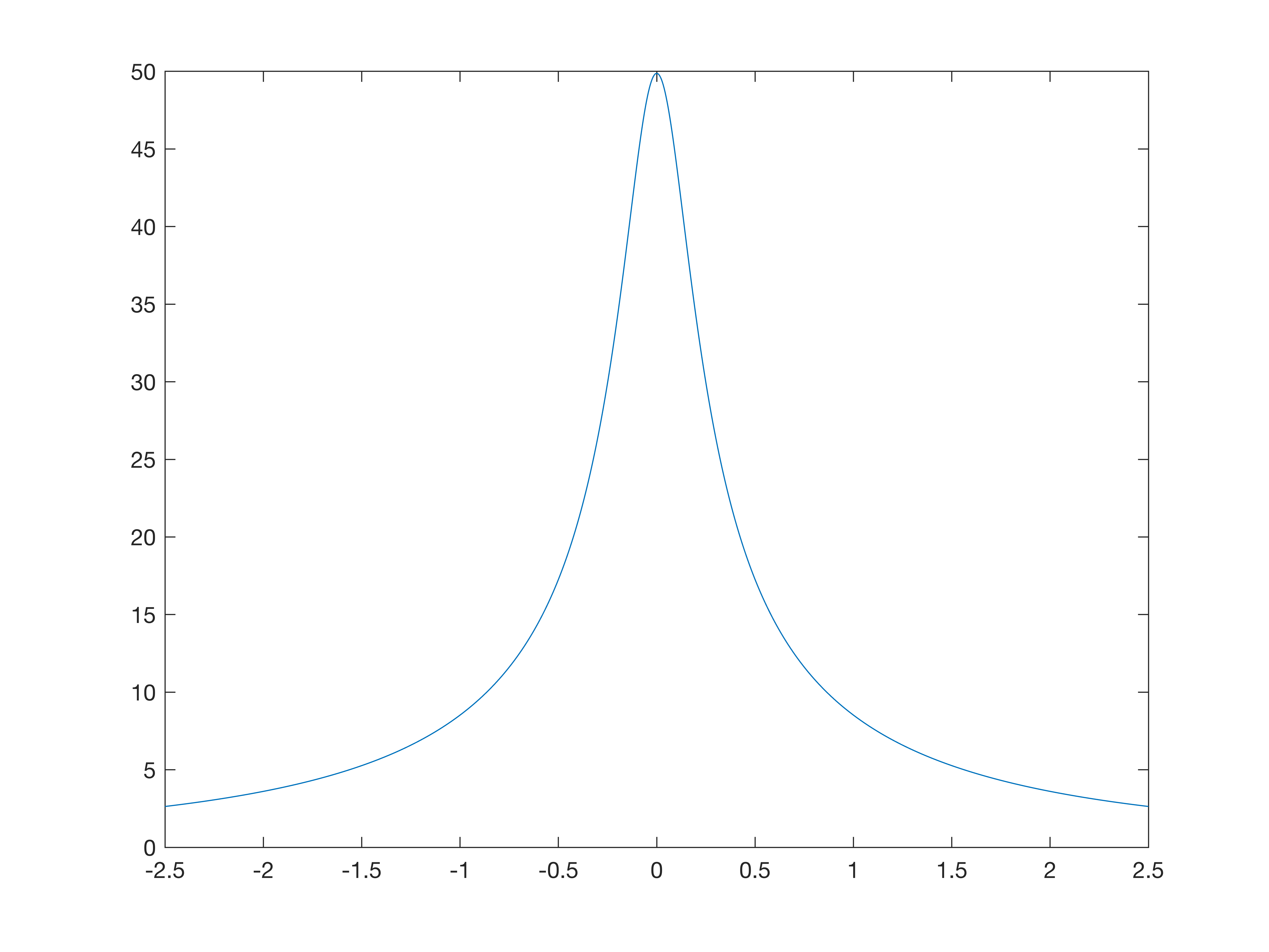}
		\caption{MUSIC image}
	\end{subfigure}
        \begin{subfigure}[b]{0.3\textwidth}
		\centering
		\includegraphics[width=\textwidth]{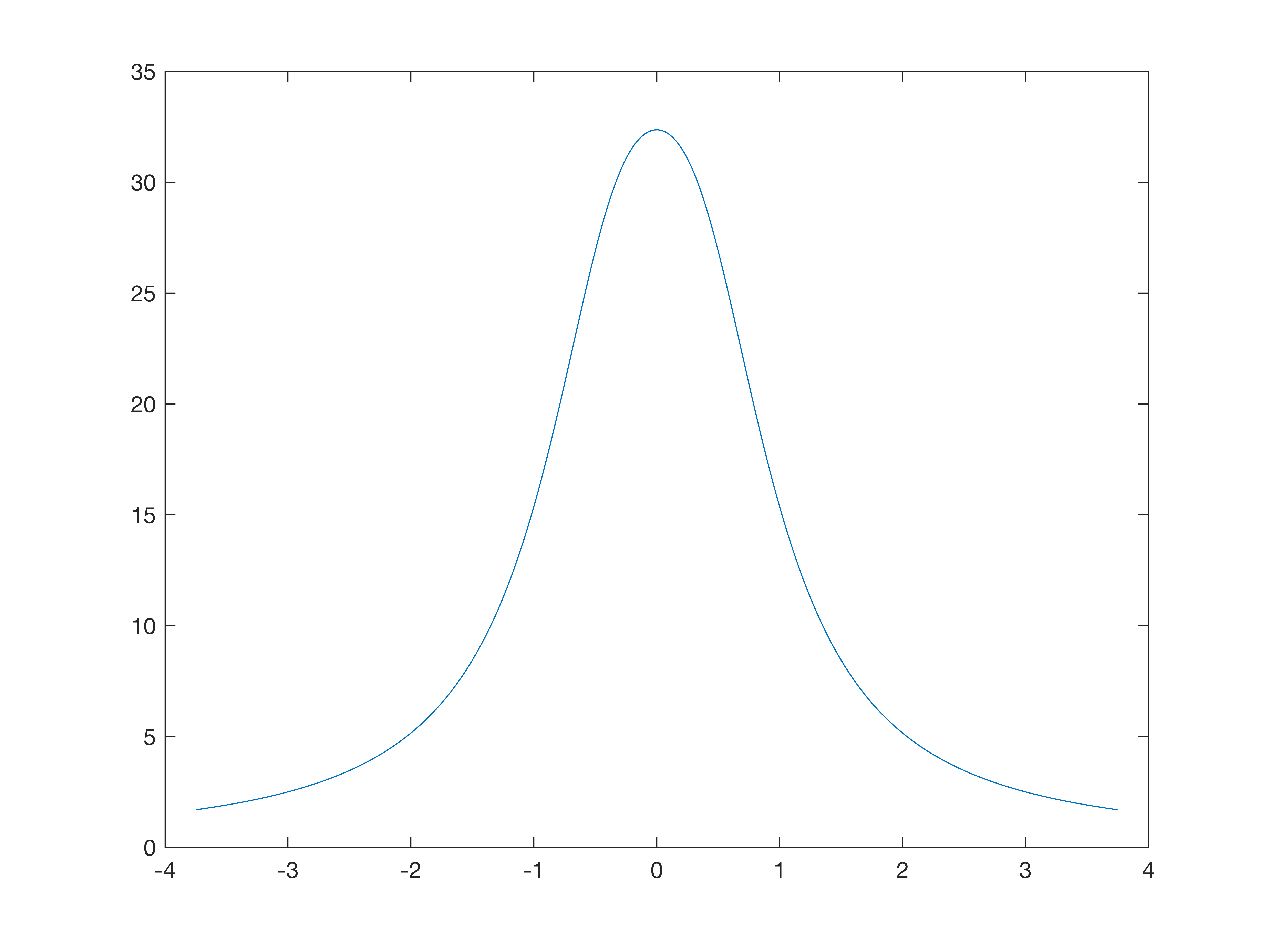}
		\caption{MUSIC image}
	\end{subfigure}
        \begin{subfigure}[b]{0.3\textwidth}
		\centering
		\includegraphics[width=\textwidth]{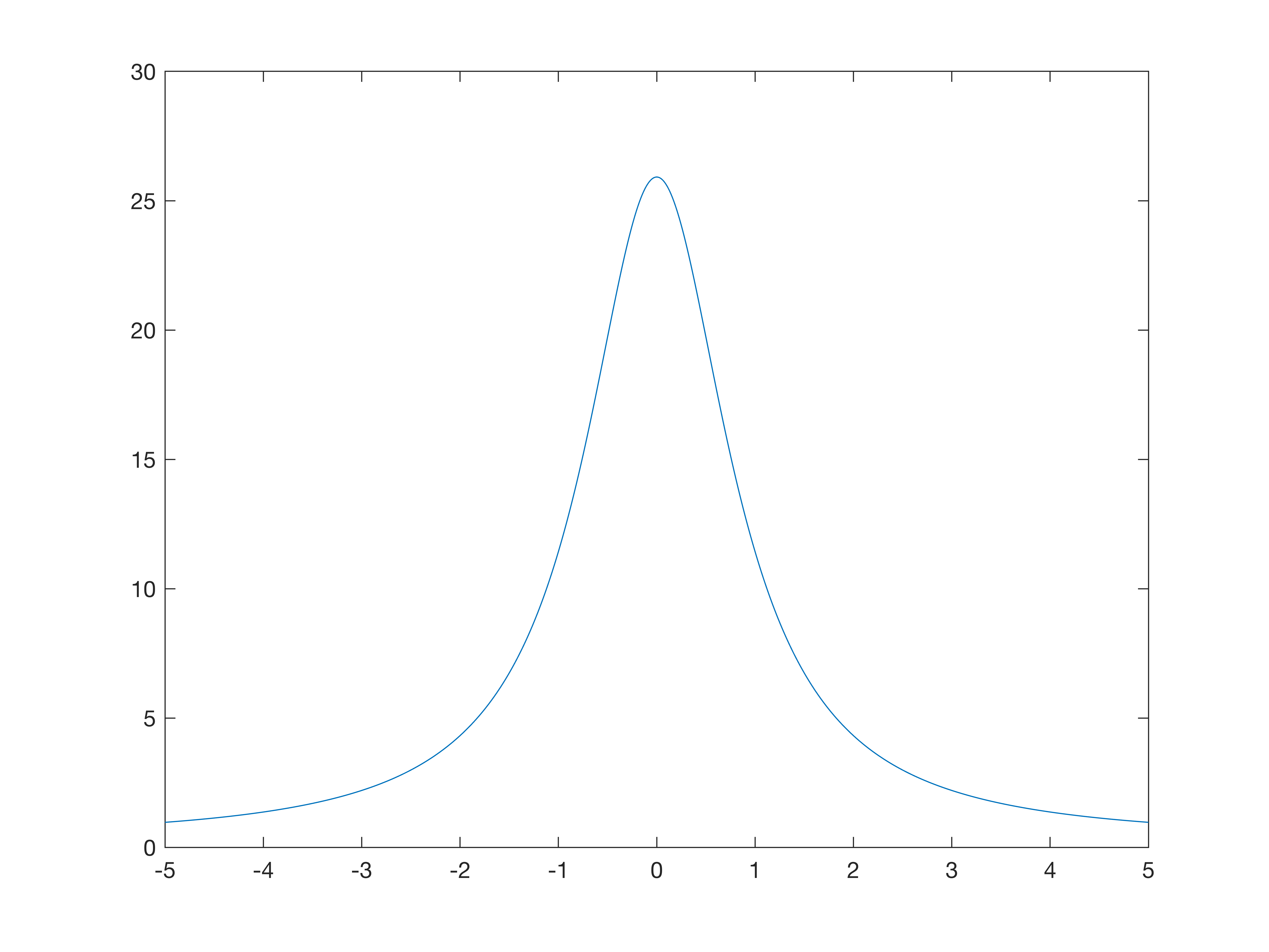}
		\caption{MUSIC image}
	\end{subfigure}
	\begin{subfigure}[b]{0.3\textwidth}
		\centering
		\includegraphics[width=\textwidth]{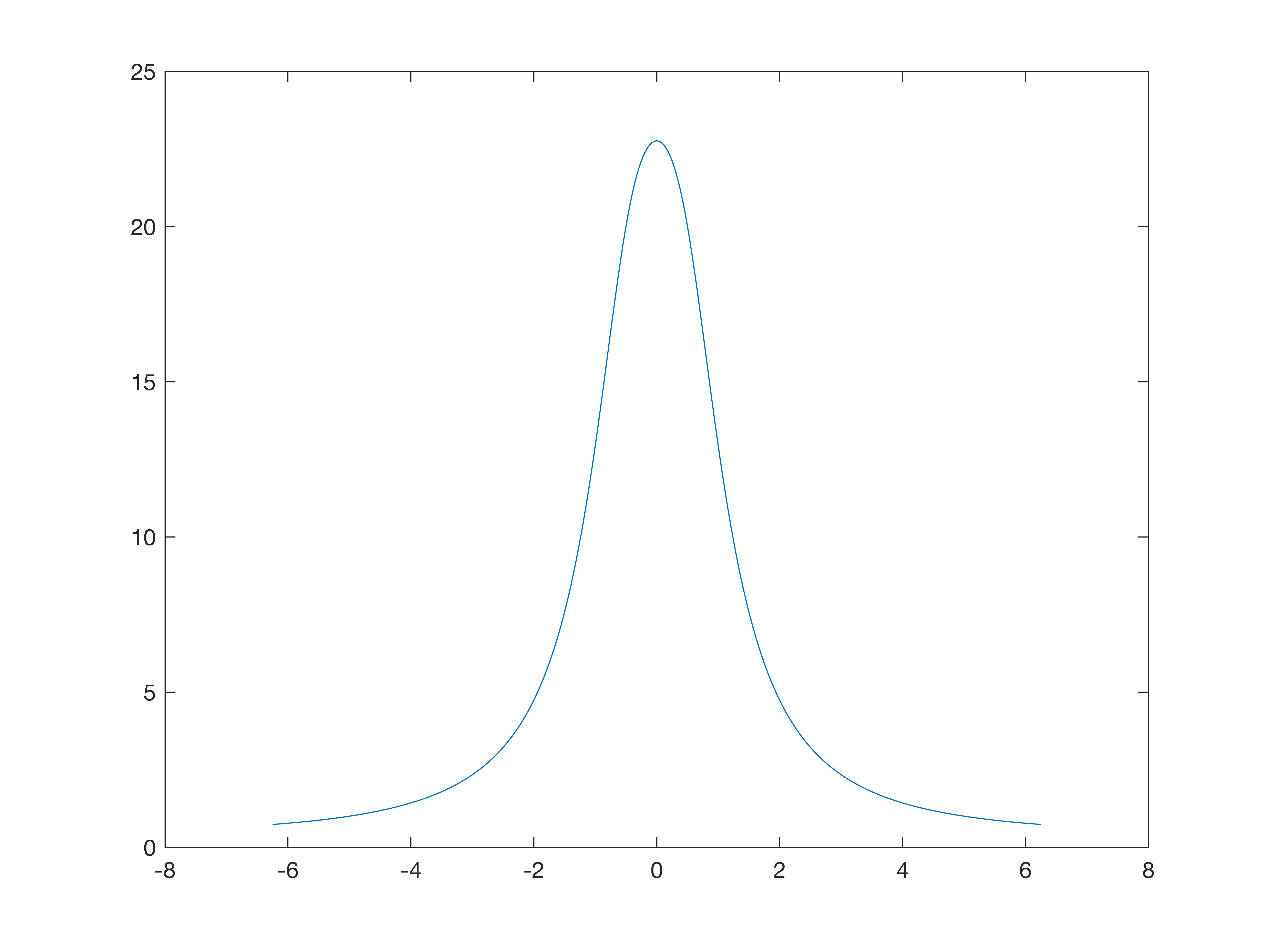}
		\caption{MUSIC image}
	\end{subfigure}
        \begin{subfigure}[b]{0.3\textwidth}
		\centering
		\includegraphics[width=\textwidth]{fig_music_result_1.png}
		\caption{MUSIC image}
	\end{subfigure}
        \caption{Plot (a) is the relation between $\frac{\sigma}{m_{\min}}$ and $d_{\min}$ for several cases. Plots (b)-(f) are MUSIC images of these cases. Note that it is impossible to detect the correct source number from these MUSIC images.}
	\label{fig:numberdetectandmusic}
\end{figure}

\subsection{An optimal algorithm for detecting two sources in multi-dimensional spaces}
For detecting two sources in multi-dimensional spaces, we can first apply \textbf{Algorithm \ref{algo:onedsinguvaluenumberalgo}} to the measurement in several one-dimensional subspaces $V_j$'s and save the outputs, then determine the source number as the maximum value among these outputs. If some of the $V_j$'s are sufficiently close to the space spanned by $\vect y_2-\vect y_1$, it actually achieves similar resolution to the one in Theorem \ref{thm:onednumberdetectalgo1}.  

To be specific, let $\mu=\sum_{j=1}^{2}a_j \delta_{\vect y_j}, \vect y_j\in \mathbb R^k$ and $\vect Y(\vect \omega), \vect \omega\in \mathbb R^k, ||\vect \omega||_2\leqs \Omega$ be the associated measurement in (\ref{equ:highdmodelsetting1}). We choose $N$ unit vectors $\vect v_j$'s in $\mathbb R^k$ and formulate the corresponding Hankel matrices $\vect H_{q}$'s as
\begin{equation}\label{equ:highdhankel1}
\mathbf H_q=\left(\begin{array}{cc}
\mathbf Y(-\Omega\vect v_q)& \mathbf Y(0)\\
\mathbf Y(0)&\mathbf Y(\Omega\vect v_q)
\end{array}\right), \quad q=1, \cdots, N.
\end{equation}
Denoting $\what \sigma_{q,j}$ the $j$-th singular value of $\vect H_{q}$, we can detect the source number by thresholding on $\what \sigma_{q,j}$'s. Moreover, we have the following theorem on the resolution and the threshold. 

\begin{thm}\label{thm:highdnumberdetectalgo1}
Consider $\mu=\sum_{j=1}^{2}a_j \delta_{\vect y_j}, \vect y_j \in B_{\frac{(n-1) \pi}{2\Omega}}^k(\mathbf{0})$ and the measurement $\vect Y$ in (\ref{equ:highdmodelsetting1}) that is generated from $\mu$. If 
\begin{equation}\label{equ:highdnumberalgo-1}
\min_{q=1,\cdots,N}\min\left(\left|\angle(\vect y_1-\vect y_2, \vect v_q)\right|,\ \pi - \left|\angle(\vect y_1-\vect y_2, \vect v_q)\right| \right)=\theta_{\min}
\end{equation}
with $\angle(\cdot, \cdot)$ denoting the angle between vectors, and the following separation condition is satisfied
	\begin{equation}\label{equ:highdnumberalgo0}
	\bnorm{\vect y_1-\vect y_2}_2\geqs \frac{4\arcsin\left(\left(\frac{\sigma}{m_{\min}}\right)^{\frac{1}{2}}\right)}{\Omega\cos{\theta_{\min}}},
	\end{equation}
then we have
	\begin{equation}\label{equ:highdnumberalgo1}
	\max_{q=1, \cdots, N} \what\sigma_{q, 2}> 2\sigma
 \end{equation}
for $\what \sigma_{q,2}$ being the minimum singular value of the Hankel matrix $\vect H_q$ that defined in (\ref{equ:highdhankel1}). On the other hand, if there exists $\what \mu$ consisting of only one source being the a $\sigma$-admissible measure of \ $\vect Y$,  then
 \begin{equation}\label{equ:highdnumberalgo2}
 \what \sigma_{q,2}<2\sigma, \quad q=1, \cdots, N.
 \end{equation}
\end{thm}
\begin{proof}
By (\ref{equ:highdnumberalgo-1}), there exists $\vect v_{q^*}$ such that 
\[
\babs{\vect y_1\cdot \vect v_{q^*}- \vect y_2\cdot \vect v_{q^*}}= \cos\theta_{\min}\bnorm{\vect y_1-\vect y_2}_2\geqs \frac{4\arcsin\left(\left(\frac{\sigma}{m_{\min}}\right)^{\frac{1}{2}}\right)}{\Omega}.
\]
Hence, similar to the proof of Theorem \ref{thm:onednumberdetectalgo1}, we can show that $\what \sigma_{q^*, 2}> 2\sigma$. This proves (\ref{equ:highdnumberalgo0}). Also, we can show (\ref{equ:highdnumberalgo2}) in the same way as the one in the proof of Theorem \ref{thm:onednumberdetectalgo1}. 
\end{proof}

We summarize the algorithm as the following \textbf{Algorithm \ref{algo:highdsinguvaluenumberalgo}}. 

\begin{algorithm}[H]
	\caption{\textbf{Multi-dimensional singular-value-thresholding number detection algorithm}}	
	\textbf{Input:} Noise level $\sigma$, measurement: $\mathbf{Y}(\vect \omega), \vect \omega \in \mathbb R^k, ||\vect \omega||_2\leqs \Omega$;\\
	\textbf{Input:} $N$ unit vectors $\vect v_q$'s;\\
	\For{$q=1, \cdots, N$}{
            Formulate the Hankel matrix:
            \[
            \mathbf H_q=\left(\begin{array}{cc}
            \mathbf Y(-\Omega\vect v_q)& \mathbf Y(0)\\
            \mathbf Y(0)&\mathbf Y(\Omega\vect v_q)
            \end{array}\right).
            \]
        Compute the singular value of $\mathbf H$ as $\what \sigma_{1}, \what \sigma_{2}$ distributed in a decreasing manner;\\
      \If{$\what \sigma_2\geqs 2\sigma$} 
       {\textbf{Return} $n=2$.
	}}
	{
	\textbf{Return} $n=1$.
	}
         \label{algo:highdsinguvaluenumberalgo}
\end{algorithm}

\noindent \textbf{Numerical experiments:}\\
We consider detecting two sources in two-dimensional spaces. For large enough $N$, we consider
\begin{equation}\label{equ:number2dvectors1}
\vect v_q=\left(\cos\left(\frac{q\pi}{N}\right),\ \sin\left(\frac{q\pi}{N}\right)\right)^{\top}\in \mathbb R^2, \quad  q=1,\cdots,N.
\end{equation}
Input $\vect v_q$'s to \textbf{Algorithm \ref{algo:highdsinguvaluenumberalgo}}, we then determine the source number by \textbf{Algorithm \ref{algo:highdsinguvaluenumberalgo}} from measurements $\vect Y(\vect \omega)$. By Theorem \ref{thm:highdnumberdetectalgo1}, we can determine the correct number when 
\[
\bnorm{\vect y_1-\vect y_2}_2\geqs \frac{4\arcsin\left(\left(\frac{\sigma}{m_{\min}}\right)^{\frac{1}{2}}\right)}{\Omega\cos\left({\frac{\pi}{2N}}\right)}. 
\]
This indicates that we already have an excellent resolution by leveraging only a few $\vect v_q$'s. We use $N=10$ unit vectors in the experiments and conduct $100000$ random experiments for both the general and worst cases. As shown in Figure \ref{fig:twodnumberphasetransition} (a) and (c), our algorithm successfully detects the source number when $d_{\min}$ is above nearly the diffraction limit and fails to detect the source number on some cases when $d_{\min}$ is below the computational resolution limit $\mathcal D_{num}(k,2)$. A very interesting phenomenon is that, as shown in Figure \ref{fig:twodnumberphasetransition} (b), there are many cases in which our algorithm detects the correct source number even when $d_{\min}$ is much lower than the diffraction limit. This indicates that the tolerance of the noise of the algorithm is in fact excellent. The reason is that the worst cases or nearly worst cases actually only happen when the noise satisfies certain patterns. Because we use the measurements in $N$ one-dimensional subspaces, it becomes more difficult for the noises in all the subspaces to satisfy these patterns. Thus the noise tolerance becomes better in the two-dimensional case. 

Note that our theoretical results and algorithms are potentially of great importance in practical applications. We will examine the super-resolving ability of our algorithm in practical examples in future work.

\begin{figure}[!h]
	\centering
	\begin{subfigure}[b]{0.3\textwidth}
		\centering
		\includegraphics[width=\textwidth]{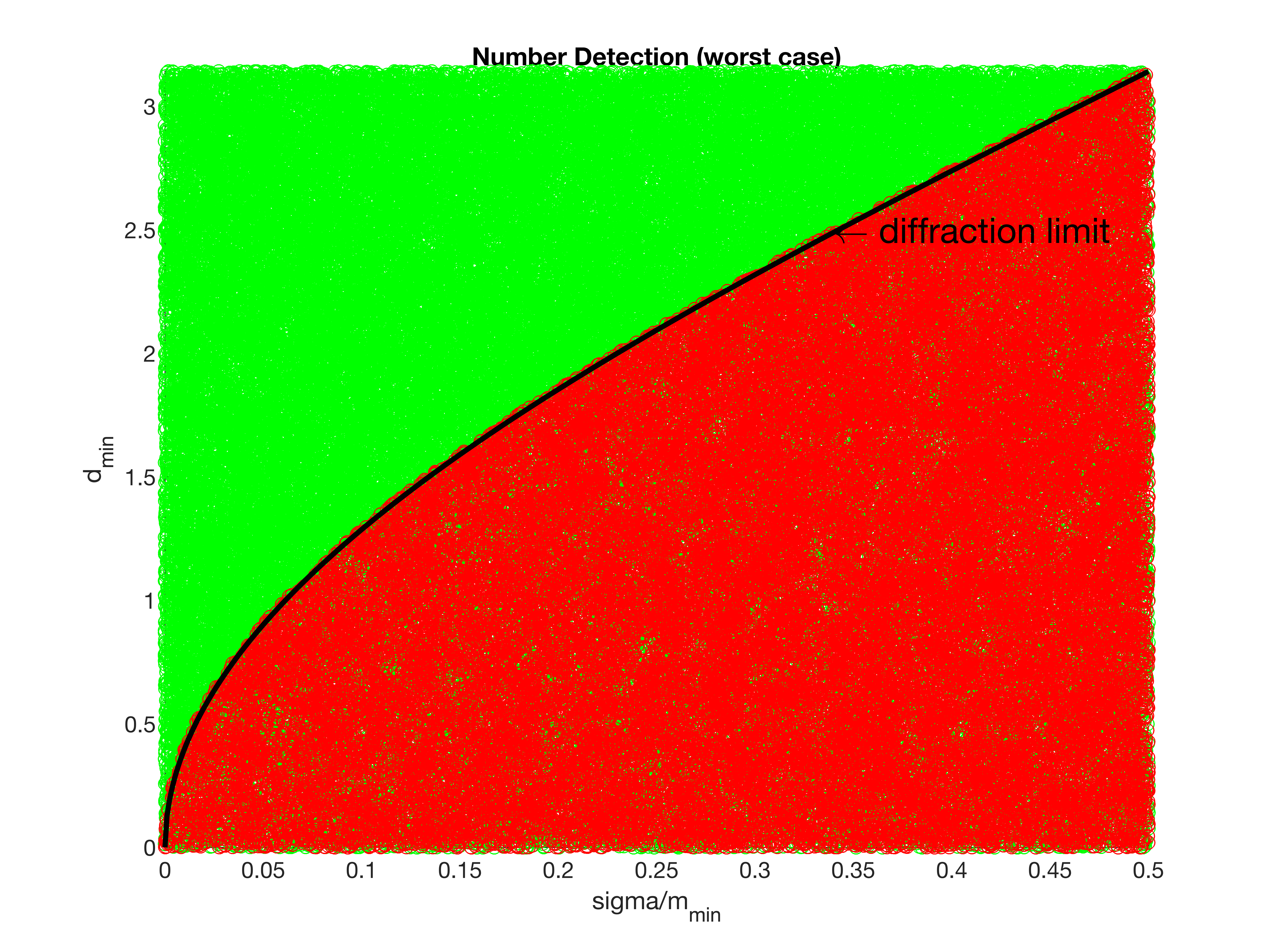}
		\caption{detection results}
	\end{subfigure}
	\begin{subfigure}[b]{0.3\textwidth}
		\centering
		\includegraphics[width=\textwidth]{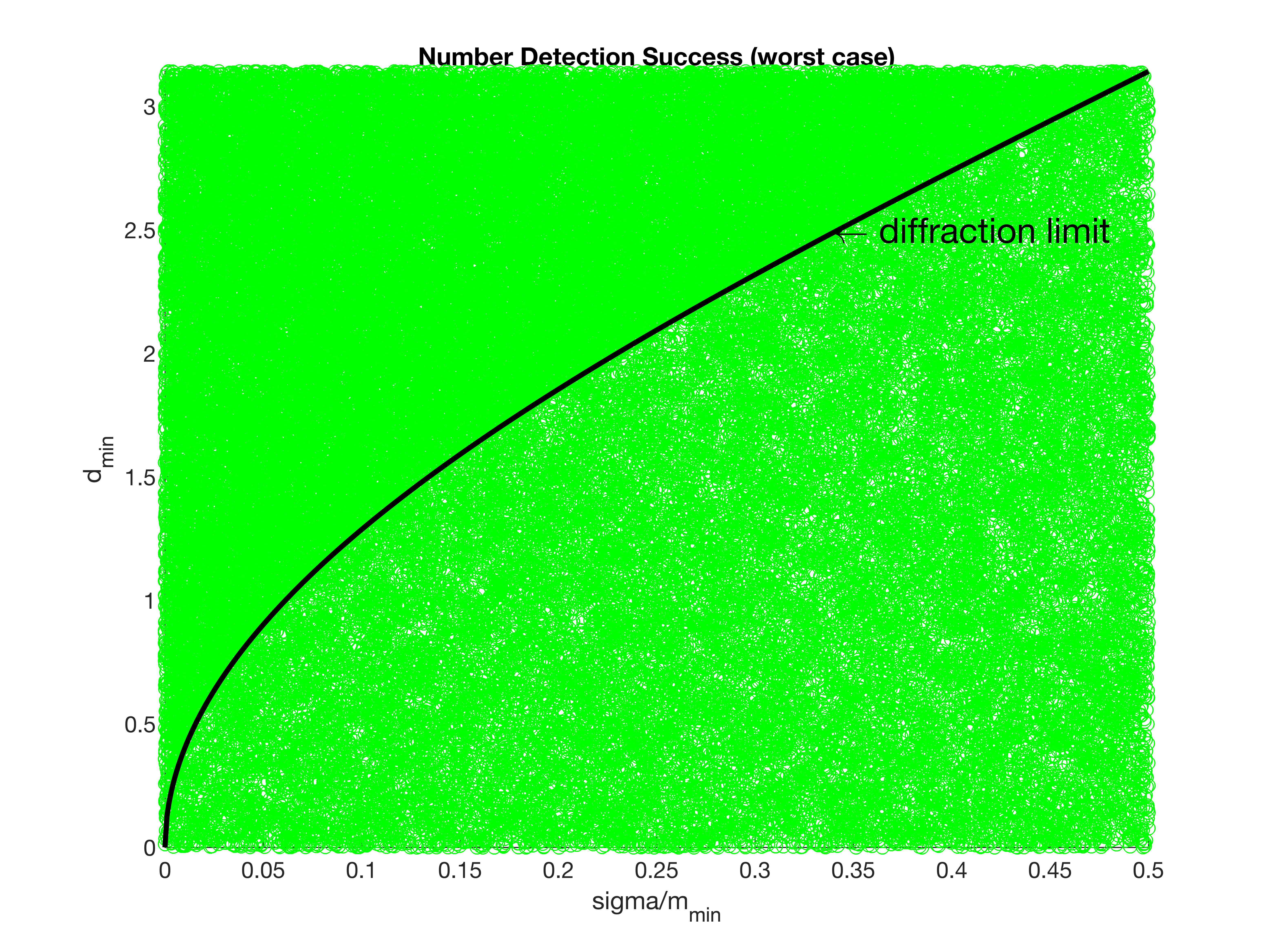}
		\caption{detection success}
	\end{subfigure}
        \begin{subfigure}[b]{0.3\textwidth}
		\centering
		\includegraphics[width=\textwidth]{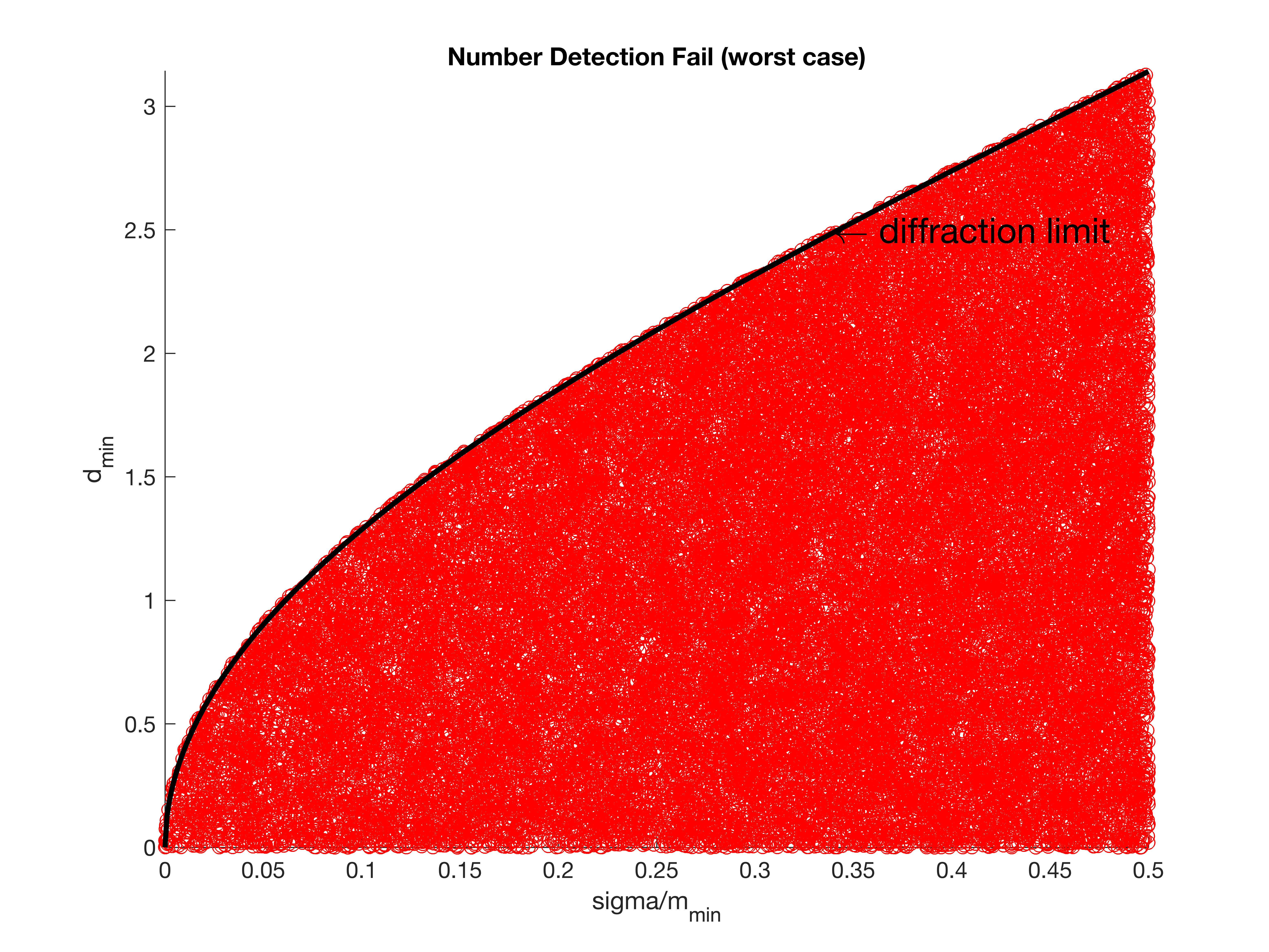}
		\caption{detection fail}
	\end{subfigure}
 \caption{
 Plots of the successful and the unsuccessful number detections by \textbf{Algorithm \ref{algo:highdsinguvaluenumberalgo}} depending on the relation between $\frac{\sigma}{m_{\min}}$ and $d_{\min}$. The green points and red points represent respectively the cases of successful detection and failed detection. The black line is the diffraction limit derived in Theorem \ref{thm:twopointresolution0}.
 }
 \label{fig:twodnumberphasetransition}
\end{figure}

\appendix

\section{Some inequalities}
In this Appendix, we present some inequalities that are used in this paper. We first recall the following Stirling approximation of factorial
\begin{equation}\label{stirlingformula}
\sqrt{2\pi} n^{n+\frac{1}{2}}e^{-n}\leqs n! \leqs e n^{n+\frac{1}{2}}e^{-n},
\end{equation}
which will be used frequently in the subsequent derivations.
\begin{lem}\label{lem:upperboundnumbercalculate1}
	Let $\zeta(n)$ and $\xi(n-1)$ be defined as in (\ref{zetaxiformula1}).		For $n\geqs 2$, we have 
	\[
	\Big(\frac{2}{\zeta(n)\xi(n-1)}\Big)^{\frac{1}{2n-2}}\leqs \frac{ 2e}{n-1}. 
	\]
\end{lem}
\begin{proof} For $n=2,3,4$, it is easy to check that the above inequality holds. Using (\ref{stirlingformula}), we have for odd $n\geqs 5$,
\begin{align*}
&\zeta(n)\xi(n-1)= (\frac{n-1}{2}!)^2\frac{(\frac{n-3}{2}!)^2}{4}\geqs \pi^2(\frac{n-1}{2})^n(\frac{n-3}{2})^{n-2}e^{-(2n-4)}\\
= & (n-1)^{n-2}\frac{\pi^2(\frac{n-1}{2})^n(\frac{n-3}{2})^{n-2}e^{-(2n-4)}}{(n-1)^{n-2}}\\
=&  \pi^2e^2\left(\frac{n-1}{2e}\right)^{2n-2} \frac{(n-3)^{n-2}}{(n-2)^{n-2}}\\
\geqs& 0.29 \pi^2e^2\left(\frac{n-1}{2e}\right)^{2n-2} \quad \Big(\text{since $n\geqs 5$}\Big)
\end{align*}
and for even $n\geqs 6$,
\begin{align*}
&\zeta(n)\xi(n-1)= (\frac{n}{2})!(\frac{n-2}{2})!\frac{(\frac{n-2}{2})!(\frac{n-4}{2})!}{4}\\
\geqs& \pi^2(\frac{n}{2})^{\frac{n+1}{2}}(\frac{n-2}{2})^{n-1}(\frac{n-4}{2})^{\frac{n-3}{2}}e^{-(2n-4)}\\
= & (n-1)^{2n-2}\frac{\pi^2(\frac{n}{2})^{\frac{n+1}{2}}(\frac{n-2}{2})^{n-1}(\frac{n-4}{2})^{\frac{n-3}{2}}e^{-(2n-4)}}{(n-1)^{2n-2}}\\
=&\pi^2 e^2 \left(\frac{n-1}{2e}\right)^{2n-2}\frac{n^{\frac{n+1}{2}}(n-2)^{n-1}(n-4)^{\frac{n-3}{2}}}{(n-1)^{2n-2}}\\
>&\pi^2 \left(\frac{n-1}{2e}\right)^{2n-2}.   
\end{align*}
Therefore, for all $n\geqs 2$, 
\[\Big(\frac{2}{\zeta(n)\xi(n-1)}\Big)^{\frac{1}{2n-2}}\leqs \frac{2e}{n-1} \Big(\frac{2}{\pi^2}\Big)^{\frac{1}{2n-2}}\leqs  \frac{2e}{n-1}.\]
\end{proof}

\begin{lem}\label{lem:upperboundsupportcalculate1}
	Let $\zeta(n)$ and $\lambda(n)$ be defined as in (\ref{zetaxiformula1}) and (\ref{equ:lambda1}), respectively. For $n\geqs 2$, we have 
	\[
	\Big(\frac{8}{\zeta(n)\lambda(n)}\Big)^{\frac{1}{2n-1}}\leqs \frac{2.36e}{n-\frac{1}{2}}.
	\]
\end{lem}
\begin{proof} For $n=2,3,4,5$, the inequality follows from direct calculation. By the Stirling approximation (\ref{stirlingformula}), we have for even $n\geqs 6$,
\begin{align*}
&\zeta(n)\lambda(n)=\zeta(n)\xi(n-2)=(\frac{n}{2})!(\frac{n-2}{2})!\frac{(\frac{n-4}{2}!)^2}{4}\\
\geqs&\pi^2(\frac{n}{2})^{\frac{n+1}{2}}(\frac{n-2}{2})^{\frac{n-1}{2}}(\frac{n-4}{2})^{n-3}e^{-(2n-5)}\\
= & (n-\frac{1}{2})^{2n-1}\frac{\pi^2(\frac{n}{2})^{\frac{n+1}{2}}(\frac{n-2}{2})^{\frac{n-1}{2}}(\frac{n-4}{2})^{n-3}e^{-(2n-5)}}{(n-\frac{1}{2})^{2n-1}}\\
=&(\frac{n-\frac{1}{2}}{2e})^{2n-1}\frac{\pi^2e^42^2}{(n-\frac{1}{2})^2}  \frac{n^{\frac{n+1}{2}}(n-2)^{\frac{n-1}{2}}(n-4)^{n-3}}{(n-\frac{1}{2})^{2n-3}}\\
\geqs& (\frac{n-\frac{1}{2}}{2e})^{2n-1}\frac{4\pi^2}{(n-\frac{1}{2})^2},
\end{align*}
and for odd $n\geqs 7$, 
\begin{align*}
&\zeta(n)\lambda(n)=\zeta(n)\xi(n-2)=(\frac{n-1}{2}!)^2\frac{(\frac{n-3}{2})!(\frac{n-5}{2})!}{4}\\
\geqs& \pi^2(\frac{n-1}{2})^{n}(\frac{n-3}{2})^{\frac{n-2}{2}}(\frac{n-5}{2})^{\frac{n-4}{2}}e^{-(2n-5)}\\
= & (n-\frac{1}{2})^{2n-1}\frac{\pi^2(\frac{n-1}{2})^{n}(\frac{n-3}{2})^{\frac{n-2}{2}}(\frac{n-5}{2})^{\frac{n-4}{2}}e^{-(2n-5)}}{(n-\frac{1}{2})^{2n-1}}\\
= &(\frac{n-\frac{1}{2}}{2e})^{2n-1}\frac{\pi^2e^42^2}{(n-\frac{1}{2})^2}  \frac{(n-1)^{n}(n-3)^{\frac{n-2}{2}}(n-5)^{\frac{n-4}{2}}}{(n-\frac{1}{2})^{2n-3}}\\
\geqs& (\frac{n-\frac{1}{2}}{2e})^{2n-1}\frac{4\pi^2}{(n-\frac{1}{2})^2}. 
\end{align*}
Therefore, for all $n\geqs 2$,
\[
\Big(\frac{8}{\zeta(n)\lambda(n)}\Big)^{\frac{1}{2n-1}}\leqs \frac{2e}{n-\frac{1}{2}}\Big(\frac{(n-\frac{1}{2})^28}{4\pi^2}\Big)^{\frac{1}{2n-1}}\leqs \frac{2.36e}{n-\frac{1}{2}}.
\]
\end{proof}

\begin{lem}\label{lem:upperboundsupportcalculate2}
	Let $\zeta(n)$ be defined as in (\ref{zetaxiformula1}).	For $n\geqs 2$, we have 
	\[
 \frac{(n-\frac{1}{2})^{2n-1}}{\zeta(n)(n-2)!}\leqs 2^{n-\frac{3}{2}}e^{2n}\sqrt{n-0.5}(4.5\pi)^{-1}. 
 \]
\end{lem}
Proof: By the Stirling approximation formula (\ref{stirlingformula}), when $n$ is odd and $n\geqs 3$, we have 
\begin{align*}
&\frac{(n-\frac{1}{2})^{2n-1}}{\zeta(n)(n-2)!}=\frac{(n-\frac{1}{2})^{2n-1}}{(\frac{n-1}{2}!)^2(n-2)!}\\
\leqs& \frac{(n-\frac{1}{2})^{2n-1}}{(\sqrt{2\pi})^3(\frac{n-1}{2})^{n}(n-2)^{n-2+\frac{1}{2}}e^{-(2n-3)}}\\
\leqs & \frac{2^{n}e^{2n}\sqrt{n-\frac{1}{2}}}{(e\sqrt{2\pi})^3}\frac{(n-\frac{1}{2})^{2n-\frac{3}{2}}}{(n-1)^n (n-2)^{n-\frac{3}{2}}}\leqs\frac{2^{n-\frac{3}{2}}e^{2n-0.95}\sqrt{n-\frac{1}{2}}}{(\sqrt{\pi})^3}\\
<& 2^{n-\frac{3}{2}}e^{2n}\sqrt{n-\frac{1}{2}}(4.5\pi)^{-1}
\end{align*}
When $n$ is even and $n\geqs 4$, we have 
\begin{align*}
&\frac{(n-\frac{1}{2})^{2n-{1}}}{\zeta(n)(n-2)!}=\frac{(n-\frac{1}{2})^{2n-{1}}}{(\frac{n}{2})!(\frac{n-2}{2})!(n-2)!}\\
\leqs& \frac{(n-\frac{1}{2})^{2n-{1}}}{(\sqrt{2\pi})^3(\frac{n}{2})^{\frac{n+1}{2}}(\frac{n-2}{2})^{\frac{n-1}{2}}(n-2)^{n-2+\frac{1}{2}}e^{-(2n-3)}}\\
\leqs & \frac{2^{n}e^{2n}\sqrt{n-\frac{1}{2}}}{(e\sqrt{2\pi})^3}\frac{(n-\frac{1}{2})^{2n-{1.5}}}{n^{\frac{n+1}{2}}(n-2)^{\frac{n-1}{2}} (n-2)^{n-{1.5}}}\leqs \frac{2^{n-\frac{3}{2}}e^{2n-0.95}\sqrt{n-\frac{1}{2}}}{(\sqrt{\pi})^3}\\
<& 2^{n-\frac{3}{2}}e^{2n}\sqrt{n-\frac{1}{2}}(4.5\pi)^{-1}
\end{align*}
For $n=2$, the inequality follows from a direct calculation. 
		
\bibliographystyle{plain}
\bibliography{reference_final}
\end{document}